\documentclass[a4paper,11pt]{article}
\usepackage{jheppub} % for details on the use of the package, please see the JINST-author-manual
\usepackage{lineno}
\usepackage{simpler-wick}
\usepackage{latexsym}
\usepackage[utf8]{inputenc}
\usepackage[english]{babel}
\usepackage{graphicx}
\usepackage{braket}
\usepackage{subcaption}
\usepackage{verbatim}
\usepackage{multirow}
\usepackage{amsmath, amssymb, amsthm, amsfonts}
\usepackage{bm}
\usepackage{mathrsfs}
\usepackage{soul}
\usepackage{hyperref}
\usepackage{natbib}
\usepackage{float}
\usepackage{mathtools}
\usepackage{pgfplots}
\usepackage{siunitx}
\usepackage{slashed}
\usepackage{amsthm}
\usepackage{amssymb}
\usepackage{physics}	% installed packages for typing maths and physics
\usepackage{graphicx} % need for figures
\usepackage{epstopdf}
\usepackage{bbold}
\usepackage{wasysym}
\usepackage[ampersand]{easylist}
\newtheorem{theorem}{Theorem}
\newtheorem{lemma}[theorem]{Lemma} % Lemma numbering follows Theorems

\theoremstyle{definition}
\newtheorem{definition}{Definition}

\theoremstyle{remark}
\newtheorem*{remark}{Remark}

\usepackage{stmaryrd}
\usepackage{simpler-wick}

\def\nn{\nonumber}

\def\be{\begin{equation}}
\def\ee{\end{equation}}

\global\long\def\mA{\mathcal{A}}%
 
\global\long\def\mB{\mathcal{B}}%

\global\long\def\mF{\mathcal{F}}%

\global\long\def\mH{\mathcal{H}}%

\global\long\def\mJ{\mathcal{J}}%
 
\global\long\def\mK{\mathcal{K}}%

\global\long\def\mM{\mathcal{M}}%
 
\global\long\def\mN{\mathcal{N}}%
 
\global\long\def\mO{\mathcal{O}}%

\global\long\def\mU{\mathcal{U}}%

\global\long\def\mW{\mathcal{W}}%
 
\global\long\def\mX{\mathcal{X}}%

\global\long\def\e{\epsilon}%
 
\global\long\def\ra{\rightarrow}%

\global\long\def\avg#1{\left\langle #1\right\rangle }%

\global\long\def\f#1#2{\frac{#1}{#2}}%
 
\global\long\def\del{\partial}%
 
\global\long\def\t{\theta}%
 
\global\long\def\a{\alpha}%
 
\global\long\def\b{\beta}%
 
\global\long\def\g{\gamma}%
 
\global\long\def\G{\Gamma}%
 
\global\long\def\s{\sigma}%
 
\global\long\def\r{\rho}%
 
\global\long\def\d{\delta}%
 
\global\long\def\Tr{\text{Tr}}%
 
\global\long\def\tr{\text{tr}}%
 
\global\long\def\bra#1{\left\langle #1\right|}%
 
\global\long\def\ket#1{\left|#1\right\rangle }%
 
\global\long\def\N{\mathbb{N}}%
\global\long\def\I{\mathbb{I}}%

\global\long\def\R{\mathbb{R}}%

\global\long\def\p{\varphi}%

\global\long\def\w{\omega}%
 
\global\long\def\D{\Delta}%

\global\long\def\l{\ell}%

\global\long\def\app{\approx}%
\global\long\def\arccosh{\text{arccosh}}%
\global\long\def\lam{\lambda}%
\global\long\def\i{\text{i}}%

\global\long\def\td{\text{Td}}%
 
\global\long\def\bri#1{\left\llbracket #1\right|}%
 
\global\long\def\kit#1{\left|#1\right\rrbracket }%
\global\long\def\avvg#1{\left\llbracket #1\right\rrbracket }%
\global\long\def\Li{\text{Li}}%

\global\long\def\chooseq#1#2{\begin{bmatrix}#1\\
#2
\end{bmatrix}_{q}}%

\title{Single-Sided Black Holes in Double-Scaled SYK Model and No Man's Island}

\author[1]{Xuchen Cao}
\affiliation[1]{Department of Physics, University of Illinois Urbana Champaign}
\author[2]{and Ping Gao}
\affiliation[2]{Kavli Institute for Theoretical Sciences (KITS),\\
University of Chinese Academy of Sciences, Beijing 100190, China}

\emailAdd{xuchenc2@illinois.edu, gaoping@ucas.ac.cn}

\abstract{
We study a single-sided black hole with an end-of-the-world (EoW) brane behind the horizon in the double-scaled SYK (DSSYK). The new Hamiltonian is a deformation of the original DSSYK Hamiltonian with an extra exponential wormhole length operator, which leads to a new chord diagram rule. The boundary algebra is defined as generated by the new Hamiltonian and boundary matter. There is an alternative but equivalent definition with a $q$-coherent state due to a nontrivial isomorphism of the vN algebra of DSSYK. This isomorphism induces a unitary equivalence, which yields a surprising result that the boundary algebra of a single-sided black hole in DSSYK has a non-trivial commutant and is a type II$_1$ vN factor. It follows that the full bulk reconstruction from the boundary is impossible, and there is a ``no man's island" behind the horizon in the semiclassical JT limit. Inspired by the EoW brane, we construct a family of matter-brane states with an arbitrary number of matter chords and behaving like an EoW brane. They exactly solve the full spectrum of DSSYK. 

We take different ways to understand the nontrivial commutant. We argue that the commutant is complex on chord number basis and thus non-geometric. In the semiclassical JT limit, the commutant becomes the canonical purification of the boundary algebra and claims the no man's island. In the context of Hawking radiation after Page time, the unitary equivalence is interpreted as encoding the canonical purification into the old Hawking radiation, and the no man's island has the same essence as the island. Including the exponential wormhole length operator independently, the boundary algebra is extended to all bounded operators and reconstructs the no man's island. This can be regarded as a different choice for the definition of boundary algebra. This type I$_\infty$ algebra is closely related to the EoW brane in Kourkoulou-Maldacena.
}

\begin{document}
\maketitle
\flushbottom

\section{Introduction}
\label{sec:intro}
Sachdev-Ye-Kitaev (SYK) model, which is a strongly interacting model of Majorana fermions originally proposed in the context of spin glass, has been extensively studied in recent years as a toy model for quantum gravity~\cite{Sachdev:1992fk,Sachdev:2010um,kitaev2015simple,Almheiri:2014cka,Polchinski:2016xgd,Jensen:2016pah,Maldacena:2016upp,Maldacena:2016hyu} due to its tractability in the large $N$ limit and resemblance to dynamics of gravity on $\text{AdS}_2$ spacetime. There are three relevant parameters in the SYK model: $N$ the total number of fermions in the system, $p$ the number of fermions coupled by the Hamiltonian, and $\beta \mJ$ the dimensionless coupling constant. As we always take the large $N$ limit the first two can be packed into a single dimensionless parameter $\lambda=\frac{2p^2}{N}$. %The physics of the model depend on different regimes of these parameters, including the large $p$ limit where $\lambda\ra 0$ and $\lam^{1/2}\b\mJ}$ fixed, the quantum JT limit where $\lambda\ra0$ and $\lam^{3/2}\b\mJ$ fixed, the semiclassical JT limit where we take $\lam^{3/2}\b\mJ\ll 1$ in the quantum JT limit. 
In this paper we will mainly study the double scaling limit where $\lam$ is held constant as $N\ra \infty$. It has been shown that in this limit the model can be solved by the so-called chord diagram technique~\cite{Berkooz:2018qkz,Berkooz:2018jqr,Berkooz:2020xne,Berkooz:2020uly,Berkooz:2020fvm,Berkooz:2024lgq}, where quantities such as partition functions and matter correlators can be calculated in a closed form. One remarkable observation is that the sum over chord diagrams can be interpreted as a gravitational path integral over discrete bulk geometries coupled to matter fields, and these geometries reduce to $\text{AdS}_2$ upon taking the triple scaling/JT limit~\cite{Berkooz:2018jqr,Lin:2022rbf,Lin:2023trc,Berkooz:2022mfk}. This observation implies that we can treat this emergent discrete bulk geometry as a UV completion of JT gravity coupled to matter. We will take this point of view throughout this paper and discuss the structure of Hilbert spaces and algebra of observables defined on this geometry. On the other hand, it has been proposed that the double-scaled SYK (DSSYK) model itself is dual to the bulk sine-dilaton gravity theory~\cite{Blommaert:2023opb,Blommaert:2023wad,Blommaert:2024ymv,Blommaert:2025avl,Blommaert:2025eps}, so we expect our discussion here to have implications in studying the algebraic structures of sine-dilaton gravity.

The application of operator algebraic methods in holography has attracted considerable attention in the past few years as it provides a natural language to study the emergence of the bulk from boundary degrees of freedom both in and away from the large $N$ limit~\cite{Leutheusser:2021qhd,Leutheusser:2021frk,Witten:2021unn,Leutheusser:2022bgi,Chandrasekaran:2022eqq,Gesteau:2024rpt,Engelhardt:2023xer,Engelhardt:2025bxy,Liu:2025cml,Kudler-Flam:2025cki}. By subregion-subalgebra duality~\cite{Leutheusser:2022bgi}, subregions in the bulk can be understood as being defined by subalgebras on the boundary. One advantage of this description is that the notion of bulk subregion is valid only in the semiclassical limit $G_N\ra 0$ where $G_N\sim\f{1}{N^2}$ is the bulk Newton constant, while subalgebras are well-defined for arbitrary $N$. This allows for the generalization of the notion of subregion to the regime where quantum fluctuation in the bulk is large and the geometric description no longer holds. In some simple models such as JT gravity with two AdS boundaries, algebras of boundary observables has been explicitly constructed and are proved to be type $\text{II}_\infty$ von Neumann factors~\cite{Penington:2023dql,Kolchmeyer:2023gwa,Penington:2024sum}. Furthermore, it has been shown that in this case the algebras defined on the left and right boundary are commutants of each other, which, together with factoriality, ensure that arbitrary bulk operators can be reconstructed from the union of algebras on the two boundaries. Furthermore, we can take semiclassical limits of these algebras to restore the notion of bulk entanglement wedge~\cite{Gao:2024gky}. The reminiscence of this construction in the DSSYK model is that the emergent bulk geometry also has two boundaries where algebras of boundary observables can be defined~\cite{Lin:2022rbf,Xu:2024hoc}. This bulk Hilbert space is known as q-Fock space and has been studied extensively in mathematical literature~\cite{bozejko1991example,speicher1993generalized,Bozejko1994,bozejko1997q,Ricard2003FactorialityOQ,bozejko2017fock,skalski2018remarks}. In particular, it has been proved that boundary algebras are of type $\text{II}_1$.

While the above discussion of JT gravity focused on the case of nearly $\text{AdS}_2$ boundary condition. It is natural to consider more general boundary conditions~\cite{Goel:2020yxl,Ferrari:2020yon}. In particular, we will focus on the end-of-the-world (EoW) brane boundary condition which is essentially the Neumann boundary condition for both bulk metric and dilaton field. In contrast to AdS boundaries, EoW branes are dynamical objects. These geometries has been conjectured to describe pure states in the SYK model~\cite{Kourkoulou:2017zaj,Almheiri:2018xdw}. Dynamics of nearly $\text{AdS}_2$ geometries in JT gravity with EoW branes and their random matrix dual has been studied in~\cite{Gao:2021uro}. In~\cite{Penington:2019kki}, nearly $\text{AdS}_2$ geometries ending on EoW branes were proposed as a toy model for black hole evaporation and Page transitions. The key feature of the Page transition is a crossover of quantum minimal surfaces (QES) in the bulk, and respectively a sudden change in the entanglement wedge of the AdS boundary. In the algebraic framework~\cite{Engelhardt:2023xer} this transition manifests as a change in the type of the boundary algebra. It is therefore meaningful to understand the boundary algebra of a single-sided nearly $\text{AdS}_2$ black hole with an EoW brane, as they potentially serve as an example of the algebraic Page transition.

\begin{figure}
	\centering
        \includegraphics[height=5cm]{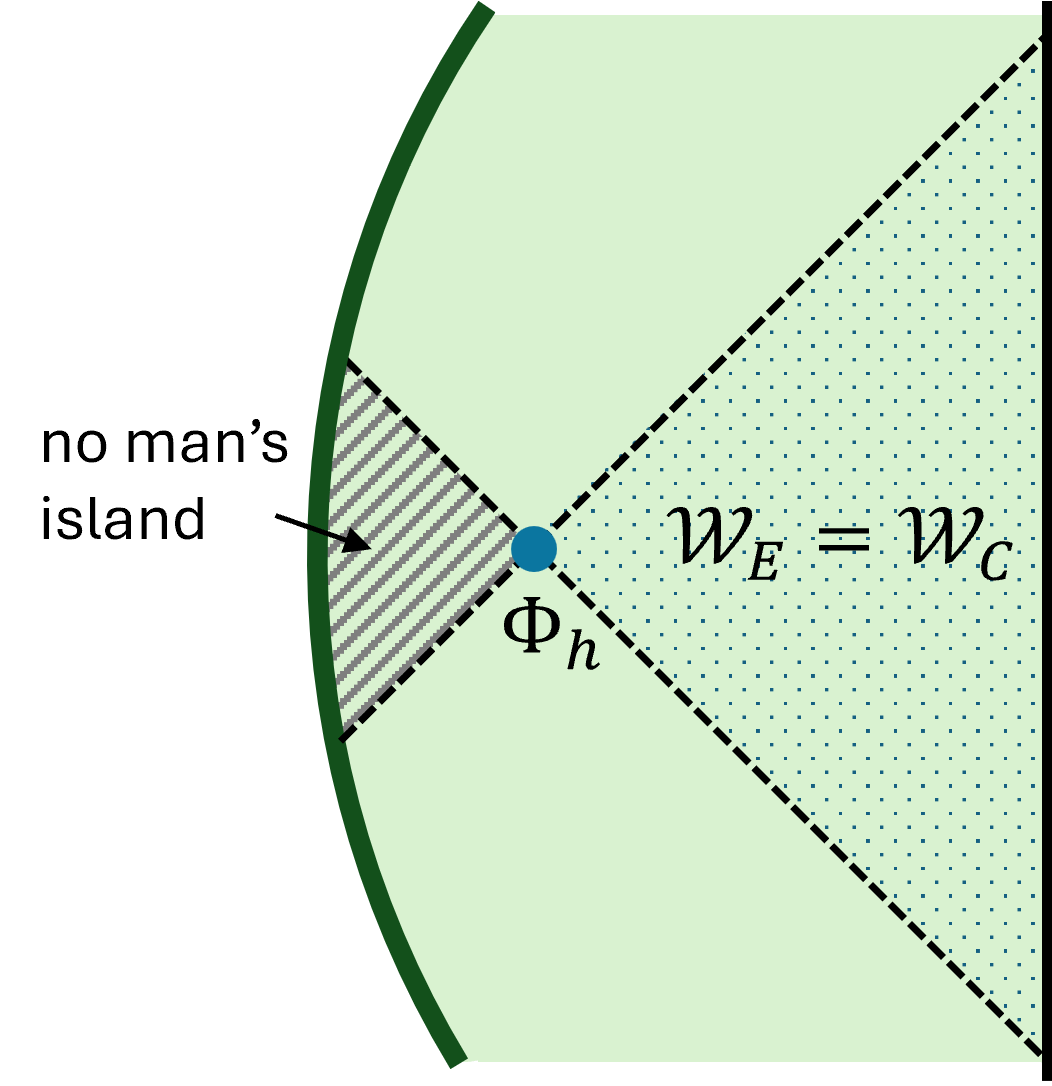}
        \caption{A single-sided black hole with an EoW brane behind the horizon in JT gravity. The dark green curve is the trajectory of the EoW brane; the blue dot is the horizon; the shaded region behind the horizon is the ``no man's island"; the entanglement wedge of the boundary is the same as the causal wedge and is the dotted region outside the horizon. \label{pic-nmisland}}
\end{figure}

In this paper we take a step in this direction by studying the boundary algebra in a different context: a single-sided black hole defined in DSSYK model with an EoW brane, which can be thought of as a discrete version of geometries in~\cite{Gao:2021uro}. EoW branes in DSSYK models were first studied in~\cite{Okuyama:2023byh} by constructing a discrete Hamiltonian which reproduce the dynamics in~\cite{Gao:2021uro} upon taking the triple scaling/JT limit. In this paper we start from re-deriving this Hamiltonian from a new set of diagrammatic rules. Building on this, the boundary algebra can then be defined and analyzed. We construct two different but equivalent ways to characterize the EoW brane: a deformation of Hamiltonian or a $q$-coherent state. In the first perspective, the single-side black hole boundary algebra $\tilde \mA_R$ is defined as generated by this deformed Hamiltonian and a matter operator, analogous to the boundary algebra defined in the ordinary DSSYK. In the second perspective, we have a special state for the EoW brane in the DSSYK Hilbert space $\mH$, where all bounded operators $\mB(\mH)$ is generated by the boundary algebras $\mA_{L,R}$ on each side \cite{Lin:2023trc,Xu:2024hoc}. The equivalence between these two descriptions is warranted by a nontrivial isomorphism between von Neumann algebras $\tilde \mA_R$ and $\mA_R$, which further induces a unitary equivalence $U$. This unitary is a very special map between two different geometric pictures: from the first perspective, the Hamiltonian deformation of EoW brane gives a new one-sided chord diagram rule that would naively look like a single-sided system\footnote{In~\cite{Berkooz:2025ydg} a similar chord diagram rule was derived in a different setup.}; from the second perspective, the two-sided Hamiltonian are the same as DSSYK, and the EoW brane is a special entangled state in a two-sided system. In fact, these two perspectives can also be understood as resulting from different slicing schemes of one-sided chord diagrams. 

It follows that the boundary algebra of a single-sided black hole with an EoW brane has a nontrivial commutant $\tilde \mA_L$ and is a type II$_1$ factor. This is a remarkable and surprising result: we are unable to reconstruct the whole bulk from the single-sided boundary algebra. In the semiclassical JT limit, where we have emergent continuous bulk spacetime, we argue that the region behind the horizon cannot be claimed by the boundary and we name it as ``no man's island" (see Figure \ref{pic-nmisland} as an illustration).

We take a few different ways to understand the commutant $\tilde \mA_L$ of the boundary algebra $\tilde \mA_R$. First, we argue that the matter operator (defined by the Tomita conjugate $\tilde J_\psi$ in \eqref{3.62}) in the commutant are complicated in the chord number basis. Since the chords can be regarded as discrete geometry in the bulk, we interpret this feature as non-geometric. This heuristically explains why the commutant $\tilde \mA_L$ is not obviously associated to an asymptotic boundary. Nevertheless, we show that a minimal extension of the boundary algebra $\tilde \mA_R$ to all bounded operators $\mB(\mH)$ is simply including the exponential wormhole length operator $e^{-\bar \l _b}$ of the single-sided black hole. 

In the semiclassical JT limit, the commutant becomes the canonical purification of the boundary algebra of Lorentzian ``single trace" operators and claims the no man's island. Since the nontrivial unitary $U$ maps to a two-sided system, an appropriate way to understand the no man's island is comparing it with the Hawking radiation process after Page time. Consider the single-sided black hole is simulated by an SYK model of $N$ Majorana fermions, and let it couple with an $N/2$-site spin chain as if the Hawking radiation (see Figure \ref{pic-syk-SC}). They have very different dynamics but have the same dimensional Hilbert space. After Page time, the spin chain claims an island in the dual spacetime even though it is not a holographic system. However, we can use a Jordan-Wigner transformation to represent all spin chain Pauli operators in terms of $N$ Majorana fermions, and then map the spin chain to another SYK model that canonically purifies the black hole as a thermofield double state. Heuristically, we implement a complex unitary transformation to collapse all old Hawking radiation into another black hole and create an Einstein-Rosen bridge with the old radiating single-sided black hole. We interpret the nontrivial unitary equivalence $U$ as the encoding map from the canonical purification to the old Hawking radiation, and thus the no man's island has the same essence as the island.

From the chord diagram viewpoint, it is also natural to define a {\it different} boundary algebra for a single-sided black hole by including the exponential wormhole length operator $e^{-\bar \l _b}$. As we show in Section \ref{sec:4.2}, this algebra is indeed the full algebra $\mB(\mH)$ and claims the no man's island in the semiclassical JT limit. Since this operator can be written as the coupling term of the eternal traversable wormhole, we can interpret it as an averaged version of the measurement-based single-sided operation in \cite{Kourkoulou:2017zaj} to reconstruct the no man's island in a pure state of a single SYK model. 

This is indeed a very interesting ambiguity about which operators should be included in the definition of a boundary algebra in DSSYK. From chord diagram (or bulk) viewpoint, there is no preferred rule to assign some operator to a specific boundary. Therefore, different choices define different subalgebras, which are dual to different bulk subregions. This is a generalization of the subregion-subalgebra duality in \cite{Leutheusser:2022bgi}. On the other hand, if we regard an algebra as defining a theory on the boundary, then different subalgebras indeed correspond to different types of boundary theories for a single-sided black hole. From the brief summary of the above two paragraphs, there are two types of single-sided black holes: 
\begin{itemize}
    \item A single-sided black hole entangled with a large amount of non-gravitational degrees of freedom, e.g. an old radiating black hole after Page time.
    \item A single-sided black hole in a pure state without entanglement, e.g. pure state of SYK model in \cite{Kourkoulou:2017zaj}.
\end{itemize}  
Including $e^{-\bar \l _b}$ or not additional to $\tilde \mA_R$ serves as two examples in a unified framework for these two types of single-sided black holes.

Inspired by the EoW brane state, we construct a family of matter-brane states, which contain an arbitrary number of matter chord and behave the same as the EoW brane when probed by boundary Hamiltonian. We can understand these states as more general EoW branes with various bulk matter boundary conditions. In terms of these matter-brane states, we are able to exactly diagonalize the DSSYK Hamiltonian (and any matter operator) in the full Hilbert space and solve the spectrum. 

This paper is organized as follows, in Section~\ref{sec:DSSYK} we review basics of DSSYK model. Readers familiar with these materials can skip this section. In Section~\ref{sec:EoW brane} we construct EoW branes in DSSYK model from diagrammatics. We define the algebra of boundary observables and prove that it is a type $\text{II}_1$ von Neumann algebra. We also define matter-brane states, which exactly solve the full spectrum of DSSYK in the whole Hilbert space. In Section~\ref{more on boundary algebra} we elaborate more on the boundary algebra by discussing its commutant. We discuss the semiclassical JT limit and prove the existence of a no man's island behind the horizon, with an emphasis on its physical significance. We also discuss the physical interpretation of the alternative definition of boundary algebra by including the exponential wormhole length operator. We argue that they correspond to two types of single-sided black holes with or without a large anount of entanglement. Finally, in Section~\ref{sec:discussion} we include discussions and propose future directions.\\ \\
Here we list some frequently used notations for reference
\begin{itemize}
    \item $a_{L/R,0},a^\dagger_{L/R,0}$: annihilation and creation operators of gravitational chords ($H$ chords) on the left and right boundaries. 
    \item $H_{L/R,0}(\equiv H_{L/R})$: Hamiltonian operators on the left and right boundaries
    \item $H_{L/R,1}(\equiv M_{L/R})$: Matter operators on the left and right boundaries
    \item $\tilde H_{R,0/1}$: Hamiltonian and matter operators in the presence of EoW branes, which we always assume to act on the right boundary
    \item $\tilde H_{bulk}$: EoW brane Hamiltonian in the pure gravity case 
    \item $\mA_{L,R}$: von Neumann algebras generated by $H_{L/R,0}$ and $H_{L/R,1}$.
    \item $\tilde\mA_{R}$: von Neumann algebra generated by $\tilde H_{R,0}$ and $\tilde H_{R,1}$.
    \item $\mathcal H$: the full Hilbert space.
    \item $B(\mathcal H)$: the algebra of all bounded operators on $\mathcal H$.
    \item $\ket{\w}$: the vacuum state (the state with no chords) in $\mathcal H$.
    \item $P_{\psi}$: the projection operator onto state $\psi$ in $\mathcal H$.
    \item $\ket{w}$: an arbitrary chord state given by a fixed sequence of $H$ and $M$ chords, such as $\ket{00111011}$ where $0$ and $1$ denote $H$ and $M$ chords respectively
    \item $[n]_q$: the shorthand notation defined by $[n]_q=\frac{1-q^n}{1-q}$
    \item $(a;q)_n$: the $q$-Pochhammer symbol defined by $(a,q)_n=\prod_{k=0}^{n-1}(1-aq^k)$, we also have the shorthand notation that $(a,b;q)_n=(a;q)_n(b;q)_n$ and $(ae^{\pm i\theta},q)_n=(ae^{i\theta},q)_n(ae^{-i\theta},q)_n$
\end{itemize}
\section{Double-scaled SYK model}
\label{sec:DSSYK}

\subsection{Double-scaled SYK model and chord diagrams} \label{sec:2.1}
This section is a brief review of the chord diagram approach to the double-scaled SYK (DSSYK) model. Acquainted readers can skip this section. The SYK model is a model consisting of $N$ flavors of Majorana fermions ($N$ even) described by the following random $p$-body Hamiltonian
\begin{equation}
    H=\i^\frac{p}{2}\sum_{i_1<i_2\ldots <i_p} J_{i_1i_2\ldots i_p}\psi_{i_1}\psi_{i_2}\ldots \psi_{i_p},\quad \{\psi_i,\psi_j\}=2\d _{ij}
\end{equation}
where $J_{i_1i_2\ldots i_p}$ are Gaussian random variables with the following distribution
\begin{equation}
    \avg{J_{i_1i_2\ldots i_p}J_{i_1'i_2'\ldots i_p'}}=\binom{N}{p}^{-1}\mathcal J^2\delta_{i_1i_1'}\ldots \delta_{i_pi_p'} \label{2.2}
\end{equation}
where $\avg{\ldots}$ denotes the ensemble average. For simplicity, we will set $\mJ=1$. The double-scaled limit is defined by taking $N$ and $p$ to infinity while keeping the following ratio fixed
\begin{equation}
    \lambda=\frac{2p^2}{N}
\end{equation}

It had been shown that in the double-scaled limit correlators such as $\avg{\tr{(H^{2k})}}$ can be calculated using combinatorial objects called chord diagrams\cite{Berkooz:2018qkz,Berkooz:2018jqr,Berkooz:2020xne,Berkooz:2020uly,Berkooz:2020fvm,Berkooz:2024lgq}. To calculate $\avg{\tr{(H^{2k})}}$, noting that the random coupling $J_{i_1\cdots i_p}$ has Gaussian distribution, it becomes the sum over all possible Wick contractions of $2k$ numbers of $H$. These Wick contractions have intersections and to evaluate them we need to use the anti-commutation relation of $\psi_i$. It turns out that the double-scaled limit largely simplifies the evaluation by replacing each intersection by its ensemble averaged value \be 
q\equiv e^{-\lambda} \label{2.4}
\ee
Therefore, we will have the following chord diagram calculus for $\avg{\tr{(H^{2k})}}$. We put $2k$ points on a circle, then we enumerate all possible pairings of these points, for each pairing we draw a chord diagram with each pair of points connected by a line called chord. Next we count the number of crossings between chords in each diagram, note that here we only consider 'topological' crossings which cannot be undone by moving chords with their end points fixed. For each crossing we assign a factor $q$, thus the value of each diagram is given by $q^{\#(\text{crosssings})}$. Finally we sum over all diagrams to get $\avg{\tr{(H^{2k})}}$.\footnote{Our normalization of $\mJ=1$ leads to $\avg{\tr(H^2)}=1$. } Figure~\ref{fig:chord-diagram-examples} shows examples of chord diagrams with $k=3$.

\begin{figure}[t]
  \centering

  \begin{subfigure}{0.18\textwidth}
    \includegraphics[width=\linewidth]{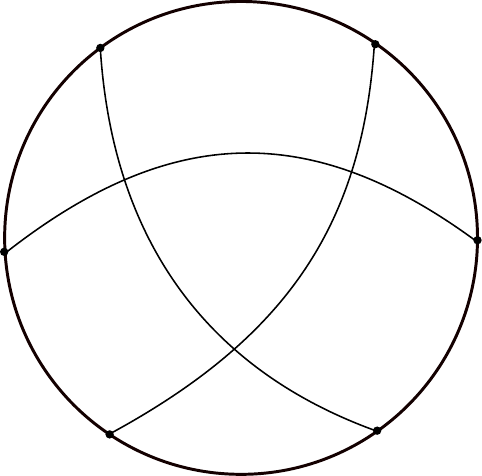}
    \caption{}\label{fig:chord1}
  \end{subfigure}
  \begin{subfigure}{0.18\textwidth}
    \includegraphics[width=\linewidth]{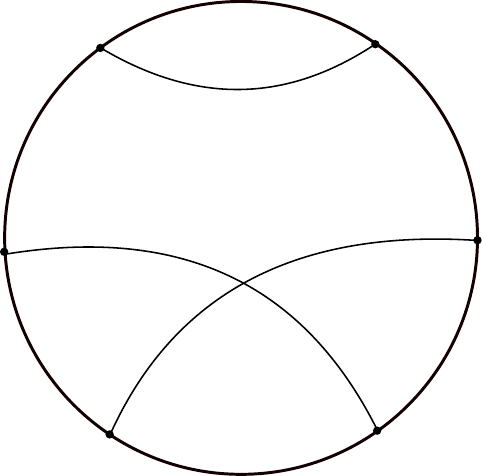}
    \caption{}\label{fig:chord2}
  \end{subfigure}
  \begin{subfigure}{0.18\textwidth}
    \includegraphics[width=\linewidth]{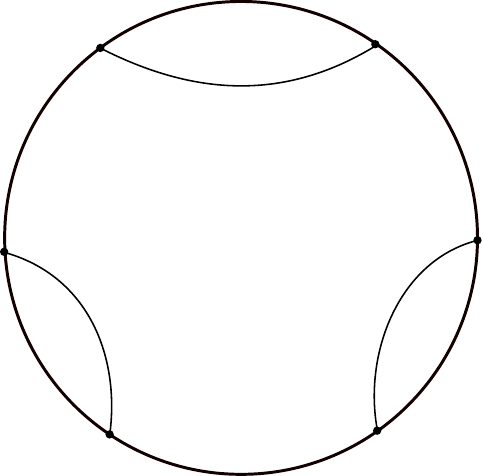}
    \caption{}\label{fig:chord3}
  \end{subfigure}
  \begin{subfigure}{0.18\textwidth}
    \includegraphics[width=\linewidth]{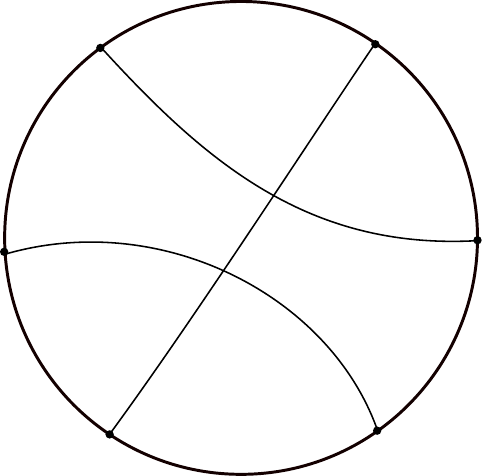}
    \caption{}\label{fig:chord7}
  \end{subfigure}
  \begin{subfigure}{0.18\textwidth}
    \includegraphics[width=\linewidth]{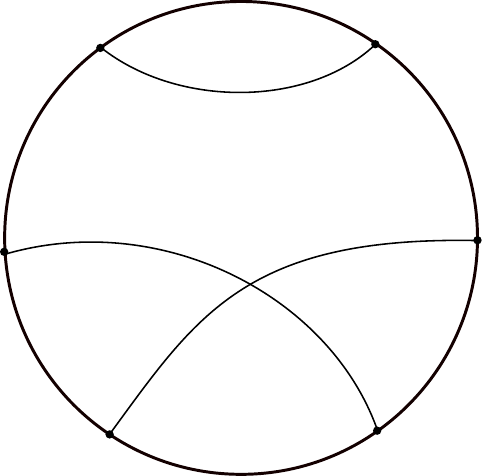}
    \caption{}\label{fig:chord8}
  \end{subfigure}

  \caption{Topologically inequivalent chord diagrams contributing to 
  $\langle \mathrm{tr}(H^6) \rangle$, respectively taking values 
  $q^3$, $q$, $1$, $q^2$ and $q$. Note that multiplicity of each diagram has to be accounted for when calculating the trace.}
  \label{fig:chord-diagram-examples}
\end{figure}

With these diagrammatic rules, quantities such as the partition function $\avg{\tr(e^{-\beta H})}$ can be explicitly computed.

Besides the Hamiltonian, we also consider general ``matter" observables of the following form
\begin{equation}
    M=\i^\frac{p'}{2}\sum_{i_1<i_2\ldots <i_{p'}} J'_{i_1i_2\ldots i_{p'}}\psi_{i_1}\psi_{i_2}\ldots \psi_{i_{p'}} \label{eq:2.5}
\end{equation}
where $J'_{i_1i_2\ldots i_{p'}}$ are again Gaussian random variables with 
\begin{equation}
    \avg{J'_{i_1i_2\ldots i_{p'}}J'_{i_1'i_2'\ldots i_{p'}'}}=\binom{N}{p'}^{-1}(\mathcal {J}')^2\delta_{i_1i_1'}\ldots \delta_{i_{p'}i_{p'}'}
\end{equation}
and we will again set $\mJ'=1$ for simplicity. One can consider multiple types of matter $M$ with different $p'$ and independent Gaussian random variables $\mJ'_{i_1\cdots i_{p'}}$. In the scope of this paper, we will limit ourselves to just one type of matter. We are interested in moments of the form $\avg{\tr(H^{k_1}M^{k_2}H^{k_2}M^{k_3}\cdots)}$, and they can be calculated using the following diagrammatic rule: draw the sequence $H^{k_1}M^{k_2}H^{k_2}M^{k_3}\cdots$ counterclockwise on a circle with $H$ and $M$ insertions denoted by different types of points, then enumerate all possible pairings of these points where only the same type of points are allowed to be paired, and connect all pairs with chords to construct a chord diagram. In this case the factor assigned to each crossing depends on the type of chords, crossings between two $H$ chords, two $M$ chords, one $H$ chord and one $M$ chord respectively give factors 
\begin{align} 
q&\equiv e^{-\lambda_{HH}},\quad q_{m}\equiv e^{-\lambda_{MM}},\quad r\equiv e^{-\lambda_{HM}} \label{2.7}\\
\lambda_{HH}&=\frac{2p^2}{N},\quad \lambda_{MM}=\frac{2p'^2}{N},\quad \lambda_{HM}=\frac{2pp'}{N}
\end{align}
Figure~\ref{mixed chord diagram examples} shows examples of chord diagrams with two different types of chords. Cases with more observables can be analyzed in a similar way with different factors assigned to all possible types of crossings. Detailed calculations of correlators of general observables can be found in~\cite{Berkooz:2018jqr}.
\begin{figure}
	\centering
	\begin{subfigure}{0.25\textwidth}
		\includegraphics[width=\textwidth]{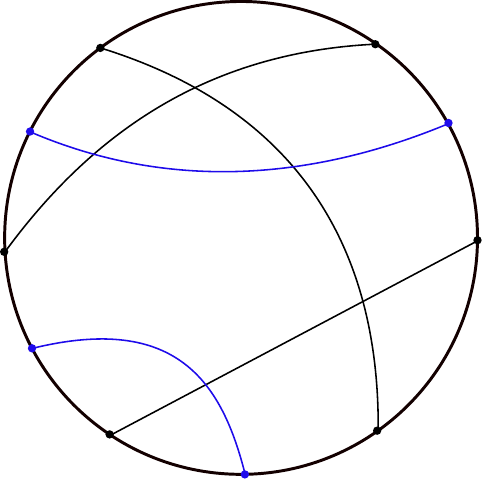}
        \captionsetup{format=hang}
		\caption{}
	\end{subfigure}
        \hspace{5em}
	\begin{subfigure}{0.25\textwidth}
		\includegraphics[width=\textwidth]{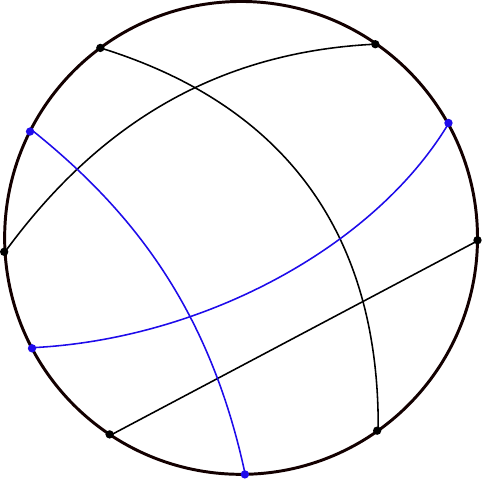}
        \captionsetup{format=hang}
        \caption{}
        \end{subfigure}
        \caption{Examples of chord diagrams contributing to $\avg{\tr(H^2MH^2MHMHM)}$, respectively taking values $q^2r^3$ and $q^2r^3q_m$. One needs to sum over all possible chord diagrams. $H$ and $M$ chords are respectively denoted by black solid lines and blue lines.}
        \label{mixed chord diagram examples}
\end{figure}

\subsection{The $q$-deformed harmonic oscillator and the two-sided algebra}

The chord diagram rule of DSSYK is equivalent to the $q$-deformed harmonic oscillator and can be solved using algebraic method. Let us first consider the case where we only have $H$, which only allows one type of chords. We can draw each chord diagram in the horizontal way (see Figure \ref{fig:horizontal_chord} for an illustration), and each $H$ can either emanate an open chord from the boundary or annihilate an an open chord to the boundary. Therefore, we can write $H$ as the sum of a creation and annihilation operators
\begin{figure}
    \centering
    \includegraphics[width=0.5\linewidth]{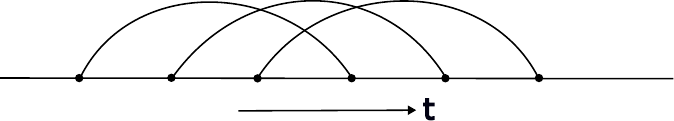}
    \caption{Chord diagrams can be equivalently drawn horizontally, here we show the horizontal version of the chord diagram Figure~\ref{fig:chord1}. We can trace the diagram from the left to the right, each Hamiltonian insertion on the boundary  then either opens a new chord or annihilates an existing one.}
    \label{fig:horizontal_chord}
\end{figure}
\be 
H=a^\dag+a \label{2.9}
\ee
where $a^\dag$ creates an open chord and $a$ annihilates an open chord. Starting with zero-chord state $\ket{\w}$, which we refer as the vacuum, the state with $n$ open chords is denoted by 
\be 
\ket{n}=(a^{\dag})^n\ket{\w} \label{2.9-1}
\ee
On the other hand, for $a$ acting on $\ket{n}$, there are $n$ different choices of annihilating an open chord, which by the crossing rule \eqref{2.4} results in a nontrivial factor
\be 
a\ket{n}=(1+q+\cdots+q^{n-1})\ket{n-1}=\f{1-q^n}{1-q}\ket{n-1}
\ee
For the vacuum $\ket{\w}$, $a$ cannot annihilate any open chords, and we must have $a\ket{\w}=0$. It follows that $a^\dag$ and $a$ obey a $q$-deformed harmonic oscillator algebra
\be
[a,a^\dag]_q\equiv a a^\dag-q a^\dag a=1 \label{2.11}
\ee
which reduces to a harmonic oscillator when $q\ra 1$.

Given this vector space spanned by $\{\ket{n}\}$, we define the hermitian conjugation as \cite{Xu:2024hoc} 
\be 
(a)^\dag \equiv a^\dag,\quad (a^\dag)^\dag\equiv a
\ee
and the bra states and inner product follow as
\be
\bra{n}\equiv\bra{\w}a^n,\quad \avg{n|m}=\avg{\w|a^n(a^\dag)^m|\w}=[n]!_q \d_{mn},\quad \avg{\w|\w}\equiv 1 \label{2.14-1}
\ee
The normalization $[n]!_q$ of the inner product is computed using the commutation relation \eqref{2.11}, which can also be understood as the sum over all possible intersections between $m$ and $n$ open chords as illustrated in Figure \ref{fig:4a}. This defines a Hilbert space $\mH_0$ of pure $H$ open chords.

Analogous to solving the harmonic oscillator, it has been shown in \cite{Berkooz:2018jqr} that $H$ can be diagonalized using $q$-Hermite polynomial
\be 
H\ket{\t}=E(\t)\ket{\t},\quad \ket{\t}=\sum_{n=0}^\infty \f{(1-q)^{n/2}}{(q;q)_n}H_n(\cos \t|q) \ket{n} \label{2.14}
\ee
where the eigenvalues are
\be 
E(\t)=\f{2\cos \t}{\sqrt{1-q}},\quad \t\in[0,\pi] \label{2.16}
\ee
and the $q$-Hermite polynomial is defined through the recurrence relation
\be 
 2xH_n(x|q)=H_{n+1}(x|q)+(1-q^n)H_{n-1}(x|q),\quad H_0(x|q)=1
\ee
which obeys the orthogonality
\begin{align}
\int_{0}^{\pi}\f{d\t}{2\pi}\f{(e^{2\i\t},e^{-2\i\t},q;q)_{\infty}}{(q;q)_{n}}H_{m}(\cos\t|q)H_{n}(\cos\t|q) & =\d_{nm}
\end{align}
By the orthogonality, the eigen states $\ket{\t}$ are normalized as
\be 
\avg{\t|\t'}=\f{2\pi }{\r(\t)}\d(\t-\t'),\quad \r(\t)\equiv (e^{2\i\t},e^{-2\i\t},q;q)_{\infty}
\ee

Including the matter operator $M$ permits a simple generalization. Since the calculus with both $H$ and $M$ are essentially the same except three different crossing factors \eqref{2.7}, we can decompose $M$ as the sum of another pair of creation and annihilation operators. However, as we have two types of creation operators, which do not commute, the space of open chords is spanned by states labeled by a string in the form of $\ket{H^{k_1}M^{k_2}H^{k_3}M^{k_4}\cdots}$. For each creation operator acting on such a state of string, there are two inequivalent ways by adding an $H$ (or $M$) letter to the right or to the left of the string. This is different from the previous case with only $H$ chords because adding an $H$ letter from either side are the same. For a simpler notation, from now on we will put  both $H$ and $M$ on equal footing as
\be 
H_0\equiv H,\quad H_1\equiv M
\ee

Let us first work on the definition of acting from the right. The operators are defined as
\be 
H_{R,0}\equiv H_R =a_{R,0}+a_{R,0}^\dag,\quad H_{R,1}\equiv M_R=a_{R,1}+a_{R,1}^\dag \label{2.15}
\ee
and the states are labeled by a binary string
\begin{equation}
\ket{i_{k}\cdots i_{1}}=a_{R,i_{1}}^{\dagger}\cdots a_{R,i_{k}}^{\dagger}\ket{\w},\quad i_j=0,1\label{eq:1-1}
\end{equation}
where each $a^\dag_{R,i}$ acts to the right of the binary string by adding a digit $i\in \{0,1\}$. In the chord diagram language, these states are illustrated in Figure \ref{chord-inner-product}. The annihilation operators $a_{R,i}$ is defined as pulling an open chord of $i$-type to the right boundary \cite{Lin:2022rbf}, which in the algebraic language means
\begin{align}
a_{R,i}\ket{i_{1}\cdots i_{k}}&=\sum_{j=1}^{k}\d_{i_{j},i}\prod_{s=j+1}^{k}Q_{i,i_{s}}\ket{i_{1}\cdots\slashed{i_{j}}\cdots i_{k}} \label{2.17}\\
Q_{0,0}=q,&\quad  Q_{0,1}=Q_{1,0}=r,\quad Q_{1,1}=q_m
\end{align}
where $\slashed{{i}_{j}}$ means the digit $i_{j}$ is deleted from the
string $i_{1}\cdots i_{k}$ when it coincides with $i$. The coefficients $Q_{ij}$ is from the crossing factor \eqref{2.7} because when the $i$-type open chord is deleted it intersects with all open chords $i_{j+1}\cdots i_k$ to its right. Same as before, the annihilation operators should also kill the vacuum, i.e. $a_{R,i}\ket{\w}=0$. Using this definition, it follows that they obey a generalized $q$-deformed harmonic oscillator algebra
\be 
a_{R,i} a_{R,j}^\dag- Q_{ij}a_{R,j}^\dag a_{R,i}=\d_{ij},\quad Q=\begin{bmatrix}q & r\\
r & q_{m}
\end{bmatrix}, \quad i=0,1 \label{2.19}
\ee

\begin{figure}
	\centering
	\begin{subfigure}{0.28\textwidth}
    \centering
		\includegraphics[height=2.2cm]{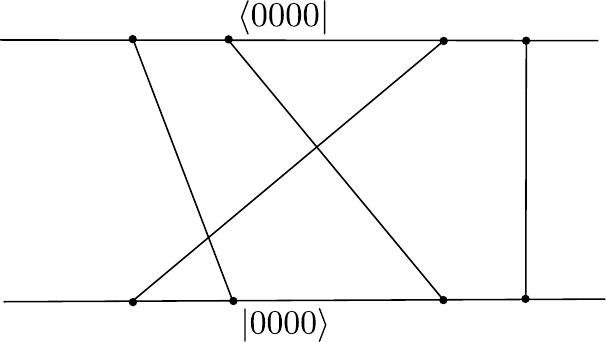}
        \captionsetup{format=hang}
		\caption{}
        \label{fig:4a}
	\end{subfigure}
	\begin{subfigure}{0.32\textwidth}
    \centering
		\includegraphics[height=2.2cm]{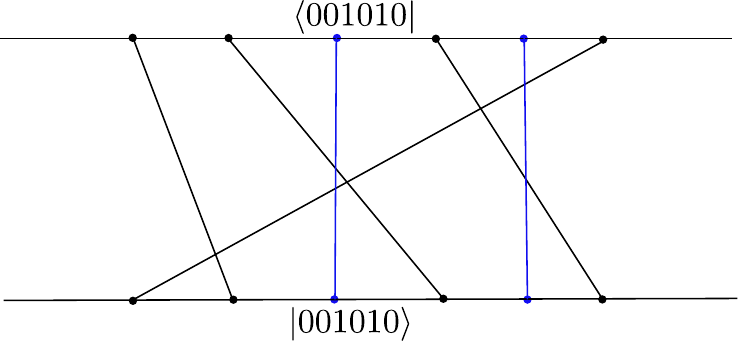}
        \captionsetup{format=hang}
        \caption{}
        \label{fig:4b}
    \end{subfigure}
	\begin{subfigure}{0.33\textwidth}
    \centering
		\includegraphics[height=2.2cm]{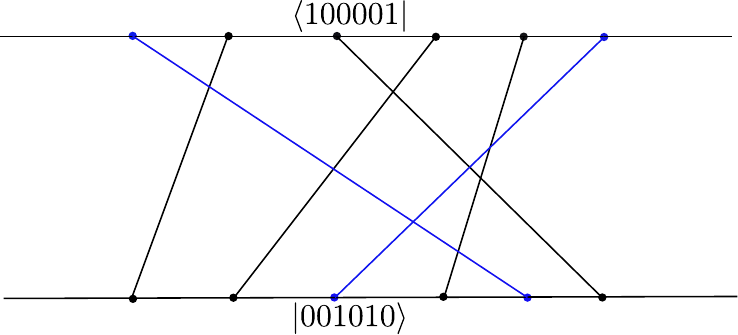}
        \captionsetup{format=hang}
        \caption{}
        \label{fig:4c}
    \end{subfigure}
    \caption{Inner products between two chord states are calculated by summing over all possible ways to pair up chords of the same type and counting crossings, with each crossing contributing a factor. Here we include some examples, Figure~\ref{fig:4a},~\ref{fig:4b},~\ref{fig:4c} respectively contribute to inner products $\braket{0000}$, $\braket{001010}$ and $\braket{100001}{001010}$, with values $q^2$, $q^3r^4$ and $q^2r^4s$.
    }
    \label{chord-inner-product}
\end{figure}

The hermitian conjugation is defined in the same way
\be 
(a_{R,i})^\dag=a_{R,i}^\dag,\quad (a_{R,i}^\dag)^\dag=a_{R,i}
\ee
and the inner product is defined accordingly as
\be 
\avg{i_{1}\cdots i_{k}|j_1\cdots j_l}=\avg{\w|a_{R,i_1}\cdots a_{R,i_k} a_{R,j_l}^\dag\cdots a_{R,j_1}^\dag|\w} \label{2.21}
\ee
which can be computed using the commutation relation \eqref{2.19} to move all $a_{R,i}$ to the right to kill the vacuum. Indeed, this is equivalent to summing over all possible crossing patterns between two states of open chords \cite{Xu:2024hoc}, as illustrated in Figure \ref{fig:4b}. It is easy to see that the Hilbert space $\mH$ equipped with this inner product has the following decomposition labeled by the number $N_i$ of open chords of $i$-type
\begin{equation}
\mH=\oplus_{N_{0},N_{1}\in\N}\mH_{N_{0},N_{1}} \label{2.28}
\end{equation}
because \eqref{2.21} is nonzero only when the numbers of 0 digit and 1 digit are the same in the two binary strings $i_1\cdots i_k$ and $j_1\cdots j_l$. For each subspace $\mH_{N_{0},N_{1}}$, there are $(N_{0}+N_{1})!/(N_{0}!N_{1}!)$
numbers of states. Under the inner product \eqref{2.21}, these binary
string states (\ref{eq:1-1}) are not orthogonal within each subspace
$\mH_{N_{0},N_{1}}$. 

It is clear that the Hilbert space $\mH_0$ with only $H_0$ chords is the subspace $\oplus_{N_0\in\N}\mH_{N_0,0}$, where each $\mH_{N_0,0}$ is one-dimensional. Within this subspace, the spectrum of $H_{R,0}$ is solved in \eqref{2.14}. However, the full spectrum of $H_{R,0}$ has not been solved yet in the full Hilbert space though solutions are available in one-matter and two-matter subspaces \cite{Lin:2023trc,Xu:2024hoc}. However, the method in \cite{Lin:2023trc,Xu:2024hoc} quickly becomes too involved in practice for higher-matter subspaces. Instead, we will develop a completely different method in Section \ref{sec:3.4}, which is inspired by the end-of-the-world brane, to solve the full spectrum of $H_{R,0}$ in a much cleaner way.

The operators of acting from the left can be defined in a similar way. We have 
\be 
H_{L,0}\equiv H_L =a_{L,0}+a_{L,0}^\dag,\quad H_{L,1}\equiv M_L=a_{L,1}+a_{L,1}^\dag \label{2.23}
\ee
and the states of binary string can also be generated by the left creation operators
\begin{equation}
\ket{i_{k}\cdots i_{1}}=a_{L,i_{k}}^{\dagger}\cdots a_{L,i_{1}}^{\dagger}\ket{\w},\quad i_j=0,1\label{eq:1-2}
\end{equation}
where each $a^\dag_{L,i}$ acts to the left of the binary string by adding a digit $i\in \{0,1\}$. The annihilation operators are defined as
\begin{align}
a_{L,i}\ket{i_{1}\cdots i_{k}}&=\sum_{j=1}^{k}\d_{i_{j},i}\prod_{s=1}^{j-1}Q_{i,i_{s}}\ket{i_{1}\cdots\slashed{i_{j}}\cdots i_{k}} \label{2.25}
\end{align} 
and they obey the same generalized $q$-deformed harmonic oscillator algebra
\be 
a_{L,i} a_{L,j}^\dag- Q_{ij}a_{L,j}^\dag a_{L,i}=\d_{ij} \label{2.26}
\ee

It is clear that the left and right operators are both defined in the same Hilbert space $\mH$. Using their definitions, one can work out their mutual commutation relations \cite{Xu:2024hoc}
\be
[a_{L,i},a_{R,j}^\dag]=[a_{R,i},a_{L,j}^\dag]=\d_{ij}Q_{0,i}^{N_0} Q_{1,j}^{N_1},\quad [a_{L,i},a_{R,j}]=[a_{L,i}^\dag,a_{R,j}^\dag]=0 \label{2.27}
\ee
where $N_0$ and $N_1$ are the operators counting how many 0 and 1 digits respectively in the binary string of a state $\ket{i_1\cdots i_k}$. As we see from \eqref{2.27} that the left and right operators are not commutative to each other. However, the hermitian operators \eqref{2.15} and \eqref{2.23} on each side do commute with each other
\be 
[H_{L,i},H_{R,j}]=[a_{L,i},a_{R,j}^\dag]+[a_{L,i}^\dag,a_{R,j}]=\d_{ij}(Q_{0,i}^{N_0}Q_{1,j}^{N_1}-Q_{0,j}^{N_0}Q_{1,i}^{N_1})=0
\ee
It follows that the left and right von Neumann algebras generated by $H_{L,i}$
 and $H_{R,i}$ respectively
\be 
\mA_L=\{H_{L,0},H_{L,1}\}'',\quad \mA_R=\{H_{R,0},H_{R,1}\}'' \label{2.29}
\ee
are two commuting algebras
\be 
\mA_L\subseteq \mA_R',\quad \mA_R\subseteq \mA_L' \label{2.30}
\ee
where the prime of $\mA'$ means the commutant of $\mA$, which consists of all bounded operators that commute with $\mA$, and the double prime in \eqref{2.29} is the crucial completion procedure to make a von Neumann algebra.\footnote{Some pedagogical introduction to von Neumann algebra can be found, e.g. in \cite{Witten:2018zxz,Liu:2025krl}. In Appendix~\ref{app:math} we included a minimal introduction to relevant concepts in von Neumann algebras.} 

The left and right algebras $\mA_{L,R}$ are related to each other by the modular conjugation operator $J$, which simply flips the ordering of the binary string when acting on a state
\be 
J\ket{i_{k}\cdots i_{1}}=\ket{i_{1}\cdots i_{k}}
\ee
By Tomita's theory, this modular conjugation $J$ is anti-unitary, unique, and $J^{-1}=J=1$. It also maps $\mA_{L,R}$ to its commutant
\be 
J\mA_{L}J=\mA_L',\quad J\mA_R J=\mA_R' \label{2.32}
\ee
On the other hand, we can easily see from the definition \eqref{eq:1-1}, \eqref{2.17}, \eqref{eq:1-2}, and \eqref{2.25} that
\be 
J H_{L,i}J=H_{R,i},\quad J H_{R,i}J=H_{L,i}
\ee
Together with \eqref{2.32}, we find that $\mA_{L,R}$ are commutant to each other
\be 
\mA_L'=\mA_R,\quad \mA_R'=\mA_L
\ee

This generalized $q$-harmonic oscillator algebra is conventionally called $q$-Gaussian algebra in mathematical literature~\cite{bozejko1991example,speicher1993generalized,Bozejko1994,bozejko1997q,Ricard2003FactorialityOQ,bozejko2017fock,skalski2018remarks}. It has been proved under different assumptions~\cite{bozejko1997q,Ricard2003FactorialityOQ,skalski2018remarks} that these algebras are type $\text{II}_1$ factors. Type II$_1$ means the Hilbert space does not allow a factorization into tensor of left and right but still allows a trace, which maps identity to one. In the DSSYK, the trace for $\mA_{L}$ and $\mA_R$ is given by $\avg{\w|\cdot|\w}$, which has the cyclic property (invariant under circular shift)\footnote{To be a trace, it also needs to be a faithful, normal, semifinite, and positive linear functional, which is proved in \cite{Xu:2024hoc}.}
\be 
\avg{\w|ab|\w}=\avg{\w|ba|\w},\quad \forall a,b\in\mA_{L,R} \label{2.40}
\ee
Factor means the intersection between the algebra and its commutant is proportional to identity. Since $\mA_L$ and $\mA_R$ are commutant to each other, it follows that $\mA_L$ and $\mA_R$ altogether generates all bounded operators $\mB(\mH)$ of the Hilbert space $\mH$
\be
\mathcal A_L\vee \mathcal A_R\equiv (\mA_L\cup \mA_R)''=\mathcal B(\mH)
\ee

\subsection{The bulk dual description and triple scaling/JT limit}

It had been proposed that the SYK model with $p\sim O(1)$ in low temperature is dual to a near-$\text{AdS}_2$ (NAdS) bulk in the sense that low energy modes in the SYK model agree with boundary graviton modes in the NAdS space, and they share the same symmetry and effective action of a Schwarzian derivative \cite{Maldacena:2016upp}. In the SYK model there is a exact conformal symmetry at zero energy which is broken down to SL(2,$\mathbb R$) by taking into consideration low energy excitations, while in the NAdS space the exact conformal symmetry is broken down to SL(2,$\mathbb R$) by a cutoff near infinity, see~\cite{Maldacena:2016upp,Sarosi:2017ykf,Kitaev:2017awl,Trunin:2020vwy} for reviews of these results. It is then natural to ask whether such a duality persists in the double scaled limit, and quite surprisingly, it is not only preserved but also enhanced to a duality between Hilbert spaces~\cite{Lin:2022rbf,Lin:2023trc}. In this subsection, we will restrict ourselves to zero-matter subspace $\mH_0$. An interesting observation from Section \ref{sec:2.1} is that the ordinary trace tr in the original SYK model, e.g. $\avg{\tr(e^{-\beta H_{SYK}})}$, becomes a transition amplitude $\avg{\w|e^{-\beta (H_{L,0}+H_{R,0})/2}|\w}$ after ensemble average and in the double-scaled limit. Note here that $H_{R,0}\ket{\w}=H_{L,0}\ket{\w}$ acting on the vacuum state. This amplitude is compatible with the disk topology of the dual bulk of AdS$_2$ while the boundary has a topology of a circle that is compatible with an ordinary trace of the original SYK model. Indeed, we can view $\ket{\w}$ as the maximally entangled state, the infinite temperature Hartle-Hawking state, of two boundaries and there is a bulk Hamiltonian 
\be 
H_{bulk}=\frac{1}{2}(H_{L,0}+H_{R,0}) \label{2.43}
\ee
This rewriting can be visualized by cutting a chord diagram into upper and lower two halves representing the bra and ket states of left and right boundaries. The vacuum states $\ket{\w}$ and $\bra{\w}$ live on the two antipodal points, and the discrete averaged boundary time $t=t_L+t_R$ now flows from one antipodal point to the other, see Figure~\ref{bulk-slices} for an illustration. 
\begin{figure}
	\centering
	\begin{subfigure}{0.38\textwidth}
		\includegraphics[width=\textwidth]{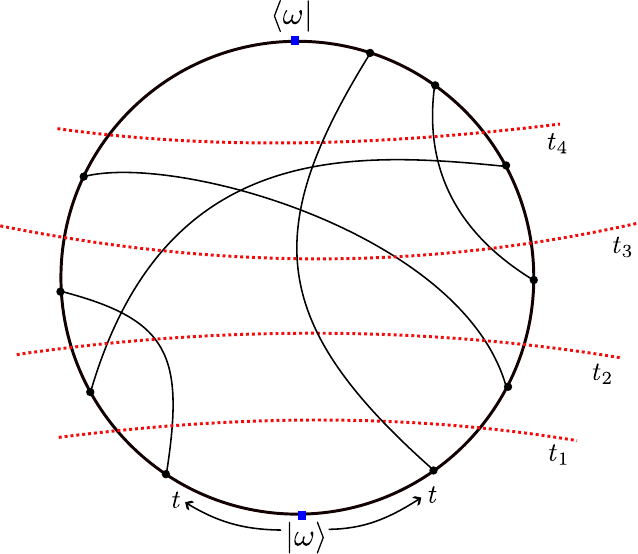}
        \captionsetup{format=hang}
		\caption{}
        \label{bulk-slices}
	\end{subfigure}
        \hspace{5em}
	\begin{subfigure}{0.3\textwidth}
		\includegraphics[width=\textwidth]{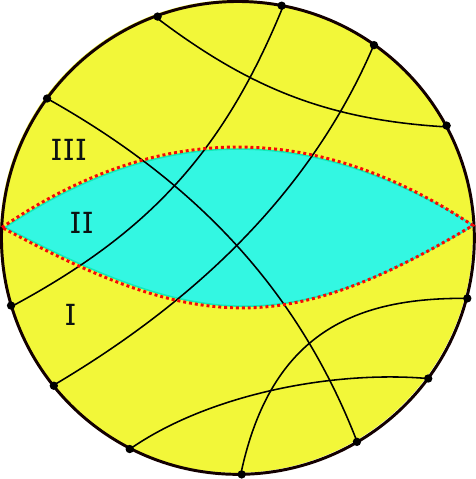}
        \captionsetup{format=hang}
        \caption{}
        \label{inner-product-diagram}
        \end{subfigure}
        \caption{(a): Illustration of two-sided boundary time and corresponding bulk slices. States defined on the four time slices $t_1$-$t_4$ are respectively $\ket{2},\ket{4},\ket{4},\ket{2}$. (b): Slicing the chord diagram into upper, middle and lower parts, where the inner product is defined by the middle part (II). In this case the middle part contributes to the inner product $\ip{3}{3}$.}
\end{figure}

On the other hand, this averaged boundary time flow provides a natural foliation of chord diagrams which contribute to the partition function. For each $t$, we define a bulk time slice as shown in Figure~\ref{bulk-slices} which is uniquely chosen in the way that all chords intersecting the slice do not cross before they intersect the slice. Thus we can define a bulk Hilbert space on the time slice with the bulk state specified by the number of open chords intersected by the slice, which is proportional to the state $\ket{n}$ defined in \eqref{2.9-1}. The boundary Hamiltonian insertions then either add a new chord or annihilate an existing chord. Note that annihilating chords necessarily results in extra factors as chords being annihilated necessarily cross other chords before reaching the boundary. Now we can interpret the upper and lower regions as preparing states $H_{bulk}^{n_1}\ket{\w}$ and $H_{bulk}^{n_2}\ket{
\w}$ on the cuts and the middle region as defining an inner product between states of different chord numbers. Since these chords could cross in all different ways in the middle region, the inner product is given by summing over all possible crossing patterns, which is exactly the inner product we defined before in \eqref{2.14-1}. See Figure~\ref{inner-product-diagram} for illustration.

In the zero-matter subspace, since $H_{L,0}\ket{\w}=H_{R,0}\ket{\w}$, the bulk Hamiltonian \eqref{2.43} is equivalent to $H$ in \eqref{2.9} and can be solved by \eqref{2.14}. To connect this bulk interpretation with JT gravity, we may first write $\l_b\equiv \lam n$ \cite{Lin:2022rbf}, and the creation and annihilation operators become translation in $\l_b$
\be 
a^\dag = \sqrt{\f{1-e^{-\l_b-\lam}}{1-e^{-\lam}}} e^{\i k \lam},\quad a = \sqrt{\f{1-e^{-\l_b}}{1-e^{-\lam}}} e^{-\i k \lam},\quad k\equiv -\i\del_{\l_b} \label{2.44}
\ee
where the prefactors in front of $e^{\pm \i k\lam}$ accounts for the normalization of $\ket{n}$ in \eqref{2.14-1}. To reduce to JT gravity, we need to take the following triple scaling/JT limit
\begin{equation}
  \lam \ra 0,\quad  e^{-\l}\equiv e^{-\l_b}/\lam^2 \text{ fixed} \label{2.46}
\end{equation}
in which the zero-matter Hilbert space basis $\ket{n}$ are projected to a continuous sector with $n\ra \infty$. To consider the low energy regime, we need to shift $k\ra\pi/\lam-k$ and expand \eqref{2.9} in quadratic order of $\lam$, which leads to
\begin{equation}
   \lam^{^{-3/2}} (H_{bulk}-E_0)=k^2+e^{-\l},\quad E_0=-\f{2}{\sqrt{1-q}}\app-2/\sqrt{\lam} \label{2.47}
\end{equation}
where $E_0$ is the ground energy. This is exactly the Liouville potential of the JT gravity in thermofield double state via canonical quantization \cite{Harlow:2018tqv}.\footnote{In our notation, to make $\ket{\l}$ basis normalized as a delta function, we need to rescale \eqref{2.2} by a factor of $1/\lam$, which is the notation used in \cite{Lin:2022rbf}. \label{ftn1}}

\section{End-of-the-world brane in DSSYK} \label{sec:2}
\label{sec:EoW brane}
\subsection{EoW brane as a deformed Hamiltonian: new chord diagrams} \label{sec:3.1}
End-of-the-world (EoW) branes are dynamical boundaries in JT gravity on which the spacetime ends. They can be used to model microstates of single-sided black holes in JT gravity~\cite{Penington:2019kki}, and are proposed to be the bulk dual of amplitudes between pure states in the SYK model~\cite{Kourkoulou:2017zaj}. In~\cite{Gao:2021uro} it has been proved that in pure JT gravity the dynamics of single-sided black holes ending on the EoW brane is equivalent to that of a single particle in the Morse potential, in contrast to the Liouville potential in the two-sided case~\cite{Harlow:2018tqv}. 

The action of JT gravity with an EoW brane is given by
\begin{equation}\label{JT action}
    S_{JT}=S_0+\frac{1}{8\pi G}\bigg\{\frac{1}{2}\int d^2x \Phi\sqrt{|g|}(R+2)
    +\int_{\text{AdS}} du\sqrt{|g_{uu}|}(K-1)+\int_{\text{EoW}} dv\sqrt{|g_{vv}|}(\Phi K-\mu_r)
    \bigg\}
\end{equation}
Here $S_0$ is the topological term which is a constant here. $\Phi$ and $g_{\mu\nu}$ are bulk dilaton and metric fields. $u$ and $v$ are coordinates on the AdS and EoW brane while $g_{uu}$ and $g_{vv}$ are the induced metrics respectively. Finally, $K$ is the extrinsic curvature of the boundary and $\mu_r$ is a parameter called the brane tension. On the AdS boundary we choose the boundary condition to be
\begin{equation}
    \Phi=\frac{\phi_b}{\epsilon},\;\;g_{uu}=-\frac{1}{\epsilon^2}
\end{equation}
while on the EoW boundary we choose
\begin{equation}\label{EoW boundary condition}
    K=0,\;\;n^\nu\nabla_\nu\Phi=\mu_r
\end{equation}
where $n^\nu$ is the normal derivative of the boundary. Note that $K=0$ means that the EoW brane is a geodesic in AdS$_2$. Using the phase space formulation and canonical quantization \cite{Gao:2021uro}, it can be shown that the dynamics in the presence of the EoW brane is captured by the following Hamiltonian
\begin{equation}
    H=\frac{2}{\phi_b}\left(\frac{P^2}{4}+\mu_r e^{-L}+e^{-2L}\right) \label{3.4}
\end{equation}
where $L$ is the normalized length of the geodesic normally enamating from the EoW brane and ending on the AdS boundary, and $P$ is the conjugate momentum. This potential is called the Morse potential, and it reduces to the Liouville potential in the $\mu_r\rightarrow 0$ limit. 

While it is unclear how to define EoW brane in DSSYK in prior because there is no geodesic in chord diagrams, comparing this Hamiltonian with \eqref{2.47}, we see that the additional term is mainly due to $\mu_r e^{-L}$. On the other hand, from the triple scaling/JT limit, \eqref{2.46} naturally suggest the following modification to the bulk Hamiltonian\footnote{A set of one-parameter deformations of this discrete Hamiltonian has been proposed in~\cite{Blommaert:2025avl} from a sine-dilaton gravity construction. These deformed Hamiltonians agree with ours in the triple scaling limit, while their chord diagrammatic interpretation is not clear (some comparison with DSSYK was discussed in \cite{Aguilar-Gutierrez:2025hty}). See also~\cite{Cui:2025sgy} for recent progress. This bulk Hamiltonian was also considered recently in \cite{Aguilar-Gutierrez:2025hty}. Another discussion of the diagrammatic interpretation of deformed Hamiltonians can be found in~\cite{Berkooz:2025ydg}.}
\be 
\tilde H_{bulk}=a+a^\dag +\mu q^n \label{3.5}
\ee
where $n$ is the number operator reading out the chord number of $\ket{n}$. Since $q^n\sim e^{-\l_b}$ in the semiclassical limit, it is analogous to the additional middle term of \eqref{3.4}. Indeed, the proposal is far from unique, but this Hamiltonian~\eqref{3.5} also agrees with the one proposed by Okuyama in~\cite{Okuyama:2023byh}. Before we go to the solution of this problem, let us first discuss what the EoW brane means in the chord diagram language.

From the definition \eqref{3.5}, it is clear that the role of $\tilde H_{bulk}$ now has three parts: 
\begin{enumerate}
    \item Creating an $H$ open chord;
    \item Annihilating an $H$ open chord;
    \item Measuring how many $H$ open chords in a state $\ket{n}$ with a weight $\mu q^n$.
\end{enumerate} 
While the first two actions are the same as $H_{bulk}$, the third action could be imagined as drawing an $H$ chord from a wall behind all $H$ open chords, and the intersection gives the factor $q^n$, multiplied by an coefficient $\mu$. This wall is the EoW brane in DSSYK. Diagrammatically, these three parts are illustrated in Figure \ref{fig:5}, where we draw the asymptotic boundary on the right and the EoW brane on the left.

\begin{figure}
	\centering
	\begin{subfigure}{0.22\textwidth}
		\includegraphics[width=\textwidth]{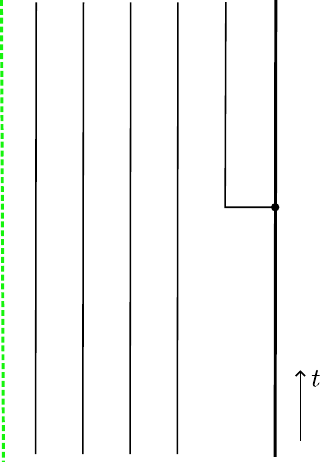}
		\caption{}
	\end{subfigure}
        \hspace{2em}
	\begin{subfigure}{0.22\textwidth}
		\includegraphics[width=\textwidth]{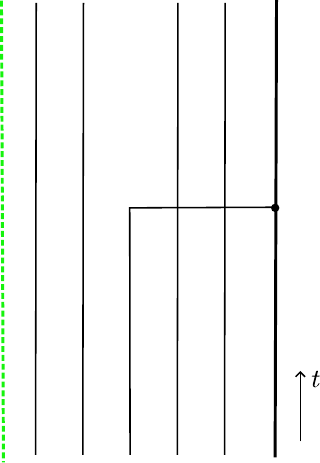}
        \caption{}
        \end{subfigure}
        \hspace{2em}
        \begin{subfigure}{0.22\textwidth}
		\includegraphics[width=\textwidth]{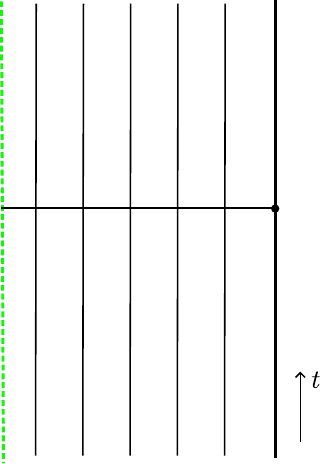}
        \caption{}
        \end{subfigure}
        \caption{The black line on the right the asymptotic boundary, and the green dashed line on the left is the EoW brane. In (a) $H$ creates a new open chord on the right. In (b) $H$ annihilates an existing open chord, in this case resulting in a factor $q^2$. In (c) $H$ draws a chord from the EoW brane, in this case multiplying the state with a factor $\mu q^5$.}
        \label{fig:5}
\end{figure}

Similar to the bulk dual description of DSSYK, where a complete chord diagram starts and ends with the vacuum state $\ket{\w}$, in the case of the EoW brane we can also consider the bulk amplitude $\avg{\w|\tilde H_{bulk}^n|\w}$, which corresponds to the chord diagrams illustrated in Figure \ref{fig:6}. We can easily summarize its diagrammatic rule as follows:
\begin{itemize}
    \item Choose $2k$ points where $0\leq k\leq \lfloor\frac{n}{2}\rfloor$ and pair them up in all possible ways, connect each pair with a chord. These chords will be denoted as bulk chords. 
    \item Draw a chord from each of the remaining $n-2k$ points to the EoW (left) boundary where these chords do not cross each other. This chords will be denoted as EoW chords.  
    \item Evaluate the chord diagram, each crossing between chords contributes a factor $q$, while each chord ending on the EoW brane contributes a factor $\mu$. 
    \item Sum over all chord diagrams to get $\avg{\w|\tilde H_{bulk}^n|\w}$.
\end{itemize}
The transition amplitude of Euclidean time evolution is then given by summing over all chord diagrams in an infinite series
\begin{equation}
    \avg{\w|e^{-\b \tilde H_{bulk}}|\w}=\sum_{n=0}^{\infty} \frac{(-\beta)^n}{n!}\avg{\w|\tilde H_{bulk}^n|\w}
\end{equation}

\begin{figure}
	\centering
	\begin{subfigure}{0.3\textwidth}
    \centering
		\includegraphics[height=6cm]{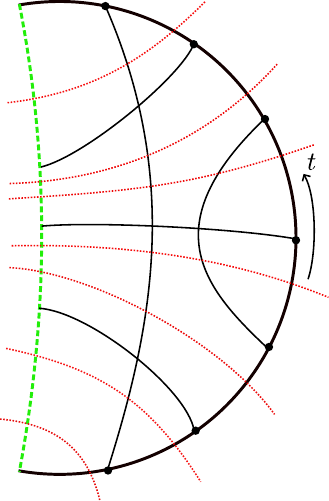}
		\caption{}
        \label{fig:6a} 
	\end{subfigure}
        \hspace{5em}
	\begin{subfigure}{0.35\textwidth}
    \centering
		\includegraphics[height=6cm]{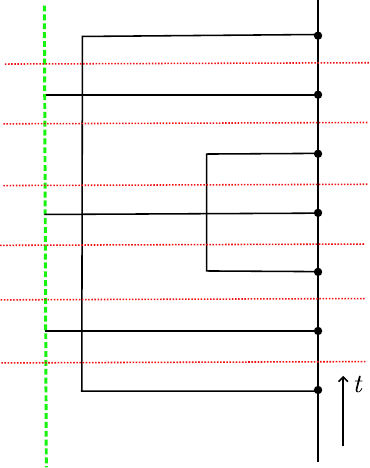}
        \caption{}
        \label{fig:6b}
        \end{subfigure}
        \caption{(a) Time slices (red dashed lines) in an chord diagram with EoW brane. (b) The diagram in (a) can be deformed into a strip. \label{fig:6}}
\end{figure}

There is also an unique way to assign the bulk slicing for a chord diagram (see Figure \ref{fig:6a}) analogous to the canonical Cauchy slicing in JT gravity with an EoW brane:
\begin{itemize}
    \item Bulk chords intersecting a time slice do not cross below the slice (this rule is the same as in the two-sided case in Figure \ref{bulk-slices}).
    \item EoW chords do not intersect time slices.
\end{itemize}
With these rules a chord diagram with EoW brane can be deformed into a strip shape as shown in Figure~\ref{fig:6b}. The state on each time is again specified by the number of open chords on the slice, which is actually proportional to the states $\ket{n}$ in the zero-matter subspace $\mH_0$. 

The Hamiltonian \eqref{3.5} is exactly diagonalized in zero-matter subspace $\mH_0$ by Okuyama in~\cite{Okuyama:2023byh}, and the eigen functions are another orthogonal polynomials called continuous big $q$-Hermite polynomials~\cite{Okuyama:2023byh,Floreanini1995AnAI}
\be 
\tilde H_{bulk}\ket{\mu,\t}=E(\t)\ket{\mu,\t},\quad \ket{\mu,\t}=\sum_{n=0}^\infty \f{(1-q)^{n/2}}{(q;q)_n}H_n(\cos \t;\mu\sqrt{1-q}|q) \ket{n} \label{3.7}
\ee
where the eigenvalues are the same as before
\be 
E(\t)=\f{2\cos \t}{\sqrt{1-q}},\quad \t\in[0,\pi] \label{3.8}
\ee
and the big $q$-Hermite polynomial is defined through the recurrence relation
\be 
 2xH_n(x;a|q)=H_{n+1}(x;a|q)+aq^n H_n(x;a|q)+(1-q^n)H_{n-1}(x;a|q),~ H_0(x;a|q)=1 \label{eq:3.9}
\ee
which for $a\in\R$ and $|a|<1$ obeys the orthogonality
\begin{align}
\int_{0}^{\pi}\f{d\t}{2\pi}\f{(e^{2\i\t},e^{-2\i\t},q;q)_{\infty}}{(ae^{\i\t},ae^{-\i\t};q)_{\infty}(q;q)_{n}}H_{m}(\cos\t;a|q)H_{n}(\cos\t;a|q) & =\d_{nm} \label{eq:3.10}
\end{align}
By the orthogonality, the eigen states $\ket{\mu,\t}$ are normalized as
\be 
\avg{\mu,\t|\mu,\t'}=\f{2\pi}{\r(\t;\mu)}\d(\t-\t'),\quad \r(\t;\mu)\equiv\f{(e^{2\i\t},e^{-2\i\t},q;q)_{\infty}}{(\mu\sqrt{1-q}e^{\i\t},\mu\sqrt{1-q}e^{-\i\t};q)_{\infty}} \label{3.11}
\ee
It follows that the states $\ket{n}$ can be written as
\begin{equation}
\ket{n}=\int_{0}^{\pi}\f{d\t}{2\pi}\f{\r(\t;\mu)}{(1-q)^{n/2}}H_{n}(\cos\t;\mu\sqrt{1-q}|q)\ket{\mu,\t}\label{eq:69}
\end{equation}
Taking $\mu\ra 0$, the big $q$-Hermite polynomial reduces back to the $q$-Hermite polynomial.

Now we confirm that the new Hamiltonian \eqref{3.5} with EoW brane does reduce to the Morse potential in \eqref{3.4} in the triple scaling/JT limit \cite{Okuyama:2023byh}. Using $\l_b$ we can rewrite
\be 
\tilde H_{bulk}=\sqrt{\f{1-e^{-\l_b-\lam}}{1-e^{-\lam}}} e^{i k \lam}+\sqrt{\f{1-e^{-\l_b}}{1-e^{-\lam}}} e^{-i k \lam} +\mu e^{-\l_b} \label{3.12}
\ee
Since we only have one asymptotic boundary in the EoW brane case, the triple scaling/JT limit is slightly modified as
\be
  \lam \ra 0,\quad  \mu= -\f{e^{-(\mu_r+1/2) \lam}}{\sqrt{1-e^{-\lam}}},\quad  e^{-L}\equiv e^{-\l_b}/\lam\text{ fixed} \label{3.13}
\ee
where we should note that $\mu$ has a minus sign. Similar to \eqref{2.47}, we shift $k\ra \pi/\lam -k$ for the low energy regime and expanding \eqref{3.12} in quadratic order of $\lam$, we find
\be 
 \lam^{-3/2}(\tilde H_{bulk}-E_0)=k^2+\mu_r e^{-L}+\f 1 4 e^{-2L} \label{3.14}
\ee
With a slightly redefinition of $L$ and $\mu_r$, this is identical to the Morse potential \eqref{3.4} from the bulk canonical quantization of JT gravity with an EoW brane.

\subsubsection*{Including matter}

Above definition for the EoW brane as a modification of the Hamiltonian in DSSYK is only in the zero-matter subspace $\mH_0$. From \eqref{2.28}, the full Hilbert space includes states with arbitrary open matter chords, on which we need to generalize the action of $\tilde H_{bulk}$. From the chord diagram interpretation in Figure \ref{fig:5}, a natural generalization is to include additional intersection factor $r^{N_1}$ when $\tilde H_{bulk}$ draws a chord from the EoW brane on a state with $N_1$ open chords. On the other hand, we have to distinguish operators acting from the left or right on the Hilbert space after including matter operators just like the ordinary DSSYK case. Since the chord diagram rules in Figure \ref{fig:5} have a natural interpretation as single side operation, and the left asymptotic boundary now is pictorially replaced by an EoW brane, we will use the binary notation and understand $\tilde H_{bulk}$ as a part of the right boundary operator. More precisely, we have
\be 
\tilde H_{R,0}\equiv H_{R,0}+\mu q^{N_0} r^{N_1}\simeq \tilde H_{bulk}
\ee
where $H_{R,0}$ is defined the same as \eqref{2.15}. As for the boundary matter operator, we will keep it the same as $H_{R,1}$ in \eqref{2.15}. In chord diagram language, this means that it only creates or annihilates matter chords from or to the right boundary but never draw a chord from the EoW brane. If we understand a chord from the boundary to the EoW brane as a creation-annihilation process between them, keeping $H_{R,1}$ unchanged should be understood as a reflection boundary condition for a bulk particle on the EoW brane because the matter cannot be created or absorbed by the EoW brane. On the other hand, the modification of $H_{R,0}$ to $\tilde H_{R,0}$ with an additional term $\mu q^{N_0} r^{N_1}$ reflects the fact that the EoW brane is a source of gravitational interaction.

With this setup, we can define the von Neumann algebra of the right side similar to the ordinary DSSYK
\be 
\tilde \mA_R=\{\tilde H_{R,0},\tilde H_{R,1}\}'',\quad \tilde H_{R,1}\equiv H_{R,1}\label{3.16}
\ee
It is very interesting ask what the property  $\tilde \mA_R$ has, and in particular its type. From the JT gravity viewpoint, the bulk spacetime only has one asymptotic boundary and it corresponds to a single-sided black hole with an EoW brane behind its horizon. By a naive expectation from AdS/CFT, the full boundary algebra should reconstruct everything in the bulk and those are all physical observables in the Hilbert space. Therefore, it is naturally conjectured that the boundary algebra in JT with an EoW brane is type I \cite{Kolchmeyer:2023gwa}. 

Surprisingly, this naive expectation is wrong, and we will show in Section \ref{sec:3.3} that this algebra $\tilde \mA_R$ is not type I but still type II$_1$, which means that it has a nontrivial commutant! In other words, there are physical observables that cannot be reconstructed by the right boundary only, at least before the triple scaling/JT limit. One fact that is important but implicit behind this statement is that we have different algebras with or without the EoW brane, but both of them act on the same Hilbert space $\mH$. Therefore, the full set of bounded operators $\mB(\mH)$ with the existence of an EoW brane is the same as the ordinary DSSYK, which is crucial to discuss the commutant. However, this is far from obvious in the JT gravity because different geometries, one-sided versus two-sided, are usually thought as dual to different states in different Hilbert spaces, which are spanned by acting local bulk fields on the that state \cite{Leutheusser:2021frk}. 

While \eqref{3.16} is a reasonable definition generalized from the ordinary DSSYK for the (right) boundary algebra of an EoW brane, there indeed exists an alternative definition of the boundary algebra by including more operators. This boundary algebra turns out to be equivalent to $\mB(\mH)$ as we show in Section \ref{sec:4.2}, and is thus type I$_\infty$. We will discuss its implication in Section \ref{sec:4.5}.

\subsection{EoW brane as a $q$-coherent state: alternative bulk slicing} \label{sec:3.2}

Representing an EoW brane in terms of a deformation of Hamiltonian has a direct connection to JT gravity. However, it is not obvious how to write the deformation of the Hamiltonian $H_R\ra\tilde H_R$ in the microscopic fermion basis $\psi_i$ of an SYK model. To have a better connection to the original SYK model, we will take a different but
equivalent viewpoint for the EoW brane in this section. It is proposed by Okuyama \cite{Okuyama:2023byh} that we can understand the EoW brane in the zero-matter subspace as a $q$-coherent state
\begin{equation}
\ket{\mB_{\a}}=\sum_{n=0}^{\infty}\f{\a^{n}}{[n]_{q}!}\ket{0^n}=\sum_{n=0}^{\infty}\f{(1-q)^{n}\a^{n}a_{R,0}^{\dag n}}{(q;q)_{n}}\ket{\w}\label{eq:12}
\end{equation}
where we use a short notation of binary string 
\be 
\ket{\ldots 0^n \ldots}=|\ldots \underbrace{0\cdots 0}_{n}\ldots\rangle,\quad \ket{\ldots 1^n \ldots}=|\ldots \underbrace{1\cdots 1}_{n}\ldots\rangle,\quad \ket{0^0}\equiv\ket{\w}
\ee
and formulate it in terms of right operators. The norm and inner product of $\ket{\mB_{\a}}$
states are
\begin{equation}
\avg{\mB_{\a}|\mB_{\b}}=\sum_{n=0}^{\infty}\f{(\a^{*}\b)^{n}}{[n]_{q}!}=\f 1{((1-q)\a^{*}\b;q)_{\infty}}\label{eq:10}
\end{equation}
which is finite for $|\a|,|\b|<1/\sqrt{1-q}$. All $\ket{\mB_{\a}}$
with complex $\a$ and $|\a|<1/\sqrt{1-q}$ form an overcomplete basis for the zero-matter
subspace $\mH_0$. To see this, we can use residue theorem to write $\ket{n}=\ket{0^n}$
in terms of a contour integral of $\a$ 
\begin{equation}
\ket{0^n}=\f{[n]_{q}!}{2\pi\i}\oint_{0}d\a\a^{-n-1}\ket{\mB_{\a}}\label{eq:11}
\end{equation}
which holds for any contour of $\a$ around the origin with $|\a|<1/\sqrt{1-q}$ and also shows that $\ket{\mB_{\a}}$ are overcomplete. Similar to the coherent state in quantum mechanis, $\ket{\mB_{\a}}$ is the $q$-coherent state for $a_{R,0}$ and a
vacuum for $a_{R,1}$
\begin{equation}
a_{R,0}\ket{\mB_{\a}}=\a\ket{\mB_{\a}},\quad a_{R,1}\ket{\mB_{\a}}=0
\end{equation}

Now we show why the $q$-coherent state $\ket{\mB_\mu}$ (with $\a=\mu$) is an equivalent description of an EoW brane. Let us define the following deformed annihilation operator
\be 
\tilde a _{R,0} =a _{R,0}+\mu q^{N_0}r^{N_1},\quad \tilde a_{R,1} =a_{R,1} \label{3.22}
\ee
which are generalized from the zero-matter subspace version considered in \cite{Okuyama:2023byh}.\footnote{In \cite{Okuyama:2023byh}, the deformation was made for the creation operator, which can also be generalized here, but we choose to deform the annihilation operator for later convenience.} Even though $\tilde{a}_{R,i}$ is not the hermitian conjugate
of $a^\dag_{R,i}$, this tilde operator surprisingly obeys exactly the same generalized $q$-harmonic oscillator commutation relation \eqref{2.19} with replacement $a_{R,i}\ra \tilde a_{R,i}$, namely
\begin{equation}
\tilde a_{R,i}a_{R,j}^{\dag}-Q_{ij}a_{R,j}^{\dag}\tilde a_{R,i}=\d_{ij}\label{eq:6}
\end{equation}
Indeed, the only nontrivial piece is for $i=0$, where we have 
\begin{align}
 & \tilde a_{R,0} {a}_{R,i}^{\dag}\ket{i_{1}\cdots i_{k}}=\tilde a_{R,0}\ket{i_{1}\cdots i_{k}i}\nonumber \\
= & \sum_{j=1}^{k}\d_{i_{j},i}Q_{i,0}\prod_{s=j+1}^{k}Q_{i,i_{s}}\ket{i_{1}\cdots\slashed{i}_{j}\cdots i_{k}i}+\d_{i,0}\ket{i_{1}\cdots i_{k}}+\mu r^{N_{1}}q^{N_{0}}\ket{i_{1}\cdots i_{k}i}\\
 & a_{R,i}^{\dag}\tilde a_{R,0}\ket{i_{1}\cdots i_{k}}=a_{R,i}^{\dag}\left[\sum_{j=1}^{k}\d_{i_{j},i}\prod_{s=j+1}^{k}Q_{i,i_{s}}\ket{i_{1}\cdots\slashed{i_{j}}\cdots i_{k}}+\mu q^{N_0}r^{N_1}/Q_{i,0}\ket{i_1\cdots i_k }\right]\nonumber \\
= & \sum_{j=1}^{k}\d_{i_{j},i}\prod_{s=j+1}^{k}Q_{i,i_{s}}\ket{i_{1}\cdots\slashed{i_{j}}\cdots i_{k}i}+\mu q^{N_0}r^{N_1}/Q_{i,0}\ket{i_1\cdots i_k i}
\end{align}
where $N_{1}=i+\sum_{j=1}^{k}i_{j}$ and $N_{0}=k+1-N_{1}$. It is easy to see that these two
equations agree with (\ref{eq:6}).

\begin{figure}
	\centering
	\begin{subfigure}{0.3\textwidth}
    \centering
		\includegraphics[height=6cm]{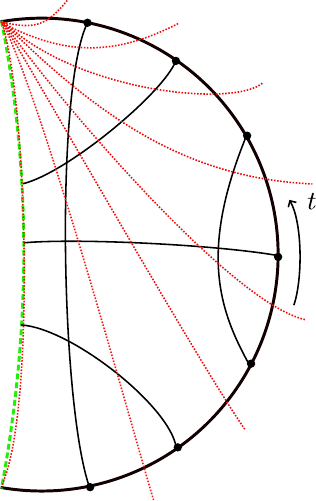}
		\caption{}
        \label{Fig:7a} 
	\end{subfigure}
        \hspace{5em}
	\begin{subfigure}{0.35\textwidth}
    \centering
		\includegraphics[height=6cm]{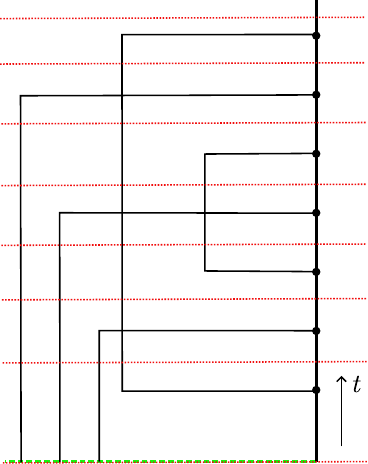}
        \caption{}
        \label{Fig:7b}
        \end{subfigure}
        \caption{(a) The bulk slicing choice for a chord diagram corresponding to the evolution from vacuum to the $q$-coherent state $\ket{\mB_\mu}$. (b) The diagram in (a) in the strip shape.}
\end{figure}

The equivalence of the $q$-coherent state $\ket{\mB_{\mu}}$ as an EoW brane is well reflected by the following polynomial identity
\begin{equation}
\avg{\w|\tilde{A}_{1}\cdots\tilde{A}_{k}|\w}=\avg{\w|A_{1}\cdots A_{k}|\mB_{\mu}}\label{eq:17-1}
\end{equation}
where $A_{j}\in\{a_{R,i}^{\dag},a_{R,i}\}$ and $\tilde{A}_{j}\in\{a_{R,i}^{\dag}, \tilde a_{R,i}\}$
and for each $j=1,\cdots,k$ we replace the operator in the LHS with
the corresponding untilde operator in the RHS. To show this identity,
for the RHS we can move the rightmost annihilation operator $a_{R,i}$ rightward using \eqref{2.19} over all creation operators to its right, and eventually it acts on $\ket{\mB_{\mu}}$ as
\begin{equation}
a_{R,0}\ket{\mB_{\mu}}=\mu \ket{\mB_{\mu}},\quad a_{R,1}\ket{\mB_{\mu}}=0\label{eq:18}
\end{equation}
Then we can repeat this process for all remaining annihilation operators $a_{R,i}$ for the RHS. For the LHS, we can apply the same calculus and eventually act the
tilde annihilation operators $\tilde{a}_{R,i}$ on $\ket{\w}$
as
\begin{equation}
\tilde{a}_{R,0}\ket{\w}=\mu q^{N_{0}}r^{N_{1}}\ket{\w}=\mu\ket{\w},\quad \tilde a_{R,1}\ket{\w}=0
\end{equation}
which is exactly the same as (\ref{eq:18}). As we have just showed that the
commutation relation for the tilde and untilde operators are exactly
the same, this implies that the contributions in each step of this
calculus are the same for both sides of (\ref{eq:17-1}). Since $H_{R,i}$ and $\tilde H_{R,i}$ are just the sum of the untilde/tilde annihilation and creation operators, a special case of the polynomial identity is
\begin{equation}
\avg{\w|\tilde H_{R,i_1}\cdots\tilde H_{R,i_k}|\w}=\avg{\w|H_{R,i_1}\cdots H_{R,i_k}|\mB_{\mu}}\label{eq:17-2}
\end{equation}

The physical interpretation is clear. The LHS is a general amplitude in Section \ref{sec:3.1} if we understand the EoW brane as a deformed Hamiltonian. The RHS is a special amplitude from $\ket{\mB_\mu}$ to vacuum if we evolve the state using the original right-side Hamiltonian from the ordinary DSSYK. The equality shows that we can equivalently understand the EoW brane as a special state, rather than a deformed operator, in the full Hilbert space. 

Indeed, this identity reflects two different bulk slicings for every chord diagram with EoW brane. In Figure \ref{fig:6a}, we have introduced a slicing corresponding to a deformed Hamiltonian. On the other hand, we can change the slicing to the one in Figure~\ref{Fig:7a}, where the initial slice $t=0$ coincides with the EoW brane and the final slice along $t$ is a point in the top. As usual, we require that all the chords intersecting a time slice should not cross before they reach the slice. Since the time slicing never ends on the EoW brane, the evolution should correspond to the original right-side Hamiltonian $H_{R,0}$. After deforming the diagram into a strip shape in Figure \ref{Fig:7b}, the EoW brane can be readily regarded as a projection to the state $\ket{\mB_\mu}$ at the beginning.\footnote{These two equivalent pictures between the deformed Hamiltonian and the $q$-coherent state, and their chord diagram interpretation, were also recently discussed in \cite{Berkooz:2025ydg} in a different context (the chords ending on the EoW brane are called reservoir chords there).}

The infinite sum in (\ref{eq:12}) defines an operator that generates
the state $\ket{\mB_\a}$ from $\ket{\w}$
\be 
\ket{\mB_\a}=\f 1{((1-q)\a a_{R,0}^{\dag};q)_{\infty}}\ket{\w}
\ee
However, this operator is
not hermitian and we will find a hermitian alternative operator using $H_{R,0}$, the
right Hamiltonian \eqref{2.15} in DSSYK. As we know from \eqref{2.14}, $H_{R,0}$ is
a bounded operator with the following spectrum representation in the zero-matter subspace
\begin{equation}
H_{R,0}=\int\f{d\t}{2\pi}\r(\t)E(\t)\ket{\t}\bra{\t}
\end{equation}
Using the expansion \eqref{eq:12}, the state $\ket{\mB_{\a}}$ can be rewritten in $\ket{\t}$ basis as \cite{Okuyama:2023byh}
\begin{align}
\ket{\mB_{\a}}=&\int\f{d\t}{2\pi}\r(\t)\ket{\t}\avg{\t|\mB_{\a}}=\int\f{d\t}{2\pi}\r(\t)\ket{\t}\sum_n\f{(\a\sqrt{1-q})^n}{(q;q)_n}H_n(\cos\t|q)\nn\\
=&\int\f{d\t}{2\pi}\f{\r(\t)}{(\a(1-q)^{1/2}e^{\pm\i\t};q)_{\infty}}\ket{\t}
\end{align}
where we used \eqref{eq:a186} to complete the sum. Note that the ground state has a simple expansion in $\ket{\t}$ basis
\begin{equation}
\ket{\w}=\int\f{d\t}{2\pi}\r(\t)\ket{\t}
\end{equation}
It is easy to see the following hermitian operator generating $\ket{\mB_{\a}}$
from $\ket{\w}$
\begin{equation}
\mB_{R,\a}=\f 1{(\a(1-q)^{1/2}e^{\pm\i\t(H_{R,0})};q)_{\infty}},\quad \t(x)\equiv\arccos[\sqrt{1-q}x/2]\label{eq:17-5}
\end{equation}

Note that this operator $\mB_{R,\a}$ is defined for the full Hilbert space $\mH$ which $H_{R,0}$ acts on, even though the construction is motivated from the zero-matter subspace to create the $q$-coherent state from the vacuum. Nevertheless, we still need to justify that the hermitian operator $\mB_{R,\a}$ is well defined when acting on any state in the full Hilbert space $\mH$, especially on multi-matter states. Indeed, as we will show in Section \ref{sec:3.4}, the full spectrum of $H_{R,0}$ can be solved exactly and it is bounded in
the range $[-2/\sqrt{1-q},2/\sqrt{1-q}]$ and can be parameterized as before with $\t\in[0,\pi]$. It follows that $\mB_{R,\a}$ is well-defined, bounded and invertible if 
\begin{equation}
|\a|(1-q)^{1/2}<1\label{eq:18-1}
\end{equation}
because $||e^{\pm \i \t (H_{R,0})}||=1$ and $|(a;q)_\infty|$ has a strictly positive lower bound and a finite upper bound when $|a|<1$. The continuous functional calculus of operators then guarantees the existence of inverse and square root of $\mB_{R,\a}$. Its inverse is simply 
\begin{equation}
\mB_{R,\a}^{-1}=(\a(1-q)^{1/2}e^{\pm\i\t(H_{R,0})};q)_{\infty}
\end{equation}
The condition (\ref{eq:18-1}) will be the scenario considered throughout this paper. 

Using the operator $\mB_{R,\mu}\in \mA_R$, there is another rewriting of the RHS of \eqref{eq:17-2} corresponding to a different bulk slicing
\begin{align} 
&\avg{\w|\tilde H_{R,i_1}\cdots\tilde H_{R,i_k}|\w}=\avg{\w|H_{R,i_1}\cdots H_{R,i_k}|\mB_{\mu}}=\avg{\w|H_{R,i_1}\cdots H_{R,i_k}\mB_{R,\mu}|\w}\nn\\
=&\avg{\w|\mB_{R,\mu}^{1/2}H_{R,i_1}\cdots H_{R,i_k}\mB_{R,\mu}^{1/2}|\w}\equiv \avg{\mB_{\mu}^{1/2}|H_{R,i_1}\cdots H_{R,i_k}|\mB_{\mu}^{1/2}} \label{3.37}
\end{align}
where in the second line we used the tracial property \eqref{2.40} of $\ket\w$ for $\mA_R$ to move $\mB_{R,\mu}^{1/2}$ from the ket to bra. This means that we can equivalently understand amplitude of  $\mA_R$ with an EoW brane as an expectation value in a half-EoW brane state $\ket{\mB_\mu^{1/2}}$. Since the Hamiltonian is still $H_{R,0}$, this corresponds to a bulk slicing shown in Figure~\ref{fig:new_slicing}. All bulk slices intercept the EoW brane at the single middle point, start from the lower half EoW brane, evolve anticlockwise, and end at the upper half EoW brane. The beginning and the end represent the projection onto $\ket{\mB_\mu^{1/2}}$ from both ket and bra states respectively.
\begin{figure}
	\centering
	\begin{subfigure}{0.3\textwidth}
    \centering
		\includegraphics[height=6cm]{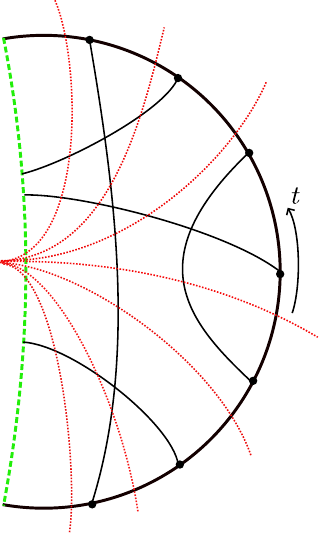}
		\caption{}
        \label{Fig:10a} 
	\end{subfigure}
        \hspace{5em}
	\begin{subfigure}{0.35\textwidth}
    \centering
		\includegraphics[height=6cm]{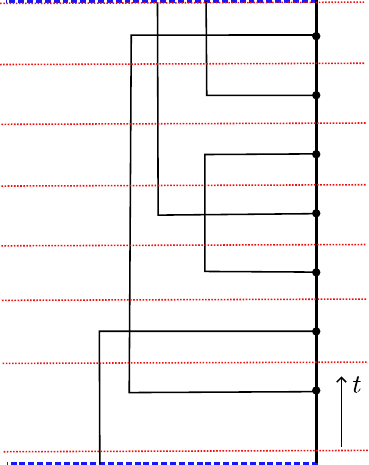}
        \caption{}
        \label{Fig:10b}
        \end{subfigure}
        \caption{(a) A new slicing scheme of the one-sided chord diagram, which can be interpreted as calculating the amplitude $\expval*{H_{R,i_1}H_{R,i_2}\ldots H_{R_,i_k}}{\mB_\mu^\f{1}{2}}$. (b) Initial and final states $|\mB_\mu^\f{1}{2}\rangle$ are denoted by blue dashed lines. Arbitrary numbers of chords are allowed to start from the bottom or end on the top, resulting in extra coefficients corresponding to the expansion of $|\mB_\mu^\f{1}{2}\rangle$ in the $\ket{n}$ basis.}
        \label{fig:new_slicing}
\end{figure}

\subsection{The von Neumann algebra of a single-sided black hole with an EoW brane is a type II$_1$ factor} \label{sec:3.3}

Let us discuss the von Neumann algebras $\mA_R$ and $\tilde \mA_R$ defined in \eqref{2.29} and \eqref{3.16} respectively. As we will see, the identity \eqref{eq:17-2} is crucial to understand the property of the von Neumann algebra $\tilde \mA_R$ of a single-sided black hole with an EoW brane. First of all, we can show the following property.
\begin{lemma} \label{lm1}
$\ket{\w}$ is a cyclic state for $\tilde \mA_R$. 
\end{lemma}
\begin{proof}
The proof is similar to the proof of $\ket{\w}$ being cyclic for $\mA_R$ in \cite{Xu:2024hoc}. The state $\ket{\w}$ being cyclic means $\tilde H_{R,i}$ generates a dense set of the Hilbert space when acting on $\ket{\w}$. As the binary string states are dense, we just need to generate all of them using $\tilde H_{R,i}$. Note that $\tilde H_{R,1}=H_{R,1}$ and $\tilde H_{R,0}=H_{R,0}+\mu q^{N_0}r^{N_1}$. Acting no more than $N_1$ $\tilde H_{R,1}$ and no more than $N_0$ $\tilde H_{R,0}$ with all possible ordering on $\ket{\w}$, we can generate states with at most $N_0$ 0-digit and $N_1$ 1-digit because $\mu q^{N_0}r^{N_1}$ does not add or kill any binary digit. Therefore, we can use a similar normal ordering process in \cite{Xu:2024hoc} to work out a linear combination of these operators to generate states with $k$ total number digits/chords from $k=1$ to $k=N_0+N_1$. Therefore, $\tilde \mA_R$ can generate all binary string basis from $\ket{\w}$.
\end{proof}

Next, we would like to define untilde as a map from $\tilde \mA_R$ to $\mA_R$. It is clear that it is well defined for all polynomials $p(\tilde H_{R,0},\tilde H_{R,1})$, which are simply mapped to another polynomial $p(H_{R,0},H_{R,1})$ in $\mA_R$. Reversely, tilde is also well defined for all polynomials $p(H_{R,0},H_{R,1})$, which are mapped to $p(\tilde H_{R,0},\tilde H_{R,1})$ in $\tilde \mA_R$. These two maps naturally extend to the full von Neumann algebra.
\begin{definition}
    $\td^{-1}$ is a map $\tilde \mA_R \ra \mA_R$. For each operator $\tilde\mO\in \tilde\mA_R$, it is a weak limit of a Cauchy sequence of polynomial
\be 
\tilde \mO=\lim^\text{w}_{n\ra \infty} f_n(\tilde H_{R,0}, \tilde H_{R,1}) \label{3.37-1}
\ee
where the superscript w means ``weak limit", and its image under $\td^{-1}$ is
\be 
\td^{-1}(\tilde \mO)=\mO\equiv \lim^\text{w}_{n\ra \infty} f_n(H_{R,0},H_{R,1}) \label{3.37-2}
\ee
Reversely, $\td$ is a map $\mA_R \ra \tilde \mA_R$. For each operator $\mO'\in \mA_R$, it is a weak limit of a Cauchy sequence of polynomial
\be 
\mO'=\lim^\text{w}_{n\ra \infty} g_n(H_{R,0},H_{R,1}) \label{3.37-3}
\ee
and its image under $\td$ is
\be 
\td(\mO')=\tilde \mO'\equiv \lim^\text{w}_{n\ra \infty} g_n(\tilde H_{R,0},\tilde H_{R,1}) \label{3.37-4}
\ee
\end{definition}

Let us check that this definition makes sense. The von Neumann algebra is defined in terms of weak operator topology, which, roughly speaking, is by the matrix elements. Since $\ket{\w}$ is cyclic for $\tilde \mA_R$, all polynomials $p(\tilde H_{R,i})$ acting on $\ket{\w}$ form a dense set of states in $\mH$. Therefore, for any $\tilde \mO\in \tilde \mA_R$ in the form of \eqref{3.37-1}, its matrix elements of a dense set of states are
\begin{align} 
&\avg{p_1(\tilde H_{R,i})|\tilde \mO|p_2(\tilde H_{R,i})}\equiv \lim_{n\ra\infty}\avg{p_1(\tilde H_{R,i})|f_n(\tilde H_{R,i})|p_2(\tilde H_{R,i})} \nn\\
=&\lim_{n\ra\infty}\avg{\w|p_1(\tilde H_{R,i})^\dag f_n(\tilde H_{R,i}) p_2(\tilde H_{R,i})|\w}=\lim_{n\ra\infty}\avg{\w|\mB_{R,\mu} p_1(H_{R,i})^\dag f_n(H_{R,i}) p_2(H_{R,i})|\w}\nn\\
=&\lim_{n\ra\infty}\avg{p_1(H_{R,i})\mB_{\mu} | f_n(H_{R,i}) |p_2(H_{R,i})}=\avg{p_1(H_{R,i})\mB_{\mu} | \mO |p_2(H_{R,i})} \label{3.41}
\end{align}
where in the second line we used the identity \eqref{eq:17-2} and the cyclic property \eqref{2.40} of $\ket{\w}$ for $\mA_R$ to move $\avg{\w|\cdots \mB_{R,\mu}|\w}\ra \avg{\w|\mB_{R,\mu}\cdots |\w}$. Since $\mB_{R,\mu}\in \mA_R$ is invertible and $\ket{\w}$ is cyclic for $\mA_R$, $\ket{p_1(H_{R,i})\mB_{\mu}}$ is also a dense set in $\mH$ because we can use a polynomial of $H_{R,0}$ to approximate $\mB^{-1}_{R,\mu}$ as a sequence limit.\footnote{More precisely, we need to show that there is a sequence of polynomial $q_n(H_{R,0})$ such that $||(q_n(H_{R,0})\mB_{R,\mu}-1)\ket{\w}||\ra0$ as $n\ra \infty$. Since $\mB_{R,\mu}$ is a function of $H_{R,0}$ only, we can expand it in terms of the spectrum of $H_{R,0}$
\be 
\mB_{R,\mu}=\int_{0}^\pi \f {d \t}{(\mu(1-q)^{1/2}e^{\pm\i\t};q)_{\infty}} P_{\t} \label{3.42}
\ee
where $P_\t$ is the projection to the eigen space of $H_{R,0}$ with eigenvalue $E(\t)$ (we will show this in Section \ref{sec:3.4}). Expanding $q_n(H_{R,0})$ on the same spectrum, we can choose $q_n$ such that for arbitrarily small $\e>0$, there exists a $N_\e$ with $|q_n(E(\t))/(\a(1-q)^{1/2}e^{\pm\i\t};q)_{\infty}-1|<\e$ for $n>N_\e$ because the integrand in \eqref{3.42} is finite for $\t\in[0,\pi]$ and its reciprocal can expanded as an absolute convergent series of $E(\t)$. It follows that $||(q_n(H_{R,0})\mB_{R,\mu}-1)\ket{\w}||$ is upper bounded by $O(\e)$.
}
The last line of \eqref{3.41} gives convergent value of the matrix elements of $\mO$ in \eqref{3.37-2} on a dense set of states, given $\tilde \mO\in \tilde \mA_R$ is well defined. Therefore, for each $\tilde \mO\in \tilde \mA_R$, the image $\td^{-1}(\tilde \mO)$ exists in $\mA_R$.\footnote{Rigorously speaking, we need to show this for all vectors in $\mH$ not just the dense set. This is straightforward by replacing $p_1$ and $p_2$ with a sequence of polynomials $p_{1,k}$ and $p_{2,s}$ respectively. For the two vectors $\ket{\tilde x}=\lim_{k\ra \infty}p_{1,k}(\tilde H_{R,i})\ket{\w}$ and $\ket{\tilde y}=\lim_{s\ra \infty}p_{2,s}(\tilde H_{R,i})\ket{\w}$, taking $k$ and $s$ limit inside the $n$ limit in \eqref{3.41} leads to the same identity
\be 
\avg{\tilde x|\tilde \mO|\tilde y}=\avg{x\mB_{\mu} | \mO |y},\quad \ket{x\mB_{\mu}}\equiv \lim_{k\ra \infty}p_{1,k}(H_{R,i})\mB_{R,\mu}\ket{\w},\quad \ket{y}\equiv \lim_{s\ra \infty}p_{2,s}(H_{R,i})\ket{\w}
\ee
where we can easily show that the vectors $\ket{x\mB_{\mu}},\ket{y}$ exist because the following Cauchy sequence holds
\begin{align} 
||(p_{1,k_1}(H_{R,i})-p_{1,k_2}(H_{R,i}))\mB_{R,\mu}\ket{\w}||&\leq ||\mB_{R,\mu}||\avg{\w|\mB_{R,\mu}|p_{1,k_1}(H_{R,i})-p_{1,k_2}(H_{R,i})|^2|\w} \nn\\
&=||\mB_{R,\mu}||\avg{\w||p_{1,k_1}(\tilde H_{R,i})-p_{1,k_2}(\tilde H_{R,i})|^2|\w}<O(\e)
\end{align}
for large enough $k_1,k_2>N_\e$, given $p_{1,k}(\tilde H_{R,i})\ket{\w}$ is a Cauchy sequence converging to $\ket{\tilde x}$. Replacing $\mB_{R,\mu}$ with identity justifies $\ket{y}$. In other places in this subsection, we can repeat similar arguments to extend the result from a dense set to the whole Hilbert space. Therefore, we will only consider dense sets in the main text for simplicity.
}
 Following the equality of \eqref{3.41} reversely, we can also show that the image $\td(\mO)$ exists in $\tilde \mA_R$ for $\forall \mO\in\mA_R$. 

The maps $\td$ and $\td^{-1}$ are also defined without ambiguity. Suppose there are two Cauchy sequence $f_n(\tilde H_{R,i})$ and $g_n(\tilde H_{R,i})$ converge to the same operator $\tilde \mO\in\tilde \mA_R$. Taking the difference of their matrix elements in the form of \eqref{3.41}, we have 
\be 
0=\lim_{n\ra\infty}\avg{p_1(H_{R,i})\mB_{\mu} | f_n(H_{R,i})-g_n(H_{R,i}) |p_2(H_{R,i})} \label{3.42-1}
\ee
Since the matrix elements of $f_n(H_{R,i})$ and $g_n(H_{R,i})$ on a dense set of states both converge to the same value, this means that they define the same operator $\mO\in\mA_R$ in weak operator topology. Therefore, the image of $\td^{-1}$ is unambiguous 
\be 
\td^{-1}(\tilde \mO)=\mO=\lim^{\text w}_{n\ra \infty}f_n(H_{R,i})=\lim^{\text w}_{n\ra \infty}g_n(H_{R,i})
\ee
Following this argument reversely, we can also show that $\td$ is unambiguously defined. 

We can also use \eqref{3.41} to show that the kernel of $\td^{-1}$ and $\td$ are both trivial because $\tilde\mO=0 \iff \tilde \mO=0$ in the sense of matrix elements of a dense set of states. These two maps are inverse to each other, and also commutative with operator product, sum and hermitian conjugate, which leads to
\begin{lemma}\label{isomorphism}
    $\td$ and $\td^{-1}$ are two isomorphisms between $\mA_R$ and $\tilde \mA_R$, and they are inverse to each other.
\end{lemma}

Since $\ket{\w}$ is a separating state for $\mA_R$, which means the only operator in $\mA_R$ annihilating $\ket{\w}$ is zero, we can use \eqref{3.41} to show 
\begin{lemma} \label{lm3}
    $\ket{\w}$ is a separating state for $\tilde \mA_R$.
\end{lemma}
\begin{proof}
    Suppose $\exists \tilde \mO\in \tilde \mA_R$ is nontrivial and $\tilde\mO \ket{\w}=0$. As the limit of a sequence $f_n(\tilde H_{R,i})$, take its matrix elements in \eqref{3.41} with $\ket{p_2(\tilde H_{R,i})}=\ket{\w}$. It follows that $\mO=\lim^{\text w}_{n\ra\infty}f_n(H_{R,i})$ annihilates $\ket{\w}$. Since $\ket{\w}$ is separating for $\mA_R$, $\mO$ must be trivial and the polynomial function $f_n$ must converge to zero, which leads to a contradiction.
\end{proof}

Another immediate consequence of \eqref{3.41} is that $\ket{\w}$ is
not a tracial state for $\tilde{\mA}_{R}$ because it does not obey
the cyclic property \eqref{2.40} as a trace. Consider two generic operators $\tilde{\mO}_{1},\tilde{\mO}_{2}\in\tilde{\mA}_{R}$.
We will have 
\begin{align}
&\avg{\w|\tilde{\mO}_{1}\tilde{\mO}_{2}|\w}=\avg{\w|\mB_{R,\mu}\mO_{1}\mO_{2}|\w}\nn\\
=&\avg{\w|\mB_{R,\mu}\mB_{R,\mu}^{-1}\mO_{2}\mB_{R,\mu}\mO_{1}|\w}=\avg{\w|\tilde{\mB}_{R,\mu}^{-1}\tilde{\mO}_{2}\tilde{\mB}_{R,\mu}\tilde{\mO}_{1}|\w}\label{eq:26}
\end{align}
where in the second step we used the fact that $\ket{\w}$ is tracial
for $\mA_{R}$ and $\mB_{R,\mu}\in\mA_{R}$ to move $\mO_{2}$ to the
left, in the third step we used \eqref{3.41} and define the tilde version of $\mB_{R,\mu}$ as 
\begin{equation}
\tilde{\mB}_{R,\mu}=\f 1{(\mu(1-q)^{1/2}e^{\pm\i\t(\tilde{H}_{R,0})};q)_{\infty}}
\end{equation}
As we will show in the next section that the full spectrum of $\tilde{H}_{R,0}$
is the same as $H_{R,0}$, and thus $\tilde{\mB}_{R,\mu}$ is also an invertible bounded operator
given (\ref{eq:18-1}). Since a generic $\tilde{\mO}_{2}$ does not
commute with $\tilde{\mB}_{R,\mu}$, (\ref{eq:26}) fails to obey the
cyclic property for a tracial state. 

Nevertheless, we can follow (\ref{eq:26})
to find the tracial state $\ket{\psi}$ for $\tilde{\mA}_{R}$. Since
$\tilde{\mA}_{R}\ket{\w}$ generates a dense set of states,
we can consider a generic state as 
\begin{equation}
\ket{\psi}=\tilde{V}\ket{\w},\quad\tilde{V}\in\tilde{\mA}_{R}
\end{equation}
It follows from (\ref{eq:26}) that
\begin{align}
\avg{\psi|\tilde{\mO}_{1}\tilde{\mO}_{2}|\psi} & =\avg{\w|\tilde{V}^{\dag}\tilde{\mO}_{1}\tilde{\mO}_{2}\tilde{V}|\w}=\avg{\w|\mB_{R,\mu}\mB_{R,\mu}^{-1/2}\mO_{2}V\mB_{R,\mu}V^{\dag}\mO_{1}\mB_{R,\mu}^{-1/2}|\w}\nonumber \\
 & =\avg{\w|\tilde{\mB}_{R,\mu}^{-1/2}\tilde{\mO}_{2}\tilde{V}\tilde{\mB}_{R,\mu}\tilde{V}^{\dag}\tilde{\mO}_{1}\tilde{\mB}_{R,\mu}^{-1/2}|\w}
\end{align}
from which we can quickly find that if we choose $\tilde{V}=\tilde{V}^{\dag}=\tilde{\mB}_{R,\mu}^{-1/2}$
the cyclic property is satisfied
\begin{equation}
\avg{\psi|\tilde{\mO}_{1}\tilde{\mO}_{2}|\psi}=\avg{\psi|\tilde{\mO}_{2}\tilde{\mO}_{1}|\psi}
\end{equation}

It also has a good property that any operator $\tilde{\mO}\in\tilde{\mA}_{R}$ has the same expectation value as the corresponding $\mO\in\mA_{R}$ under isomorphism $\td^{-1}$ in $\ket{\w}$
\begin{equation}
\avg{\psi|\tilde{\mO}|\psi}=\avg{\w|\tilde{\mB}_{R,\mu}^{-1/2}\tilde{\mO}\tilde{\mB}_{R,\mu}^{-1/2}|\w}=\avg{\w|\mB_{R,\mu}\mB_{R,\mu}^{-1/2}\mO\mB_{R,\mu}^{-1/2}|\w}=\avg{\w|\mO|\w} \label{3.49}
\end{equation}
 Physically, the state $\ket{\psi}$ behaves just like a vacuum state without an EOW brane for $\tilde \mA_R$. This means that $\td$ also induces an isometry of Hilbert space $\mH$.
%\begin{lemma}
   %We generalize $\td$ to a map $\mH\ra \mH$ by $\td:\ket{\w}\mapsto \ket{\psi}$ and $\td:\mO\ket{\w}\mapsto \tilde \mO \ket{\psi}$. It is an isometry of $\mH$. So is the inverse $\td^{-1}$.
%\end{lemma}
%\begin{remark}
    %Indeed, one can also construct $\td$ and $\td^{-1}$ in a reverse order. First define $\td$ as a map in the Hilbert space for a dense set of states $p(H_{R,i})\ket{\w}$ generated by polynomials of $H_{R,i}$. Since it is an isometry between this dense set and another dense set $p(\tilde H_{R,i}) \ket{\psi}$, one can show that it naturally extends to the whole Hibert space. Then it induces the isomorphism on algebras $\mA_R\ra \tilde \mA_R$.
%\end{remark}
\begin{lemma} \label{lm4}
    $\td$ induces a unitary transformation $U:
    \mathcal H\rightarrow \mathcal H$ defined by $Up(H_{R,i})\mB_{R,\mu}^{\f{1}{2}}\ket{\w}=p(\tilde H_{R,i})\ket{\w}$ for arbitrary polynomials $p$. Furthermore, we have $\tilde{\mA}_R=U\mA_R U^\dagger$.
\end{lemma}
\begin{proof}
    Notice that $\text{span}\{p(\tilde H_{R,i})\ket\w\}$ and $\text{span}\{p(H_{R,i})\mB^{\f{1}{2}}_{R,\mu}\ket\w\}$ are dense in $\mH$. The latter comes from the fact that for an algebra $\mA$ and a state $\ket{\varphi}$ we have $\overline{\mA\ket{\varphi}}=\overline{\mA''\ket{\varphi}}$. It is easy to see that $U$ is a bijective isometry from previous discussions, and is  therefore a unitary between two dense subsets of $\mH$, which can be extended to a unitary operator on $\mH$. It is then easy to verify that for an arbitrary polynomial $f$ we have $f(\tilde H_{R,i})=Uf(H_{R,i})U^\dagger.$\footnote{Strictly speaking, we prove this identity first on an dense subset of $\mH$ and use the fact that $f(H_{R,i})$, $f(\tilde H_{R,i})$ are both bounded operators to extend the identity to $\mH$.} Since taking weak limits commutes with unitary conjugation we immediately have $\tilde{\mA}_R=U\mA_R U^\dagger$.
\end{proof}
\begin{remark} 
    This lemma enhances the conclusion of Lemma~\ref{isomorphism} as it promotes the isomorphism to a unitary equivalence. Note that with $U$ we can write $\ket{\psi}=U\ket\w$. By definition of $U$ we can write this as $U\ket\w=U\mB_{R,\mu}^{-\f{1}{2}}\mB_{R,\mu}^{\f{1}{2}}\ket\w=\tilde{\mB}_{R,\mu}^{-\f{1}{2}}\ket\w$, which agrees with the choice of $\tilde V$ above.
\end{remark}
\noindent Since $\ket{\w}$ is the trace for the type II$_1$ algebra $\mA_R$, $\avg{\w|\cdot|\w}$ is a positive, faithful, normal and semifinite linear map for $\mA_R$. These properties are inherited by $\ket{\psi}$ because $\td$ is an isometry of $\mH$ induced by an isomorphism between $\tilde \mA_R$ and $\mA_R$.\footnote{One can follow the proof in \cite{Xu:2024hoc} to show this in a similar manner. Note that the spectrum of $\tilde H_{R,0}$ is bounded and the same as $H_{R,0}$, as we will show in Section \ref{sec:3.4}.} Since $\mA_R$ is a factor, it immediately follows that
\begin{theorem} \label{th4}
    $\tilde{\mA}_{R}$ is a type II$_{1}$ von Neumann factor with $\ket{\psi}=\tilde{\mB}_{R,\mu}^{-1/2}\ket{\w}$
as the unique tracial state up to normalization. 
\end{theorem}
\begin{proof}
    It is clear that $\tilde \mA_R$ is type II$_1$ and $\ket{\psi}$ defines a trace. Here we show it is a factor given $\mA_R$ is a factor. Suppose there is an operator $\tilde\mO\in \tilde \mA_R \cap \tilde \mA_R'$. It commutes with all operators $\tilde a\in \tilde \mA_R$, which by \eqref{3.41} and in terms of matrix elements means
    \be 
0=\avg{p_1(\tilde H_{R,i})|[\tilde a, \tilde \mO]|p_2(\tilde H_{R,i})}=\avg{p_1( H_{R,i})\mB_{\mu}|[a,\mO]|p_1(H_{R,i})}
    \ee
    Since $\td$ is an isomorphism, this means that $\mO$ commutes with all operators in $\mA_R$ and thus must be proportional to identity because $\mA_R$ is a factor. It follows that $\tilde\mO$ must be proportional to identity, and $\tilde \mA_R$ is a factor.
\end{proof}

In our view, this theorem is physically surprising because the algebra $\tilde \mA_R$ can be understood as a UV complete version of a single-sided black hole with an EoW brane behind the horizon. From an IR or JT gravity viewpoint, since there is only one asymptotic boundary, the algebra on the boundary should reconstruct all observables in the system and should be type I. The most significant consequence of the difference between these two types of algebras is as follows. If it were type I, the Hilbert space of a single-sided black hole would be factorized and we could isolate ourselves to this Hilbert space and there would be no nontrivial commutant for this algebra. In other words, there were no physical observables that cannot be reconstructed from the boundary. However, being a type II$_1$ algebra means that the Hilbert space of a single-sided black hole (at least a UV complete version) does not factorize and cannot be treated as an isolated system. Instead, it has infinite amount of entanglement with a counterpart, on which there are physical observables that cannot be reconstructed from the single-sided boundary. These unexpected physical observables constitute the commutant of the single-sided algebra $\tilde \mA_R$.

It is definitely crucial to understand the commutant of $\tilde \mA_R$ and uncover the secrets granted by the UV data of quantum gravity, if we regard DSSYK as a toy model for the latter. A very powerful tool to discuss the commutant is the Tomita-Takesaki theory \cite{takesaki2006tomita}. In Theorem \ref{th4}, we have shown that $\ket{\psi}$ is a tracial state for $\tilde \mA_R$. Moreover, it is cyclic and separating for $\tilde \mA_R$, because $\tilde \mB _{R,\mu} ^{-1/2}$ is invertible and $\ket{\w}$ is cyclic and separating for $\tilde \mA_R$ as shown in Lemma \ref{lm1} and \ref{lm3}.

Since $\ket{\psi}$ is cyclic and separating for $\tilde \mA_R$, by Tomita-Takesaki theory, we will have an antilinear Tomita operator $\tilde S_\psi$ for $\tilde{\mA}_{R}$
with respect to $\ket{\psi}$
\begin{equation}
\tilde S_{\psi}\tilde{\mO}_{R}\ket{\psi}=\tilde{\mO}_{R}^{\dag}\ket{\psi}\label{eq:31}
\end{equation}
As the generators $\tilde{H}_{R,i}$ of $\tilde{\mA}_{R}$ are hermitian
operators, if $\tilde{\mO}_{R}=\tilde{H}_{R,i_{1}}\cdots\tilde{H}_{R,i_{k}}$
is a string of generators, the act of $\tilde S_{\psi}$ is simply reverse
the ordering of $i_{1}\cdots i_{k}$. More precisely, we can define a new binary string basis of the Hilbert space by acting $\tilde{H}_{R,i}$ on $\ket{\psi}$
\begin{equation}
\kit{i_{1}\cdots i_{k}}=\tilde{H}_{R,i_{k}}\cdots\tilde{H}_{R,i_{1}}\ket{\psi},\quad \kit{0^0}\equiv \ket{\psi}
\end{equation}
and $\tilde S_\psi$ acts as
\be
\tilde S_\psi\kit{i_{1}\cdots i_{k}}=\kit{i_{k}\cdots i_{1}},\quad \tilde S_\psi\ket{\psi}=\ket{\psi}
\ee
In this basis, the action of operators $\tilde{H}_{R,i}$ is simple
\begin{equation}
\tilde{H}_{R,i}\kit{i_{1}\cdots i_{k}}=\kit{i_{1}\cdots i_{k}i}
\end{equation}
The bra states are defined canonically as
\be 
\bri{i_1\cdots i_k}=\kit{i_1\cdots i_k}^\dag= \bra{\psi} \tilde{H}_{R,i_{1}}\cdots\tilde{H}_{R,i_{k}}
\ee

The Tomita operator $\tilde S_\psi$  has a unique polar decomposition
of $S_{\psi}$ 
\begin{equation}
\tilde S_{\psi}=\tilde J_{\psi}\tilde \D_{\psi}^{1/2},\quad \tilde J^2_\psi=1,\quad \tilde \D_\psi=\tilde S_\psi^\dag \tilde S_\psi \label{3.60}
\end{equation}
By above definition, we can show the matrix elements of $\tilde S^\dag_\psi \tilde S_\psi$ as
\begin{align} 
&\avvg{i_1\cdots i_k|\tilde S^\dag_\psi \tilde S_\psi|j_1\cdots j_s}= \avvg{i_1\cdots i_k|\tilde S^\dag_\psi |j_s\cdots j_1}=\avvg{j_s\cdots j_1|\tilde S_\psi|i_1\cdots i_k} \nn\\
=&\avvg{j_s\cdots j_1|i_k\cdots i_1}=\avg{\psi|\tilde{H}_{R,j_{s}}\cdots\tilde{H}_{R,j_{1}}\tilde{H}_{R,i_{1}}\cdots\tilde{H}_{R,i_{k}}|\psi} \nn\\
=&\avg{\psi|\tilde{H}_{R,i_{1}}\cdots\tilde{H}_{R,i_{k}}\tilde{H}_{R,j_{s}}\cdots\tilde{H}_{R,j_{1}}|\psi}=\avvg{i_1\cdots i_k|j_1\cdots j_s} \label{3.57}
\end{align}
where in the third line we used the tracial property of $\ket{\psi}$ to switch the order. Since  $\kit{i_{1}\cdots i_{k}}$ are dense in $\mH$, it follows that 
\be 
\tilde \D_\psi=\I,\quad \tilde S_\psi=\tilde J_\psi
\ee

As is well-known in Tomita-Takesaki theory, the commutant of $\tilde{\mA}_{R}$
is given by $\tilde J_{\psi}\tilde{A}_{R}\tilde J_{\psi}$, which allows us to formally
define a ``left-side" algebra 
\begin{equation}
\tilde{H}_{L,i}=\tilde J_{\psi}\tilde{H}_{R,i}\tilde J_{\psi},\quad\tilde{\mA}_{L}\equiv\{\tilde{H}_{L,i}\}''=\tilde{\mA}_{R}' \label{3.62}
\end{equation}
Since $\tilde J_{\psi}=\tilde J_{\psi}^{-1}=\tilde J_{\psi}^{\dag}$, $\tilde{H}_{L,i}$
are hermitian operators. Using the $\kit{i_{1}\cdots i_{k}}$ basis,
the tilde left operators act simply as
\begin{equation}
\tilde{H}_{L,i}\kit{i_{1}\cdots i_{k}}=\kit{ii_{1}\cdots i_{k}},\quad \tilde H_{L,i}\ket{\psi}=\tilde H_{R,i}\ket{\psi}\label{3.63}
\end{equation}
Like the isomorphism $\td$ between $\mA_R$ and $\tilde \mA_R$, it also extends to an isomorphism between $\mA_L$ and $\tilde \mA_L$. 
\begin{lemma}
    $\td$ can be extended to an isomorphism between $\mA_L$ to $\tilde\mA_L$.
\end{lemma}
\begin{proof}
    The extension is defined in an obvious way, $\td:H_{L,i}\mapsto \tilde H_{L,i}$. By definition, we have an identity for the left-side algebras
    \begin{align}
    &\avg{\psi|\tilde H_{L,i_1}\cdots\tilde H_{L,i_k}|\psi}=\avg{\psi|\tilde H_{R,i_k}\cdots\tilde H_{R,i_1}|\psi}\nn\\
    =&\avg{\w|H_{R,i_k}\cdots H_{R,i_1}|\w}=\avg{\w|H_{L,i_1}\cdots H_{L,i_k}|\w}\label{eq:17-2-2}
\end{align}
    Then all discussions about the isomorphism between $\mA_R$ and $\tilde \mA_R$ apply to the left-side counterpart.
\end{proof}

Though written in the new basis $\kit{i_{1}\cdots i_{k}}_{\psi}$,
the commutant generators $\tilde{H}_{L,i}$ look simple, it is quite
nontrivial to construct them using the generators $H_{L,i}$ and $H_{R,i}$ in the algebras $\mA_{L,R}$ of the two-sided DSSYK. The essential difficulty is that $\tilde J_{\psi}=\tilde S_{\psi}$ in (\ref{eq:31}) acts very nontrivially on the original binary string basis $\ket{i_{1}\cdots i_{k}}$
as it reverses the order of the string of $\tilde{H}_{R,i}$ before
$\tilde \mB_{R,\mu} ^{-1/2}$ though the latter is a functional of $\tilde{H}_{R,0}$. It is even far from obvious that we can use some simple operators, which in a sense have a simple geometric meaning, e.g. $q^{N_0} r^{N_1}$ in $\tilde H_{R,0}$, to construct them in an explicit form. Unlike the right tilde algebra generators $\tilde{H}_{R,0}=H_{R,0}+\mu r^{N_{1}}q^{N_{0}}$ and $\tilde{H}_{R,1}=H_{R,1}$ with clear geometric interpretation as we have discussed in Section \ref{sec:3.1}, the left tilde algebra $\tilde \mA_L$, being their commutant, does not have simple geometric meaning in the EoW brane chord diagram spacetime. This explains why it is unexpected from the pure geometric argument of JT gravity with an EOW brane. In order to understand the commutant $\tilde\mA_L$ better, we will revisit it later in Section \ref{sec:4.1} from a different perspective.

Before we end this section, we would like to briefly mention two more quantities characterizing the EoW brane. The first is the relative modular operator $\tilde \D_{\psi|\w}$. We define the antilinear relative Tomita operator $\tilde S_{\psi|\w}$ as 
\be 
\tilde S_{\psi|\w} \tilde \mO \ket{\psi}=\tilde \mO^\dag \ket{\w},\quad \forall \tilde\mO\in\tilde\mA_R
\ee
which has a similar polar decomposition
\begin{equation}
\tilde S_{\psi|\w}=\tilde J_{\psi|\w}\tilde \D_{\psi|\w}^{1/2},\quad \tilde J^2_{\psi|\w}=\I,\quad \tilde \D_{\psi|\w}=\tilde S_{\psi|\w}^\dag \tilde S_{\psi|\w}
\end{equation}
While the relative Tomita operator $\tilde S_{\psi|\w}$ does not have a simple form, the relative modular operator can be easily computed as follows.
\begin{lemma}
    The relative modular operator $\tilde\D_{\psi|\w}=\tilde{\mB}_{R,\mu}$.
\end{lemma}
\begin{proof}
    Similar to \eqref{3.57}, we have 
    \begin{align} 
&\avvg{i_1\cdots i_k|\tilde S^\dag_{\psi|\w} \tilde S_{\psi|\w}|j_1\cdots j_s}= \bri{i_1\cdots i_k}\tilde S^\dag_{\psi|\w} \tilde H_{R,j_1}\cdots \tilde H_{R,j_s}\ket{\w} \nn\\
=&\bra{\w}\tilde H_{R,j_s}\cdots \tilde H_{R,j_1}\tilde S_{\psi|\w}\kit{i_1\cdots i_k}=\avg{\w|\tilde H_{R,j_s}\cdots \tilde H_{R,j_1} \tilde H_{R,i_1}\cdots \tilde H_{R,i_k}|\w}  \nn\\
=&\avg{\w|\mB_{R,\mu} H_{R,j_s}\cdots H_{R,j_1}  H_{R,i_1}\cdots  H_{R,i_k}|\w} =\avg{\w|H_{R,i_1}\cdots  H_{R,i_k}\mB_{R,\mu} H_{R,j_s}\cdots H_{R,j_1}  |\w}  \nn\\
=&\avg{\psi|\tilde H_{R,i_1}\cdots  \tilde H_{R,i_k}\tilde\mB_{R,\mu}\tilde H_{R,j_s}\cdots \tilde H_{R,j_1}|\psi}=\avvg{i_1\cdots i_k|\tilde\mB_{R,\mu}|j_1\cdots j_s}
\end{align}
where in the third line we used \eqref{eq:17-2} and tracial property of $\ket{\w}$ for $\mA_R$, and in the last line we used isometry \eqref{3.49}.
\end{proof}

Using the relative modular operator, we can define the relative entropy and relative modular flow, which characterizes how different the EoW brane is from the vacuum in the view of $\tilde \mA_R$. We will leave this investigation as a future work.

Besides the state $\ket{\psi}$, we have introduced the half EoW brane state $\ket{\mB_\mu^{1/2}}$ in Section \ref{sec:3.2} to characterize the EoW brane in terms of the original right side algebra $\mA_R$. Its modular operator $\D_{\mB_\mu ^{1/2}}$ will be useful in Section \ref{sec:4.3}. The Tomita operator $S_{\mB_\mu ^{1/2}}$ is defined as
\be 
S_{\mB_\mu ^{1/2}}\mO \ket{\mB_\mu ^{1/2}}= \mO^\dag \ket{\mB_\mu ^{1/2}},\quad \forall \mO\in\mA_R
\ee
which has a similar polar decomposition to \eqref{3.60} and leads to
\begin{lemma} \label{lemma8}
    The modular operator $\D_{\mB_\mu ^{1/2}}=\mB_{R,\mu}\mB_{L,\mu}^{-1}$.
\end{lemma}
\begin{proof}
Checking matrix elements, we have 
    \begin{align} 
&\avg{i_1\cdots i_k|S^\dag_{\mB_\mu ^{1/2}} S_{\mB_\mu ^{1/2}}|j_1\cdots j_s}= \avg{i_1\cdots i_k|S^\dag_{\mB_\mu ^{1/2}} \mB_{R,\mu}^{-1/2} H_{R,j_1}\cdots H_{R,j_s}|\mB_\mu^{1/2}} \nn\\
=&\avg{\mB_\mu^{1/2}|H_{R,j_s}\cdots H_{R,j_1} \mB_{R,\mu}^{-1/2} S_{\mB_\mu ^{1/2}}|i_1\cdots i_k}=\avg{\mB_\mu^{1/2}|H_{R,j_s}\cdots H_{R,j_1} \mB_{R,\mu}^{-1} H_{R,i_1}\cdots H_{R,i_k}|\mB_\mu^{1/2}}  \nn\\
=&\avg{\w|\mB_{R,\mu}^{-1} H_{R,i_1}\cdots H_{R,i_k}\mB_{R,\mu}H_{R,j_s}\cdots H_{R,j_1} |\w}  =\avg{i_1\cdots i_k|\mB_{R,\mu} \mB_{L,\mu} ^{-1}|j_1\cdots j_s} \label{4.16}
\end{align}
where in the last line we used the tracial property \eqref{2.40} of $\ket{\w}$ and $H_{R,0}\ket{\w}=H_{L,0}\ket{\w}$.
\end{proof}

\subsection{EoW brane as a screening matter-brane: full spectrum of DSSYK} \label{sec:3.4}

From the JT limit of the EoW brane, whose trajectory is a geodesic in the AdS$_2$ spacetime \cite{Gao:2021uro}, we may naturally regard the EoW brane as a kind of heavy particle. In this section, we will make this idea explicit in DSSYK and construct a family of states called matter-brane states, which are generalizations of the $q$-coherent state with a fixed number of matter chords dressed with gravity chords in a coherent manner. These states have an equivalent description as EoW branes if they are probed by the gravitational Hamiltonian $H_{R,0}$, but can be distinguished by probes of matter operators $H_{R,1}$. In particular, We can regard the $q$-coherent state of the EOW brane as a spacial case of a matter-brane state with zero matter chord. 

A significant byproduct of the construction of matter-brane states is that  we can use this formalism to solve the full spectrum of $\tilde H_{R,0}$ in $\mH$ with arbitrary numbers of matter chords. It follows that this construction can be blindly applied to DSSYK and gives the explicit spectrum of the latter. It has been attempted in \cite{Lin:2023trc,Xu:2024hoc} to solve the spectrum of DSSYK in one- and two-matter subspaces, but their methods intrinsically treat states as two-sided, and become quite involved in higher matter subspaces, where no explicit construction was written down.\footnote{Along this two-sided treatment, some recent progress was made in \cite{Aguilar-Gutierrez:2025pqp,Aguilar-Gutierrez:2025hty,Aguilar-Gutierrez:2025mxf}.}  Instead, the construction in our method is intrinsically one-sided and turns out to be much simpler than \cite{Lin:2023trc,Xu:2024hoc} and can be easily generalized to multi-matter subspaces.

Let us start with the one-matter states. All such states are spanned
by the basis $\ket{0^{n}10^{m}}$ for all $m,n\in\N$. For an annihilation
operator $a_{R,0}$ acting on this basis, we will have 
\begin{equation}
a_{R,0}\ket{0^{n}10^{m}}=[m]_{q}\ket{0^{n}10^{m-1}}+r[n]_{q}q^{m}\ket{0^{n-1}10^{m}}\label{eq:37}
\end{equation}
Let us imagine the matter chord is a kind of ``EOW brane'' that
separates the $n$ 0-chords behind it and $m$ 0-chords in front of
it. The first term in (\ref{eq:37}) is an ordinary annihilation operation
acting on the 0-matter subspace spanned by $\ket{0^{m}}$. The second
term in (\ref{eq:37}) is proportional to $q^{m}$, which measures
the chord number in front of matter chord. However, this naive ``EOW
brane'' understanding is not quite right because the second term
in (\ref{eq:37}) is not the same as the additional term $\propto q^{N_{0}}r^{N_{1}}$
(here $N_{1}=0$) of $\tilde{a}_{R,0}$ in \eqref{3.22} due to the additional $n$-dependent
coefficient $r[n]_{q}$. Nevertheless, from the construction of $\ket{\mB_{\a}}$,
it is easy to fix this issue by replacing $\ket{0^{n}}$ by a $q$-coherent
state $\ket{\mB_{\a_{1}}}$. Let us define\footnote{In this paper, we have abused the notation of $\otimes$ for the coproduct defined in \cite{Lin:2023trc}. It literally means mapping to a new state labeled by the joining of two binary strings. }  
\begin{equation}
\ket{(\a_{1});0^{m}}=\ket{\mB_{\a_{1}}10^{m}}=\ket{\mB_{\a_1}1}\otimes \ket{0^m}\equiv \sum_{n=0}^{\infty}\f{\a_{1}^{n}}{[n]_{q}!}\ket{0^{n}10^{m}} \label{3.68}
\end{equation}
It is straightforward to check
\begin{align}
a_{R,0}\ket{(\a_{1});0^{m}} &=\sum_{n=0}^{\infty}\f{\a_{1}^{n}}{[n]_{q}!}\left([m]_{q}\ket{0^{n}10^{m-1}}+r[n]_{q}q^{m}\ket{0^{n-1}10^{m}}\right)\nonumber \\
 & =[m]_{q}\ket{(\a_{1});0^{m-1}}+\a_{1}rq^{m}\ket{(\a_{1});0^{m}}\label{eq:40}
\end{align}
which implies that the effective EOW parameter is $\mu=\a_{1}r$. 

Writing in this $q$-coherent basis, we simplifies the 0-chord configurations
behind the matter chord in the one-matter subspace as an effective EOW
brane with an ``zero-matter" subspace. Since $\ket{\mB_{\a}}$ is an
overcomplete basis for zero-matter subspace as seen in \eqref{eq:11}, the states $\ket{(\a_{1});0^{m}}$
is also overcomplete for the one-matter subspace. The original right Hamiltonian
$H_{R,0}$ acting on this basis is simple
\begin{equation}
H_{R,0}\ket{(\a_{1});0^{m}}=\a_{1}rq^{m}\ket{(\a_{1});0^{m}}+[m]_{q}\ket{(\a_{1});0^{m-1}}+\ket{(\a_{1});0^{m+1}}
\end{equation}
which is equivalent to an EOW brane Hamiltonian $\tilde{H}_{R,0}=a_{R,0}^{\dag}+\tilde a_{R,0}$ with $N_{1}=0$ and $\mu=\a_1 r$. As we explained in Section \ref{sec:3.1}, the eigenfunction
of $\tilde{H}_{R,0}$ in zero-matter subspace is the big $q$-Hermite
polynomial and the eigenvalue is parameterized by $E(\t)$. Here for $H_{R,0}$,
we just need to simply replace $\mu\ra\a_{1}r$. In particular, the
spectrum of $H_{R,0}$ in the one-matter subspace is in the range of
$[-2/\sqrt{1-q},2/\sqrt{1-q}]$, which is the same as the zero-matter
subspace. Note that we have implicitly analytically continued the real parameter $\mu$ of the EoW brane to a complex parameter $\a$ with constraint \eqref{eq:18-1}. One can show that this analytic continuation is legitimate and all formula from \eqref{3.7} to \eqref{3.11} are valid.\footnote{A simple way to see this is by rewriting big $q$-Hermite polynomial in terms of $q$-Hypergeometric functions, which is analytic in parameters and arguments in appropriate ranges \cite{gasper2011basic}.}

It is clear from this construction that the EoW brane can understood as a matter chord $q$-coherently dressed with gravity chords behind it. Physically speaking, this is microscopically how the gravitating EoW brane is generated from the matter. As we will see shortly, by going to higher-matter subspaces, the gravitational dressing has a screening effect in a sense that the configurations behind the EoW brane formed by the rightmost matter is invisible from the gravitational probe on right boundary.

Using this technique, we can easily generalize it to higher-matter
subspaces. Let us consider the two-matter subspace, which we expand in the following basis
\be 
\ket{\mB_{\a_{2}}1\mB_{\a_{1}}10^{m}}=\ket{\mB_{\a_2}1}\otimes\ket{(\a_1);0^m}\equiv \sum_{n_{1,2}=0}^{\infty}\f{\a_{2}^{n_{2}}\a_{1}^{n_{1}}}{[n_{2}]_{q}![n_{1}]_{q}!}\ket{0^{n_{2}}10^{n_{1}}10^{m}} \label{3.71}
\ee
Acting $a_{R,0}$ on this state, we have 
\begin{align}
 & a_{R,0}\ket{\mB_{\a_{2}}1\mB_{\a_{1}}10^{m}}=\sum_{n_{1,2}=0}^{\infty}\f{\a_{2}^{n_{2}}\a_{1}^{n_{1}}}{[n_{2}]_{q}![n_{1}]_{q}!}\left([m]_{q}\ket{0^{n_{2}}10^{n_{1}}10^{m-1}}\right. \nn\\
 & \left.+r[n_{1}]_{q}q^{m}\ket{0^{n_{2}}10^{n_1-1}0^{m}}+r^{2}[n_{2}]_{q}q^{m+n_{1}}\ket{0^{n_{2}-1}1 0^{n_1}10^{m}} \right)\nonumber \\
= & [m]_{q}\ket{\mB_{\a_{2}}1\mB_{\a_{1}}10^{m-1}}+\a_{1}rq^{m}\ket{\mB_{\a_{2}}1\mB_{\a_{1}}10^{m}}+\a_{2}r^{2}q^{m}\ket{\mB_{\a_{2}}1\mB_{q\a_{1}}10^{m}}\label{eq:42}
\end{align}
Here we see that $a_{R,0}$ acts on $\mB_{\a_{1}}$ as if on an EOW
brane but produces an additional term with shift $\a_{1}\ra q\a_{1}$
acting on $\mB_{\a_{2}}$. To solve this issue, we need to consider
the linear combination
\begin{equation}
\ket{(\a_{2},\a_{1});0^{m}}\equiv\sum_{n=0}^{\infty}c_{n}\ket{\mB_{\a_{2}}1\mB_{q^{n}\a_{1}}10^{m}}\label{eq:43}
\end{equation}
such that it only produces a constant coefficient $\a_{1}r$ besides
the first term in (\ref{eq:42}). One can easily find from (\ref{eq:42})
that 
\begin{equation}
c_{n}/c_{n-1}=\f{\a_{2}r}{\a_{1}(1-q^{n})}\implies c_{n}=\f 1{[n]_{q}!}\left(\f{\a_{2}r}{\a_{1}(1-q)}\right)^{n}
\end{equation}
which leads to 
\begin{equation}
a_{R,0}\ket{(\a_{2},\a_{1});0^{m}}=[m]_{q}\ket{(\a_{2},\a_{1});0^{m}}+\a_{1}rq^{m}\ket{(\a_{2},\a_{1});0^{m}}
\end{equation}

Note that this is exactly in the same form as (\ref{eq:40}) in the
1-matter subspace and the effective EOW parameter is unchanged $\a=\a_{1}r$,
which only depends on the right most $q$-coherent sector parameter
$\a_{1}$. We can physically understand this as a kind of screening
effect. The right boundary observer with access only to $a_{R,0}$
and $a_{R,0}^{\dag}$ cannot tell the difference between $\ket{(\a_{2},\a_{1});0^{m}}$
and $\ket{(\a_{1});0^{m}}$ because the effect of the additional matter
chord on the left is completely screened by the delicate superposition
of $\ket{\mB_{q^{n}\a_{1}}}$ in the middle sector in (\ref{eq:43}). However, the screening effect only holds for $a_{R,0},a_{R,0}^\dag,a_{R,1}^\dag$ but does not hold for $a_{R,1}$ because the latter will destroy the delicate dressing by annihilating a matter chord in the middle. This is consistent with the fact that $\ket{\w}$ is a cyclic state for the right algebra $\mA_R$.

An immediate consequence is that the spectrum of $H_{R,0}$ in the
two-matter subspace is identical to that in the one- and zero-matter subspace. Expanding
(\ref{eq:43}) and complete the sum over $n$, we can rewrite it as  
\begin{align}
\ket{(\a_{2},\a_{1});0^{m}} & =\sum_{n}\f{(\a_{2}rq^{n_{1}}/\a_{1})^{n}}{(q;q)_{n}}\sum_{n_{1},n_{2}}\f{\a_{1}^{n_{1}}\a_{2}^{n_{2}}}{[n_{1}]!_{q}[n_{2}]!_{q}}\ket{0^{n_{2}}10^{n_{1}}10^{m}}\nonumber \\
 & =\sum_{n_{1},n_{2}}\f{\a_{1}^{n_{1}}\a_{2}^{n_{2}}}{(\a_{2}rq^{n_{1}}/\a_{1};q)_{\infty}[n_{1}]!_{q}[n_{2}]!_{q}}\ket{0^{n_{2}}10^{n_{1}}10^{m}}\label{eq:47}
\end{align}
which is free of singularity when $|\a_2|r<|\a_1|$.

For three-matter subspace, we can iterate above manipulation. We start
with 
\be
\ket{\mB_{\a_{3}}1(\a_{2},\a_{1});0^{m}}\equiv \ket{\mB_{\a_3}1}\otimes \ket{(\a_{2},\a_{1});0^{m}}
\ee
where the tensor product should be clear from the examples \eqref{3.68} and \eqref{3.71}. Straightforward computation leads to
\begin{align}
a_{R,0}&\ket{\mB_{\a_{3}}1(\a_{2},\a_{1});0^{m}}=  [m]_{q}\ket{\mB_{\a_{3}}1(\a_{2},\a_{1});0^{m}}\nonumber \\
 & +\a_{1}rq^{m}\ket{\mB_{\a_{3}}1(\a_{2},\a_{1});0^{m}}+\a_{3}r^{3}q^{m}\ket{\mB_{\a_{3}}1(q\a_{2},q\a_{1});0^{m}}
\end{align}
which is in the same form as (\ref{eq:42}) with replacement $\a_{2}r^{2}\ra\a_{3}r^{3}$
in the last term. Therefore, we should define
\begin{equation}
\ket{(\a_{3},\a_{2},\a_{1});0^{m}}\equiv\sum_{n=0}^{\infty}\f 1{[n]_{q}!}\left(\f{\a_{3}r^{2}}{\a_{1}(1-q)}\right)^{n}\ket{\mB_{\a_{3}}1(q^{n}\a_{2},q^{n}\a_{1});0^{m}}\label{eq:49}
\end{equation}
which leads to
\begin{equation}
a_{R,0}\ket{(\a_{3},\a_{2},\a_{1});0^{m}}=[m]_{q}\ket{(\a_{3},\a_{2},\a_{1});0^{m}}+\a_{1}rq^{m}\ket{(\a_{3},\a_{2},\a_{1});0^{m}}
\end{equation}
Similar to (\ref{eq:47}), we can rewrite the state as 
\begin{align}
&\ket{(\a_{3},\a_{2},\a_{1});0^{m}} =\sum_{n,n_{i}}\f{(\a_{3}r^{2}q^{n_{1}+n_{2}}/\a_{1})^{n}\a_{1}^{n_{1}}\a_{2}^{n_{2}}\a_{3}^{n_{3}}}{(q;q)_{n}(\a_{2}rq^{n_{1}}/\a_{1};q)_{\infty}[n_{1}]!_{q}[n_{2}]!_{q}[n_{3}]!_{q}}\ket{0^{n_{3}}10^{n_{2}}10^{n_{1}}10^{m}}\nonumber \\
 =&\sum_{n_{i}}\f{\a_{1}^{n_{1}}\a_{2}^{n_{2}}\a_{3}^{n_{3}}}{(\a_{3}r^{2}q^{n_{1}+n_{2}}/\a_{1};q)_{\infty}(\a_{2}rq^{n_{1}}/\a_{1};q)_{\infty}[n_{1}]!_{q}[n_{2}]!_{q}[n_{3}]!_{q}}\ket{0^{n_{3}}10^{n_{2}}10^{n_{1}}10^{m}}\label{eq:47-1}
\end{align}
which is free of singularity when $|\a_3|r^2<|\a_1|$ and $|\a_2|r<|\a_1|$.

Now the rule of recurrence is clear for higher number matter subspaces.
Let us define the following state for $k$-matter subspace
\begin{align}
&\ket{(\a_{k},\a_{k-1},\cdots,\a_{1});0^{m}} \equiv\sum_{n=0}^{\infty}\f 1{[n]_{q}!}\left(\f{\a_{k}r^{k-1}}{\a_{1}(1-q)}\right)^{n}\ket{\mB_{\a_{k}}1(q^{n}\a_{k-1},\cdots,q^{n}\a_{1});0^{m}}\label{eq:88-1}\\
=&\sum_{n_{i}}\left(\prod_{i=1}^{k}\f{\a_{i}^{n_{i}}}{[n_{i}]!_{q}}\right)\left(\prod_{i=1}^{k-1}\f 1{(\a_{i+1}r^{i}q^{\sum_{j=1}^{i}n_{j}}/\a_{1};q)_{\infty}}\right)\ket{0^{n_{k}}1\cdots0^{n_{1}}10^{m}}\label{eq:52}
\end{align}
on which $a_{R,0}$ acts as
\begin{equation}
a_{R,0}\ket{(\a_{k},\cdots,\a_{1});0^{m}}=[m]_{q}\ket{(\a_{k},\cdots,\a_{1});0^{m-1}}+\a_{1}rq^{m}\ket{(\a_{k},\cdots,\a_{1});0^{m}}
\end{equation}
For the coefficients in (\ref{eq:52}) to be free of singularity, we have
\begin{equation}
|\a_{i}|r^{i-1}<|\a_{1}|,\quad |\a_1|,|\a_i|<(1-q)^{-1/2},\quad  \forall i=2,3,\cdots,k
\end{equation}
for all states $\ket{(\a_{k},\a_{k-1},\cdots,\a_{1});0^{m}}$. For convenience, let us call the states in \eqref{eq:88-1} as ($k$-)matter-brane states.

By construction, the screening effect generalizes to arbitrary matter subspaces and an effective EoW brane can be microscopically generated by arbitrary numbers of matter chords with appropriate but delicate gravitational dressing, and the effective tension $\mu$ only relies on the rightmost parameter $\a_1=\mu/r$. This idea is reminiscent of the proposal in \cite{Lin:2023trc} to decompose higher matter states into the irreducible representation of single matter states. Using $\ket{(\a_k,\cdots,\a_1);0^n}$ to replace $\ket{n}$ in the RHS of \eqref{3.7} with $\mu=\a_1 r$ and replacing $\tilde H_{bulk}$ with $H_{R,0}$, we can diagonalize $H_{R,0}$ in the subspace with any number of matter chords. This completely solves the spectrum and the eigenstates of $H_{R,0}$, and in particular shows that its spectrum can be parameterized by $E(\t)$ with $\t\in[0,\pi]$. This justifies our earlier statement that $\mB_{R,\a}$ in \eqref{eq:17-5} is a bounded and invertible operator as long as \eqref{eq:18-1} holds. Indeed, this also justifies that the tilde version $\tilde H_{R,0}$ has the same spectrum as $H_{R,0}$ simply because the $k$-matter subspace for $\tilde H_{R,0}$ is just equivalent to the $(k+1)$-matter subspace for $H_{R,0}$ (see Appendix \ref{app:c}). Moreover, in our treatment, $H_{R,0}$ and $H_{R,1}$ are put on equal footing, and they can be switched to each other by changing parameters $q\leftrightarrow  q_m$. Therefore, the spectrum of the matter operator $H_{R,1}$ can also be solved in the same way with parameterization $E_m(\t)=\f{2}{\sqrt{1-q_m}}\cos\t$ with $\t\in[0,\pi]$.

Above construction of the matter-brane states is by a linear combination of the overcomplete basis $\ket{\mB_{\a_k}1\cdots\mB_{\a_1}10^m}$. In Appendix \ref{app:a}, we show that all matter-brane states also form a overcomplete basis of $\mH$ by working out the inverse transformation to the chord number (binary string) basis. We show that each chord number basis can be written as
\begin{align}
\ket{0^{n_{k}}1\cdots10^{n_{1}}10^{m}}=&\oint_{}\prod_{i=1}^{k}\f{[n_{i}]!_{q}d\a_{i}}{2\pi\i\a_{i}^{n_{i}+1}}\oint\prod_{i=1}^{k-1}\f{dy_{i}}{2\pi\i}K(\a_{i+1}r^{i}/\a_{1};y_{i}) \nn\\
&\times \ket{(\a_{k},\cdots,\f{\a_{j}}{y_{j}\cdots y_{k-1}},\cdots,\f{\a_{1}}{y_{1}\cdots y_{k-1}});0^{m}}\\
 K(z;y) \equiv &\f{(zqy;q)_{\infty}}y{}_{2}\phi_{1}\left(\left.\begin{array}{c}
q,0\\
zqy
\end{array}\right|q;y\right)
\end{align}
where $_2\phi_1$ is a $q$-Hypergeometric function (see Appendix \ref{app:b}), and the integral contour obeys $|\a_i|r^{i-1}<|\a_{1}|<|y_j|<1$, $|\a_j|/|y_j\cdots y_{k-1}|,|\a_k|<(1-q)^{-1/2}$ for $i=2,\cdots,k$ and $j=1,\cdots,k-1$. This shows that the full Hilbert space $\mH$ can also be expanded non-uniquely by the overcomplete basis of matter-brane states. 

Taking advantage of the matter-brane states, we can compute correlation functions in DSSYK using simple complex calculus without dealing with the involved combinatorics of chord diagrams in \cite{Berkooz:2018jqr}. Beyond the Feynmann rules developed in \cite{Berkooz:2018jqr}, this method can also be used to systematically compute correlation functions in generic states (not just vacuum $\ket{\w}$) without counting chords, which might be subtle in practice. In Appendix \ref{app:c}, we illustrate the computation of two-point and OTOC four-point function in the vacuum, and show consistency with \cite{Berkooz:2018jqr}. 

\section{Boundary algebra extension and no man's island}
\label{more on boundary algebra}

\subsection{Revisit the commutant of $\tilde{\mA}_R$} \label{sec:4.1}

As shown in Section \ref{sec:3.3}, the von Neumann algebra of the single-sided black hole with EoW brane in DSSYK is a type II$_1$ factor and has a nontrivial commutant. However, the commutant is warrant only by the Tomita theory and its construction in \eqref{3.62} is quite formal on the non-chord number basis $\kit{i_1 \cdots i_k}$. To understand this mysterious commutant better, let us try to reconstruct it in terms of creation and annihilation operators $a_{L,i}^\dag$ and $a_{L,i}$. 

We have attempted two ways to rewrite generators of $\tilde \mA_L$. The first way is to consider a simple case with all $Q_{ij}=0$, in which the two types of chords are prohibited to intersect each other and they obey a free probability distribution. Similar limit of free probability in DSSYK with multiple types of chords was considered in \cite{Gao:2024lve} for novel quantum chaotic behavior. From this special limit, it is relatively easy to construct $\tilde H_{L,0}$ explicitly from a series expansion as shown in Appendix \ref{app:c}. Then it is not hard to generalize it to the generic $Q_{ij}$ case. The second way is to work on a new basis $\tilde H_{R,i_1}\cdots \tilde H_{R,i_k}\ket{\w}$, on which 
\begin{align} 
&\left(\tilde \mB_{L,\mu}^{1/2}\tilde H_{L,i}\tilde\mB_{L,\mu}^{-1/2}\right)\tilde H_{R,i_1}\cdots \tilde H_{R,i_k}\ket{\w}=\tilde H_{R,i_1}\cdots \tilde H_{R,i_k}\tilde \mB_{R,\mu}^{1/2}\tilde \mB_{L,\mu}^{1/2}\tilde H_{L,i}\tilde \mB_{L,\mu}^{-1/2} \ket{\psi} \nn\\
=&\tilde H_{R,i_1}\cdots \tilde H_{R,i_k} \tilde H_{R,i}\tilde\mB_{R,\mu}^{1/2}\ket{\psi} =\tilde H_{R,i_1}\cdots \tilde H_{R,i_k} \tilde H_{R,i}\ket{\w} \label{4.1}
\end{align}
where $\tilde \mB_{L,\mu}$ is by replacing $H_{R,0}$ in \eqref{eq:17-5} with $\tilde H_{L,0}$, and in the second step we used \eqref{3.63}. This means that $\mB_{L,\mu}^{1/2}\tilde H_{L,i}\tilde\mB_{L,\mu}^{-1/2}$  is just like a left operator of chord number states with $H_{R,i}\ra \tilde H_{R,i}$. Then we can expand RHS of \eqref{4.1} in the original chord number basis and compare with the action of $a_{L,i}$ and $a_{L,i}^\dag$. It turns out that $\tilde H_{L,0}=\mB_{L,\mu}^{1/2}\tilde H_{L,0}\tilde\mB_{L,\mu}^{-1/2}$ has a relatively simple action, which can be rewritten into a compact form identical to the one derived from the first method. But the expression of $\mB_{L,\mu}^{1/2}\tilde H_{L,1}\tilde\mB_{L,\mu}^{-1/2}$  is quite complicated and we suspect that it does not have a simple compact form.

The construction we find for $\tilde H_{L,0}$ is
\begin{equation}
    \tilde H_{L,0}=a_{L,0}+a_{L,0}^\dagger-\mu(1-q)a_{L,0}^\dagger a_{L,0}+\mu \I \label{4.2}
\end{equation}
Let us first verify it commutes with $\tilde H_{R,0}$ and $\tilde H_{R,1}$. The latter is obvious due to \eqref{2.27} and $[H_{L,0},H_{R,1}]=0$. For the former, we have\footnote{In the second line of \eqref{4.3}, the fact that the operator $a^\dag_{L,0}a_{L,0}$ commutes with $q^{N_0}r^{N_1}$ was used in \cite{Lin:2023trc} to construct the symmetry algebra $U(J)$.}
\begin{align}
    &[\tilde H_{L,0},\tilde H_{R,0}]=[a_{L,0}+a_{L,0}^\dagger-\mu(1-q)a_{L,0}^\dagger a_{L,0},a_{R,0}+a_{R,0}^\dagger+\mu q^{N_0}r^{N_1}]\nonumber\\
    =&-\mu(1-q)[a_{L,0}^\dagger a_{L,0},a_{R,0}+a_{R,0}^\dagger]+\mu[a_{L,0}+a_{L,0}^\dagger,q^{N_0}r^{N_1}]-\mu^2(1-q)[a^\dagger_{L,0} a_{L,0},q^{N_0}r^{N_1}]\nonumber\\
    =&\mu(1-q)(q^{N_0}r^{N_1} a_{L,0}-a^\dagger_{L,0} q^{N_0}r^{N_1}-q^{N_0}r^{N_1} a_{L,0}+a^\dagger_{L,0} q^{N_0}r^{N_1})=0 \label{4.3}
\end{align}
where we used \eqref{2.27} and 
\begin{equation}
    [a_{L,0},q^{N_0}r^{N_1}]_q=0\quad [q^{N_0}r^{N_1},a_{L,0}^\dagger]_q=0
\end{equation}
On the other hand, it acts on $\ket{\w}$ in the same ways as $\tilde H_{R,0}$
\be 
\left(a_{L,0}+a_{L,0}^\dagger-\mu(1-q)a_{L,0}^\dagger a_{L,0}+\mu \I\right)\ket{\w}=\left(H_{R,0}+\mu \I\right)\ket{\w}=\tilde H_{R,0}\ket{\w}
\ee
which is consistent with \eqref{4.1}.

As we see from \eqref{4.2} that $\tilde H_{L,0}$ is built from $H_{L,0}$ with an additional simple term $a_{L,0}^\dag a_{L,0}$ that does not does change the number of chords. However, for $\tilde H_{L,1}$, if we take the ansatz of $H_{L,1}+X(a_{L,i},a_{L,i}^\dag,N_0,N_1)$, it seems hard to find a simple $X$ such that it commutes with $H_{R,1}$ and $\tilde H_{R,0}$ simultaneously. A simple technical explanation is that $\tilde H_{R,0}$ has an additional term $\mu q^{N_0}r^{N_1}$, which can probe how many chords of each type in a state. For $H_{L,0}$, its commutator with $\mu q^{N_0}r^{N_1}$ generate 0-digit operators, and we only need to solve one equation \eqref{4.3} by adding 0-digit operators into $\tilde H_{R,0}$ to erase this probe using its commutator with $H_{R,0}$; For $H_{L,1}$, its commutator with $\mu q^{N_0}r^{N_1}$ generate 1-digit operators, we have to add both 0 and 1 digit operators into $\tilde H_{L,1}$ to cancel it, which requires solving two nontrivial equations $[\tilde H_{L,1},\tilde H_{R,i}]=0$ for both $i=0,1$. 

Somehow this is not too surprising from the JT viewpoint. Acting on the vacuum $\ket{\w}$, the operator $\tilde H_{L,0}$, relative to the EoW brane Hamiltonian $\tilde H_{R,0}$, is adding gravity chords avoid being probed from the right boundary. We can heuristically say that these gravity chords are ``behind" the EoW brane, which is perturbatively not too awkward from JT perspective since EoW brane is gravitating; however, adding a matter chord ``behind" the EoW brane to avoid being probed should not be a simple task from JT perspective because there is no spacetime behind the EoW brane for the matter to live in and generating a spacetime region from the void behind the EoW brane should be quite non-perturbative. Since the chord basis becomes the continuum spacetime in the JT limit, we may heuristically say that $\tilde \mA_L$ is non-geometric.

\subsection{A minimal extension to all operators} \label{sec:4.2}

Another way to understand the commutant of $\tilde \mA_R$ is to ask an inverse question: what additional operator is required for $\tilde \mA_R$ to generate all bounded operators $\mB(\mH)$ of the full Hilbert space $\mH$? Given this additional operator, the algebra generated by them becomes type I. This fills the gap between the algebra $\tilde \mA_R$ of a single-sided black hole with EoW brane and the naive expectation of the complete bulk reconstruction as discussed in the end of Section \ref{sec:3.1}.

A sort of trivial answer is adding $\tilde \mH_{L,i}$ since $\tilde \mA _{L,R}$ are factors and thus $\tilde \mA_L \vee\tilde \mA_R=\mB(\mH)$. However, as we show in \eqref{4.2}, the generators of the left algebra $\tilde \mA_L$ look like some operation behind the EoW brane, and cannot be easily understood from the right boundary point of view. We would like to ask a refined question: are there some ``simple" operators from the right boundary point of view, together with $\tilde \mA_R$ generating all bounded operators $\mB(\mH)$?

The answer is surprisingly simple. We just need one little operator $q^{N_0}r^{N_1}$ that measures how many gravity and matter chords in a state by drawing a gravity chord from the EoW brane to the right boundary. Let us first introduce a lemma.
\begin{lemma}\label{completelemma}
    $\{a_{R,0},a^\dagger_{R,0},a_{R,1},a^\dagger_{R,1}\}''=B(\mathcal H)$
\end{lemma}
\begin{proof}
    We first prove that the vacuum $\ket{\w}$ is unique in the sense that it is the only state satisfying $a_{R,i}\ket{\w}=0$ for $i=0,1$. The proof is by contradiction. Suppose there is another state $\ket{\w'}$ satisfying the same property, without loss of generality we can assume that $\ip{\w}{\w'}=0$, if this is not the case we can use Gram-Schmidt process to construct a new $\ket{\w'}$. As $a_{R,i}\ket{\w'}=0$, we conclude that $\avg{\w'|\phi}=0$ for arbitrary chord number state $\ket\phi$, which is in the form of $a_{R,i_1}^\dagger\cdots a_{R,i_n}^\dagger\ket{\w}$. Since the Hilbert space $\mH$ is spanned by all chord number states, $\ket{\w'}$ is orthogonal to a dense set of states in $\mathcal H$ and it has to be $0$. Now suppose $\hat O\in\{a_{R,0},a^\dagger_{R,0},a_{R,1},a^\dagger_{R,1}\}'$. We then have $a_{R,i}\hat O\ket{\w}=\hat Oa_{R,i}\ket{\w}=0$, and thus $\hat O\ket{\w}=c\ket{\w}$ for some number $c$. As $[\hat O,a_{R,i}^\dagger]=0$ we have $\hat O\ket{\psi}=c\ket{\phi}$ for arbitrary chord number state $\ket{\phi}$. It follows that $\hat O=cI$ on a dense set in $\mathcal H$ and therefore $\hat O$ is a scalar multiplication, that is $\{a_{R,0},a^\dagger_{R,0},a_{R,1},a^\dagger_{R,1}\}'=\mathbb C$. Or equivalently $\{a_{R,0},a^\dagger_{R,0},a_{R,1},a^\dagger_{R,1}\}''=B(\mathcal H)$. 
\end{proof}

\begin{theorem}\label{thm7}
    $\{\tilde H_{R,0},\tilde H_{R,1},q^{N_0}r^{N_1}\}''=\{H_{R,0},H_{R,1},q^{N_0}r^{N_1}\}''=B(\mathcal H)$.
\end{theorem}
\begin{proof}
    We will prove the second equality since the first one is obvious. Suppose $Q_{ij}$ are all generic positive numbers and less than one (the case for $Q_{ij}=0$ will be discussed in Appendix \ref{freecaseproof}). In this case $q^{N_0}r^{N_1}$ is a compact operator with discrete spectrum $\lambda_i=q^{n_0}r^{n_1}$, where $n_{0},n_1\geq 0$ are non-negative integers. From the spectral theory of von-Neumann algebras,\footnote{See, for example, Theorem 4.1.11 of~\cite{murphy2014c}.} we can show that all projections onto eigen-spaces with fixed values $q^{n_0}r^{n_1}$ can be generated by $q^{N_0}r^{N_1}$. Indeed, we can write
    \be 
     q^{N_0}r^{N_1}=\sum_{n_0,n_1\geq 0}q^{n_0}r^{n_1} P_{n_0,n_1}=\sum_i \lam_i P_{\lam_i} \label{4.6}
    \ee
    where $P_{n_0,n_1}$ is the projection onto the subspace $\mH_{N_0,N_1}$ with $N_0=n_0$ and $N_1=n_1$, and $P_{\lam_i}$ is the projection to the subspaces $\mH_{n_0,n_1}$ with the same eigenvalue $\lam_i$. There could be multiple $P_{n_0,n_1}$ in one $P_{\lam_i}$ when $r^{p_1}=q^{p_2}$ for two coprime integers $p_{1,2}$. A limit form of the Fourier transformation of the LHS of \eqref{4.6} gives all projections
    \be
    P_{\lam_i}=\lim_{M\ra\infty}\f{1}{2M}\int_{-M}^{M} ds \left(\f{q^{N_0}r^{N_1}}{\lam_i}\right)^{\i s}\label{4.7}
    \ee
    
    Now consider the operator $\sum_{i} P_{q \lam_i}H_{R,0}P_{\lam_i}$. Clearly this is just $a_{R,0}^\dagger$ because it only has nontrivial matrix elements from eigenspace of $\lam_i$ to $q\lam_i$ and $a_{R,0}$ only has nontrivial matrix elements from eigenspace of $\lam_i$ to $q^{-1}\lam_i$. Since the projectors with different eigenvalue are orthogonal to each other, we see that $a_{R,0}$ and $a_{R,0}^\dagger$ in $H_{R,0}$ can be separated using spectral projections. Similarly, $\sum_{i} P_{r\lam_i} H_{R,1} P_{\lam_i}=a_{R,1}^\dagger$, and thus we have all creation and annihilation operators $a_{R,i}^\dag$ and $a_{R,i}$ from $\{H_{R,i},q^{N_0}r^{N_1}\}''$. By Lemma~\ref{completelemma},  the algebra is $B(\mathcal H)$.  
\end{proof} 

Theorem~\ref{thm7} tells us that we do not need much to recover the whole algebra: we only need to add the exponential wormhole length operator $e^{-\bar\l_b}$ back. Here the length operator is a generalization of $\l_b$ in \eqref{2.44}\footnote{In \cite{Lin:2023trc}, $\bar\l_b$ was used to generate $a_{R,0}^\dag$ and $a_{R,0}$ by $a_{R,0}^\dag=\f 1 2(H_{R,0}-\lam^{-1}[H_{R,0},\bar \l_b])$ and $a_{R,0}=\f 1 2(H_{R,0}+\lam^{-1}[H_{R,0},\bar \l_b])$. Similar equality applies to $a_{R,1}^\dag$ and $a_{R,1}$. However, $\bar\l_b$ is not a bounded operator and we must use its exponential form for Theorem \ref{thm7}.}
\be 
\bar \l_b=\lam (N_0+\D N_1),\quad r\equiv e^{-\D \lam}
\ee
From the proof we can see that the conclusion generalizes to the case where we have more types of matter chords. Note that this theorem also holds for the single side algebra $\mA_R$ in the ordinary two-sided DSSYK, where it says as long as the right boundary observer gets access to the wormhole length, then the whole algebra $\mA_L$ on the opposite left boundary becomes accessible, no matter how many generators $\mA_L$ has (say, in the case of multiple types of matter operators). In this sense, this simple geometric operator of exponential wormhole length $e^{-\bar\l_b}$ encodes the entanglement between $\mA_R$ and $\mA_L$ in a universal way. We regard the extension from $\mA_R$ or $\tilde \mA_R$ to $\mB(\mH)$ with $e^{-\bar\l_b}$ as a minimal extension to the whole algebra from a single side. 

This exponential wormhole operator indeed raises an important question about how to define the boundary algebra. In the main text we follow the idea of the original DSSYK and only include the appropriate deformation of the right operators $H_{R,i}\ra \tilde H_{R,i}$ as the boundary algebra. However, from the chord diagram rules of $e^{-\bar \l _b}$, it is also very natural to understand this operator as a boundary operator since it has appeared in $\tilde H_{R,0}$. The only issue is whether we need to include it {\it independently}. Taking these two perspectives with seemingly small difference actually leads to a big change in the consequence: $\tilde \mA_R$ without $e^{-\bar \l _b}$ is a type II$_1$ factor with nontrivial commutant, but $\tilde \mA_R\vee \{e^{-\bar \l _b}\}$ is the full algebra $\mB(\mH)$ and is type I$_\infty$. 

These two perspectives define two different systems and should have different interpretations. In Section \ref{sec:4.3} and \ref{sec:4.4}, we will take the first perspective to discuss the nature of the commutant of $\tilde \mA_R$ or $\mA_R$ in the semiclassical JT limit. In Section \ref{sec:4.5}, we will take the second perspective and discuss its connection to another EoW brane setup in \cite{Kourkoulou:2017zaj}.

\subsection{No man's island behind horizon in the semiclassical JT limit} \label{sec:4.3}

We have shown in Section \ref{sec:3.3} that the von Neumann algebra $\tilde \mA_R$ of the single-sided black hole with EoW brane is a type II$_1$ factor. Despite this refutes the naive expectation of type I from JT perspective, one could still wonder if the JT limit of $\tilde\mA_R$ becomes type I and reconciles this tension. In general, to rigorously prove how an algebra changes type in certain limit is very hard because the Hilbert space may change and only part of operators survives in that limit, and thus we need to make sure all survival operators are carefully examined. Nevertheless, we will argue in the following that the semiclassical JT limit of $\tilde \mA_R$ with a heavy EoW brane is of type III$_1$, which is equivalent to the bulk matter algebra in the causal wedge $\mW_C$ outside of the black hole horizon (see Figure \ref{pic-jt-eow}). In other words, the entanglement wedge $\mW_E$ of the right boundary algebra is the same as the causal wedge $\mW_C$ in this strict semiclassical JT limit, which will be defined shortly.

Generally, for the infinite $N$ limit of an algebra $\mA$ in a holographic theory, we should not just consider the limit of the algebra itself, but also the state $\ket{\phi}$, on which the algebra $\mA$ acts. In a generic case at finite $N$, all states can be generated from a vacuum state by acting $\mA$ on it. However, the Hilbert space in large $N$ limit may lose some states and the rest are factorized into different sectors. Some sectors correspond to low energy excitations on a semiclassical geometry, while others may in principle not have a geometric representation. For the former, each sector can be described by a GNS Hilbert space $\mH_\phi$, where $\ket{\phi}\in\mH_\phi$ is a state dual to a spacetime geometry $\mM_\phi$, and all states in $\mH_\phi$ are generated by bulk local fields $\chi_i$ in $\mM$ acting on $\ket{\phi}$, namely $\mH_\phi=\text{span}\{\ket{\phi},\chi_i\ket{\phi},\chi_i\chi_j\ket{\phi},\cdots\}$ \cite{Leutheusser:2021frk,Leutheusser:2022bgi}.\footnote{In the original language of subregion-subalgebra duality \cite{Leutheusser:2022bgi}, the bulk Hilbert space is called Fock space, which is equivalent to the GNS Hilbert space constructed by boundary operators. Due to this equivalence, here we use a mixed language of GNS construction for simplicity.} Different sectors are equipped with different geometry and should be orthogonal to each other in the infinite $N$ limit. The algebra of all bulk local fields $\chi$ is denoted as $\mA_\chi$ and each sector is labeled by the pairing $\{\mA_\chi,\ket{\phi}\}$. An example of the Hilbert space of JT gravity with conformal matter under large $N$ limit (there $N$ is the boundary dilaton value $\phi_b$) was recently studied in \cite{Gao:2024gky}, where each sector is based on a partially entangled thermal state (PETS) $\ket{\phi}$ dual to a long wormhole geometry.

\begin{figure}
	\centering
    \begin{subfigure}{0.32\textwidth}
    \centering
        \includegraphics[height=4.2cm]{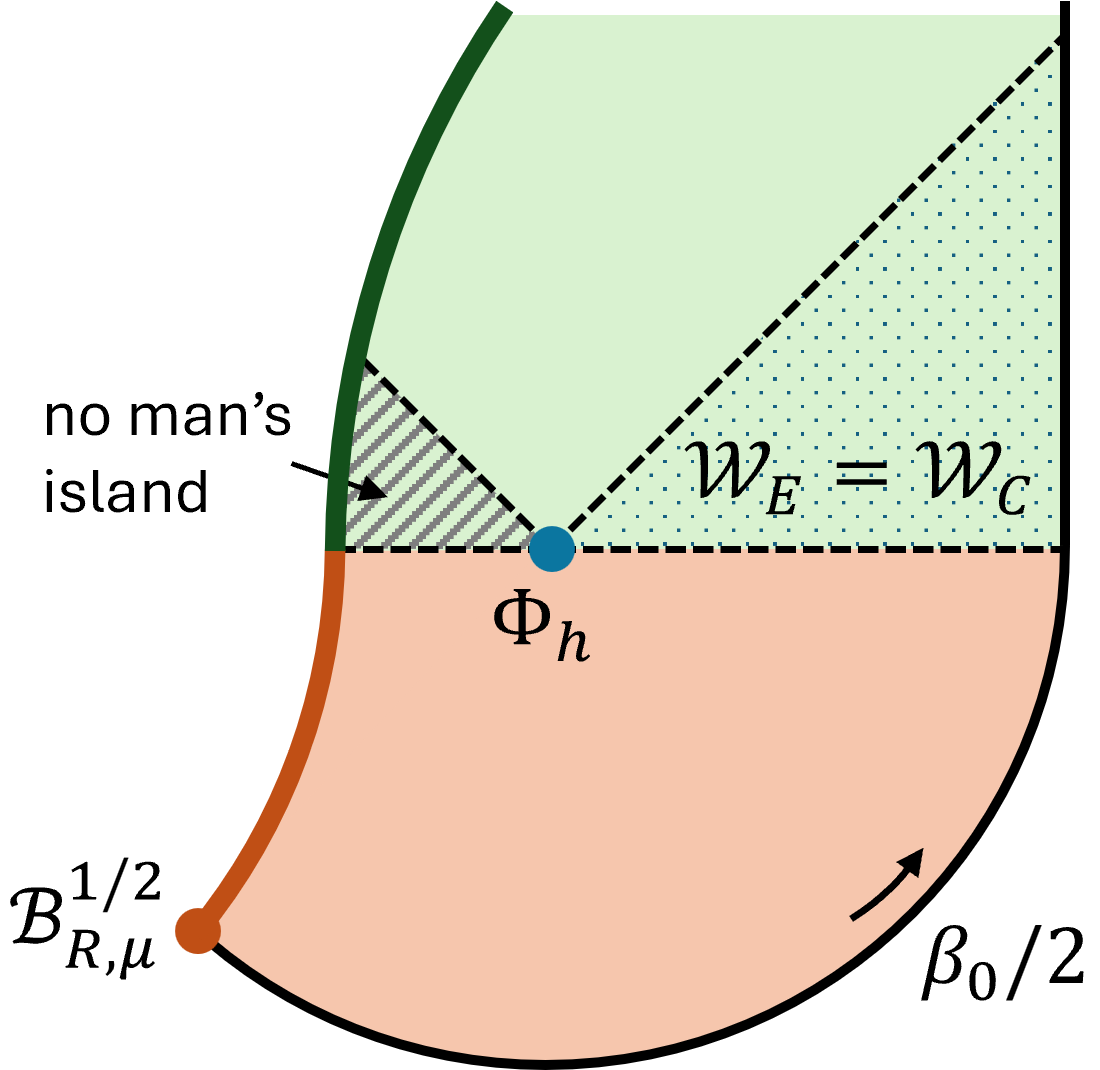}
        \caption{}\label{pic-jt-eow}
    \end{subfigure}
    \begin{subfigure}{0.32\textwidth}
    \centering
        \includegraphics[height=4.2cm]{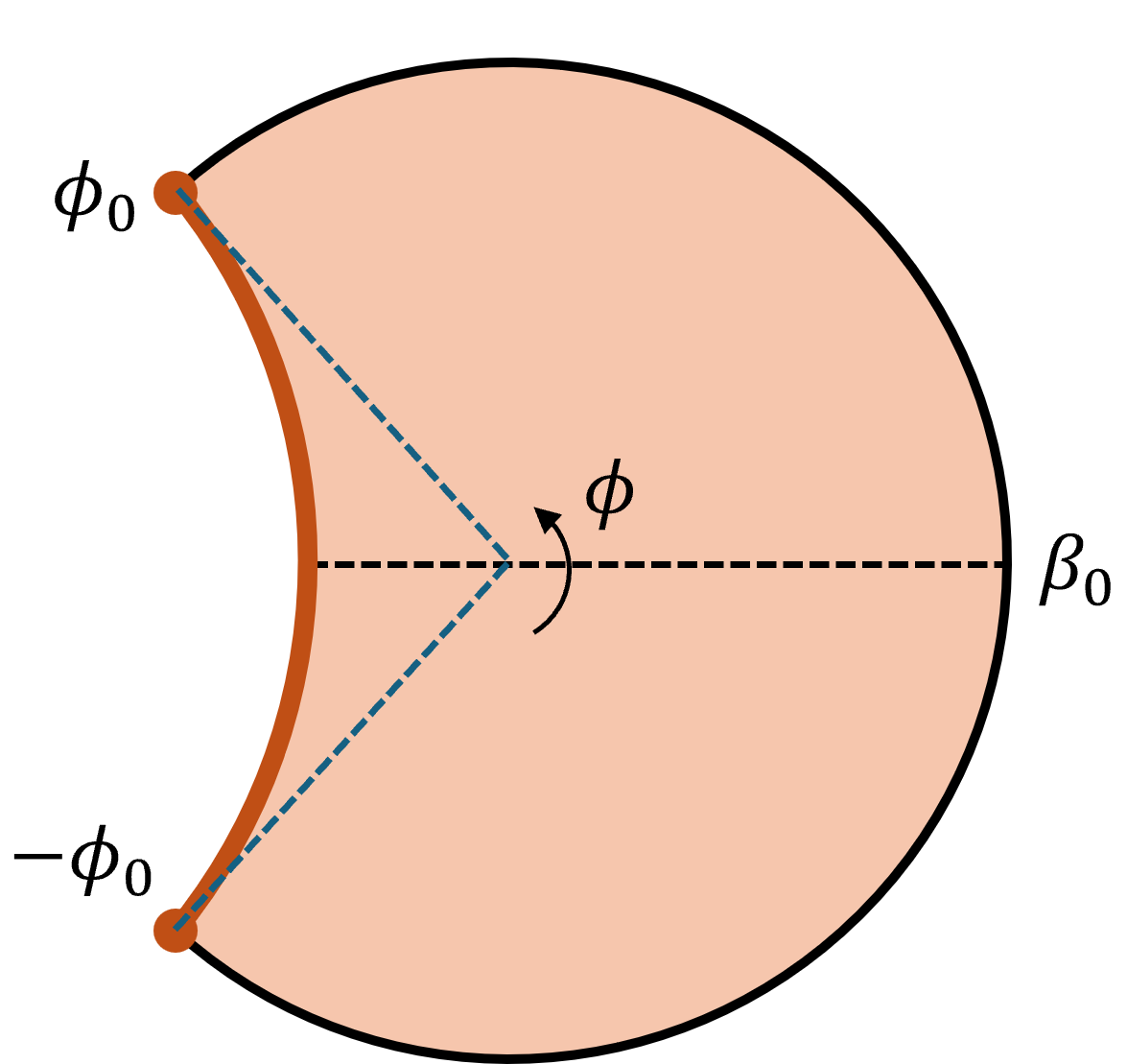}
         \caption{}\label{pic-jt-eow-2}
    \end{subfigure}
    \begin{subfigure}{0.32\textwidth}
    \centering
        \includegraphics[height=4.2cm]{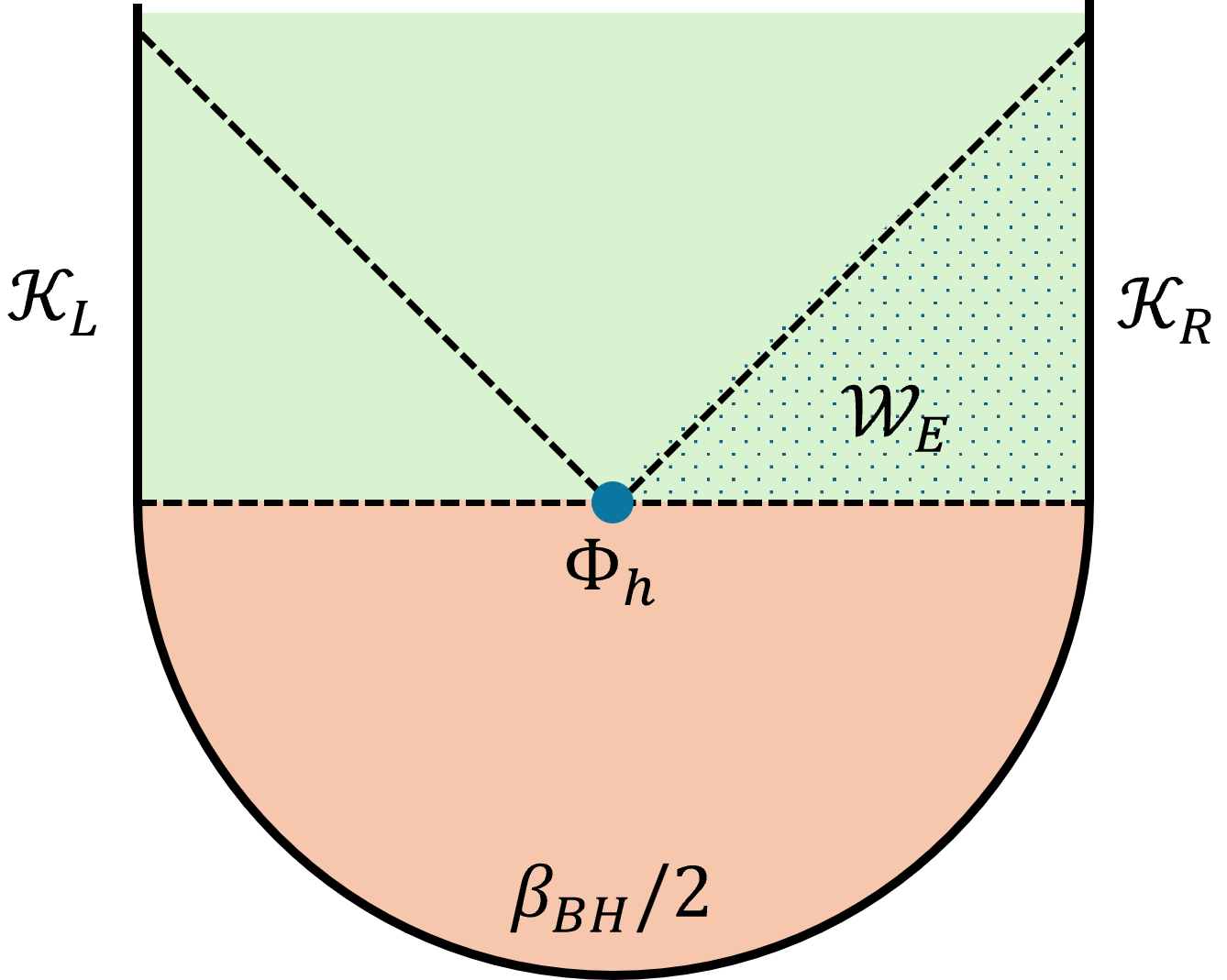}
         \caption{}\label{pic-jt-eow-3}
    \end{subfigure}
        \caption{(a) The Euclidean and Lorentzian spacetime of a single-sided black hole with an EoW brane. The red region is the bulk Euclidean region corresponds to the thermal EoW brane state $\ket{\mu,\b_0/2}$, where the dark red curve is the Euclidean EoW brane. The green region is the Lorentzian spacetime, where the dark green curve is the Lorentzian trajectory of the EoW brane. The blue dot is the horizon in the semiclassical JT limit. The shaded region behind horizon is the no man's island. The wedge outside of horizon (the dotted region) is the causal wedge $\mW_C$, which coincides with the entanglement wedge $\mW_E$. (b) The norm of the thermal EoW brane state, where $\phi$ is the angular Euclidean Rindler coordinate. (c) The canonical purification of $\mW_E$ as a thermofield double state with physical inverse temperature $\b_{BH}$.}
\end{figure}

For the semiclassical JT limit of DSSYK with an EoW brane, we should take the same viewpoint. In other words, we need to work in a limit where a state $\ket{\phi}$ has a fixed geometric dual $\mM_\phi$ and surviving operators in $\tilde \mA_R$ behave as if bulk fields living on $\mM_\phi$. Since the gravity and matter chords are treated in different manners, we will discard the equal footing notation in this subsection and take $H_{R}=H_{R,0}$ and $M_R=H_{R,1}$ (and the tilde versions) for convenience. 

Let us define the thermofield double state
\be 
\ket{\b_0/2}=e^{-\b_0 \tilde H_R/2}\ket{\w} \label{eq:4.9}
\ee
Due to the identity \eqref{3.37}, we also define the thermal EoW brane state
\be 
\ket{\mu,\b_0/2}=e^{-\b_0  H_R/2}\ket{\mB_\mu^{1/2}} \label{eq:4.10}
\ee
The norm of these two states are the same
\be 
\tilde Z(\b_0)=\avg{\w|e^{-\b_0 \tilde H_R}|\w}=\avg{\mB_\mu|e^{-\b_0 H_R}|\w}=\int\f{d\t}{2\pi}\f{(e^{\pm2\i\t},q;q)_{\infty}e^{-\b_{0}E(\t)}}{(\mu(1-q)^{1/2}e^{\pm\i\t};q)_{\infty}} \label{4.10}
\ee
where the computation is given in Appendix \ref{app:d2}. The semiclassical JT limit is defined such that the $\t$ integral of \eqref{4.10} is dominated by a saddle and this saddle should have a geometric interpretation as a single-sided black hole with an on-shell heavy EoW brane in JT gravity (see Figure \ref{pic-jt-eow}). As analyzed in Appendix \ref{app:d2}, the appropriate definition for the semiclassical JT limit contains two separate scalings 
\be 
\lam,\d\ra0,~ \d/\lam\ra+\infty,~\b_{0}=\f{\b}{\d\sqrt{1-q}},~\mu=-\f{e^{-\tilde \mu\d/2}}{\sqrt{1-q}},~\t=\pi-\d k/2,~ k\sim O(1) \label{semiJT}
\ee
This limit can be understood as two steps.\footnote{An alternative way to understand this limit is first define $\phi_b=\d/(2\lam)$ and take $\d,\lam\ra 0$ while keeping $\phi_b$ fixed. This is equivalent to the quantum JT gravity coupled to matter in \cite{Kolchmeyer:2023gwa} as shown in Appendix \ref{app:d6}. Then we take the semiclassical limit with $\phi_b\ra+\infty$.} For $\d\sim O(1)$ and $\lam\ra 0$, \eqref{semiJT} is equivalent to the large $p$ limit \cite{Mukhametzhanov:2023tcg} with an heavy EoW brane $\mu_r\sim O(1/\lam)$, where the norm \eqref{4.10} is dominated by a saddle but with finite temperature effect; then we take $\d\ra0$ limit to cool down to low temperature and tune the EoW brane slightly lighter as $O(\d/\lam)$. We show in Appendix \ref{app:d1} and \ref{app:d2} that this semiclassical JT limit of $\ket{\b_0/2}$ can be precisely matched with the semiclassical JT gravity. It turns out that the thermal EoW brane state has a bulk dual of a single-sided black hole with inverse temperature
\be 
\b_{BH}=2\pi/k_s,\quad \pi-\b k_s/2=\arctan(k_s/\tilde \mu)
\ee
where $k_s$ is determined by the second saddle equation and is related to the extremal dilaton value of the black hole as 
\be 
\Phi_h=\d k_s/(2\lam )\sim O(\d/\lam)
\ee
which is analogous to the horizon area of a higher dimensional near extremal black hole.

While we take $\lam\ra 0$ limit, the gravity chord number $N_0$ should scale as $O(1/\lam)$ due to \eqref{3.13}, and the matter chord number $N_1$ is kept finite. Moreover, we will keep the crossing factor $r=q^\D$ with $\D\sim O(1)$. This means that the semiclassical JT limit also includes a projection to a subspace of $\mH$ with finite $N_1$. Given the general procedure of the infinite $N$ limit, we will consider the semiclassical JT limit of the pair $\{\tilde \mA_R,\ket{\b_0/2}\}$, which is denoted as $\{\tilde \mK_R,\ket{\phi_{\b_0/2}}\}$, where the state $\ket{\phi_{\b_0/2}}$ is an algebraic state that is proportional to $\ket{\b_0/2}$ but with a vanishing normalization $\mN_\lam\sim O(e^{-\d S/\lam})$ with $S>0$ such that $\avg{\phi_{\b_0/2}|\phi_{\b_0/2}}=1$ (see Appendix \ref{app:d2} for more details). 

The surviving algebra $\{\tilde \mK_R,\ket{\phi_{\b_0/2}}\}$ is indeed a GNS algebra, which is defined through correlation functions 
\be 
\avg{\phi_{\b_0/2}|\tilde \mO_1\cdots\tilde \mO_n|\phi_{\b_0/2}}=\lim_{\lam\ra0}\mN_\lam \avg{\b_0/2|\tilde \mO_1\cdots\tilde \mO_n|\b_0/2},\quad \tilde \mO_i\in\tilde \mK_R\subset\tilde \mA_R \label{4.10-2}
\ee
On the other hand, due to the identity \eqref{3.37}, where all expectation values of $\tilde\mO\in\tilde\mA_R$ in $\ket{\w}$ is one-to-one equivalent to the expectation values of $\mO\in\mA_R$ in $\ket{\mu,\b_0/2}$
\begin{align} 
&\avg{\phi_{\b_0/2}|\tilde \mO_1\cdots\tilde \mO_n|\phi_{\b_0/2}}=\lim_{\lam\ra0}\mN_\lam \avg{\mu,\b_0/2|\mO_1\cdots\mO_n|\mu,\b_0/2}\nn\\
\equiv&\avg{\phi^\mu_{\b_0/2}|\mO_1\cdots\mO_n|\phi^\mu_{\b_0/2}},\quad \tilde \mO_i\in\tilde \mK_R\subset\tilde \mA_R,\quad \mO_i\in\mK_R\subset\mA_R \label{4.11}
\end{align}
where $\ket{\phi^\mu_{\b_0/2}}$ is another algebraic state proportional to $\ket{\mu,\b_0/2}$ with the same divergent normalization $\mN_\lam$. Here $\mK_R$ is corresponding algebra of $\tilde \mK_R$ by the $\td$ isomorphism and acts on $\ket{\phi^\mu_{\b_0/2}}$. Indeed, \eqref{4.11} induces another GNS algebra $\{\mK_R,\ket{\phi^\mu_{\b_0/2}}\}$, which is isomorphic to the GNS algebra $\{\tilde \mK_R, \ket{\phi_{\b_0/2}}\}$ from the semiclassical JT limit.

In general, an isomorphism between two algebras may not be necessarily related to the same physics context. However, by our discussion in Section \ref{sec:3.2}, this isomorphism is special because it corresponds to two different bulk slicing of any chord diagram. Moreover, this isomorphism induces an unitary equivalence by Lemma \ref{lm4}, and identifies all correlation functions of boundary operators and also reconstructed bulk fields. As chord diagrams in the semiclassical JT limit can be regarded as the continuum bulk spacetime with a heavy EoW brane, and different bulk slicing should be regarded as the gauge redundancy of diffeomorphism of the geometry, it is natural to expect that the two GNS algebras $\{\tilde \mK_R, \ket{\phi_{\b_0/2}}\}$ and $\{\mK_R, \ket{\phi^\mu_{\b_0/2}}\}$ should be represented by the same bulk QFT in the same background of a single-sided AdS$_2$ black hole with a heavy EoW brane behind horizon. In particular, they should be dual to the bulk fields living in the same spacetime region, namely the entanglement wedge $\mW_E$ of the right boundary.

In the following, we will utilize this isomorphism to discuss the entanglement wedge of $\{\mK_R, \ket{\phi^\mu_{\b_0/2}}\}$. Note that the semiclassical JT limit \eqref{semiJT} has an inverse temperature scale  $O(1/(\d\sqrt{1-q}))$, the time evolution should be taken in the same scale\footnote{Different time scale of Lorentzian evolved operators may form different algebras (see e.g. \cite{Penington:2025hrc}).}
\be 
U_t=e^{-\i h_R t},\quad h_R=H_R/(\d \sqrt{1-e^{-\lam}}) \label{4.17}
\ee
Just like the single-trace operators surviving in the infinite $N$ limit, we should expect all Lorentzian time evolved matter operators $M_R(t)=e^{\i  h_R t}M_Re^{-\i h_R t}$, as an analog of the single-trace operators, form a subalgebra $\mK_{R,*}\equiv \{M_R(t)\}''\subset \mK_R$ (by the isomorphism, there is a corresponding subalgebra $\tilde \mK_{R,*}\subset\tilde \mK_R$). This GNS subalgebra is defined by the correlation functions
\be 
\avg{\phi^\mu_{\b_0/2}|M_R(t_1)\cdots M_R(t_n)|\phi^\mu_{\b_0/2}}=\lim_{\lam\ra0}\mN_\lam \avg{\mu,\b_0/2|M_R(t_1)\cdots M_R(t_n)|\mu,\b_0/2}
\ee
The matter correlation functions can be computed by summing over all Wick contractions of $M_R(t_i)$, which can be grouped into different matter chord diagrams. It turns out that only non-crossing matter chord diagrams survive in the semiclassical JT limit and they factorize into product of two-point functions as we show in Appendix \ref{app:d}. In other words, Lorentzian matter operators reduce to a generalized free field theory. For the two-point function, we have
\begin{equation}
\avg{\phi^\mu_{\b_0/2}|M_R(t_1)M_R(t_2)|\phi^\mu_{\b_0/2}}=\left(\f{(\d\pi/\b_{BH})}{\i\sinh(\f{\pi}{\b_{BH}}t_{12})}\right)^{2\D}
\end{equation}
where the overall vanishing factor $\d^{2\D}$ can be absorbed into the definition of $M_R(t)$. It follows that $\mK_{R,*}$ is just an algebra of generalized free fields in 0+1 dimension in a thermal state with inverse temperature $\b_{BH}$.

By subregion-subalgebra duality \cite{Leutheusser:2022bgi}, this ``single-trace" algebra $\mK_{R,*}$ is dual to the causal wedge $\mW_C$ (see Figure \ref{pic-jt-eow}). The bulk fields living in $\mW_C$ can be reconstructed from $\mK_{R,*}$ using the standard HKLL kernel \cite{Hamilton:2005ju,Hamilton:2006az}. It has been conjectured \cite{Jafferis:2015del,Faulkner:2017vdd} that the one-parameter extension of $\mK_{R,*}$ by the modular flow
\be 
\mK_{R,m}=\{\D_{\mu}^{\i s}\mK_{R,*}\D_{\mu}^{-\i s}|s\in\R\}'' \label{4.14}
\ee
is equivalent to the algebra $\mK_R$ dual to the entanglement wedge. Here $\D_{\mu}$ is the modular operator of the state $\ket{\phi^\mu_{\b_0/2}}$ for $\mK_R$.\footnote{By personal communication with H. Liu, there is a theorem in an unpublished work of T. Faulkner guarantees the modular extension if the state $\ket{\phi^\mu_{\b_0/2}}$ is cyclic and separating for $\mK_{R,0}$. This is true for our case because $\ket{\phi^\mu_{\b_0/2}}$ has a semiclassical bulk geometry, and it is cyclic and separating for the bulk fields in the causal wedge $\mW_C$ due to the Reeh–Schlieder theorem \cite{Witten:2018zxz} of bulk quantum field theory.} While the explicit form of the modular operator $\D_{\mu}$ of $\ket{\phi^\mu_{\b_0/2}}$ for $\mK_R$ is not obvious, its type II$_1$ version before the JT limit is easy to find as 
\be 
\D_{(\mu,\b_0/2)}=e^{-\b_0H_R}\mB_{R,\mu}\mB_{L,\mu}^{-1}e^{\b_0 H_L} \label{4.21}
\ee 
following a similar computation of Lemma \ref{lemma8}. Since the modular operator exists in all types of von Neumann algebra, we expect the matter operators $M_R(t)$ modular-flowed by $\D_{(\mu,\b_0/2)}$ has well-defined semiclassical JT limit, and coincide with \eqref{4.14}. In other words, supposing the semiclassical JT limit is interchangable with modular flow, we have\footnote{The double commutant is outside of the semiclassical JT limit because the Hilbert space and all bounded operators therein also have changed under the limit.}
\be
\mK_{R,m}=\left(\lim_{\lam\ra 0}\{\D_{(\mu,\b_0/2)}^{\i s}\mK_{R,*}\D_{(\mu,\b_0/2)}^{-\i s}|s\in\R\}\right)''
\ee
This method is very useful to study the modular flow of a system, which has a type I or type II completion because the modular operator is simple in the type I/II case. For example, it was applied for the modular reconstruction of a generic PETS dual to a long wormhole spacetime in JT gravity with AdS$_2$ background \cite{Gao:2024gky} because the JT plus matter theory allows a type II$_\infty$ quantum description \cite{Kolchmeyer:2023gwa,Penington:2023dql}.

In Appendix \ref{app:d4}, we compute the modular flowed two-point function as
\be 
\avg{\phi^\mu_{\b_0/2}|\D_{\mu}^{\i s} M_R(t_1)\D_{\mu}^{-\i s} M_R(t_2) |\phi^\mu_{\b_0/2}}=\left(\f{(\d\pi/\b_{BH})}{\i\sinh\left[\f{\pi}{\b_{BH}}t_{12}-\pi s\right]}\right)^{2\D}
\ee
which shows that the modular flow is essentially the same as the time evolution with a constant rescaling. Therefore, the modular extension is trivially and we have 
\be 
\mK_R=\mK_{R,m}=\mK_{R,*}
\ee
Since the causal wedge bulk QFT is type III$_1$, the full surviving right boundary algebra $\mK_R$ is not type I but type III$_1$ as well. From the bulk viewpoint, this implies that the entanglement wedge $\mW_E$ is exactly the same size as the causal wedge $\mW_C$. Correspondingly, the wedge behind the horizon and before the EoW brane cannot be reconstructed from the right boundary algebra $\mK_R$, and must belong to its commutant $\mK_R'$! Inspired from the island in black hole evaporation \cite{Almheiri:2019psf,Penington:2019npb}, where the island does not belong to the black hole but its commutant, the old Hawking radiation, we would like to call this wedge behind horizon as ``no man's island" to stress that there is no asymptotic boundary (since we only have one single boundary in the semiclassical JT limit) that can claim the physical observables therein (see Figure \ref{pic-jt-eow}).

\subsection{No man's island is for a single-sided black hole with entanglement} \label{sec:4.4}

The existence of the no man's island could be very surprising at the first glance because the black hole is single-sided and by ordinary AdS/CFT dictionary the only asymptotic boundary should claim everything in the black hole. However, the existence of no man's island implies a nontrivial commutant to the boundary operators surviving in the semiclassical JT limit. As we know at the type II$_1$ level that there is a nontrivial commutant $\tilde \mA_L$ to the right boundary algebra $\tilde \mA_R$ though it does not have a geometric picture as discussed in Section \ref{sec:4.1}. Nevertheless, we should expect some operators from $\tilde\mA_L$ also surviving in the semiclassical JT limit and they are dual to the bulk fields living on the no man's island. 

One way to understand these survivals from $\tilde \mA_L$ is to compute the matter correlation functions in $\tilde \mA_L$, and the two-sided correlation function between $\tilde \mA_L$ and $\tilde \mA_R$ in the thermofield double state $\ket{\b_0/2}$. The former can be rewritten as matter correlation functions in $\tilde\mA_R$ by noting 
\begin{align} 
\avg{\b_0/2|\tilde H_{L,i_1}\cdots \tilde H_{L,i_k}|\b_0/2}=&\avg{\psi|\tilde H_{R,i_k}\cdots \tilde H_{R,i_1}e^{-\b_0 \tilde H_{R,0}}\tilde \mB_{R,\mu}|\psi}\nn\\
=&\avg{\b_0/2|\tilde H_{R,i_k}\cdots \tilde H_{R,i_1}|\b_0/2}
\end{align}
where we used the tracial property of $\ket{\psi}$ for $\tilde \mA_R$. Therefore, the survival algebra of the pair $\{\tilde \mA_L, \ket{\b_0/2}\}$ is $\{\tilde \mK_L,\ket{\phi_{\b_0/2}}\}$, where $\tilde \mK_L$ is generated by all Lorentzian evolved matter operators
\be 
\tilde \mK_L=\{e^{\i \tilde h_L t}\tilde M_L e^{-\i \tilde h_L t}\}'',\quad \tilde h_L=\tilde H_L/(\d\sqrt{1-e^{-\lam}})
\ee
Under the isomorphism $\td$, we also have the pair $\{\mA_L,\ket{\mu,\b_0/2}\}$ reducing to $\{\mK_L,\phi^\mu_{\b_0/2}\}$ in the semiclassical JT limit, where $\mK_L$ is the same as $\tilde \mK_L$ with tilde removed. By the computation in the last section, this implies that $\mK_L$ is also a generalized free field in a thermal state with inverse temperature $\b_{BH}$. 

To understand the entanglement between $\mK_L$ and $\mK_R$ (or equivalently $\tilde \mK_L$ and $\tilde \mK_R$), we need to compute the two-sided two-point function
\begin{align} 
\avg{\phi^\mu_{\b_0/2}|M_L(t_1)M_R(t_2)|\phi^\mu_{\b_0/2}}=&\avg{\phi_{\b_0/2}|\tilde M_L(t_1)\tilde M_R(t_2)|\phi_{\b_0/2}}\nn\\
=&\lim_{\lam \ra 0} \mN_\lam \avg{\mu,\b_0/2|M_L(t_1)M_R(t_2)|\mu,\b_0/2}
\end{align}
where the first equality is due to \eqref{eq:17-2-2}. We compute it in Appendix \ref{app:d5}, and the result is very simple
\be 
\avg{\phi^\mu_{\b_0/2}|M_L(t_1)M_R(t_2)|\phi^\mu_{\b_0/2}}=\left(\f{(\d\pi/\b_{BH})}{\cosh\f{\pi}{\b_{BH}}(t_{1}+t_{2})}\right)^{2\D} \label{4.28}
\ee
which is exactly the left-right two-point function of thermofield double state with the black hole inverse temperature $\b_{BH}$.

This indicates that the left algebra $\{\mK_L,\ket{\phi^\mu_{\b_0/2}}\}$ is the canonical purification \cite{Engelhardt:2018kcs,Dutta:2019gen,Engelhardt:2022qts} of the right algebra $\{\mK_R,\ket{\phi^\mu_{\b_0/2}}\}$. In other words, we can equivalently understand $\mK_L$ and $\mK_R$ as the single-trace operators on each boundary in the thermofield double state
\be 
\ket{\text{tfd};\b_{BH}}\propto \sum_n e^{-\b_{BH}E_n/2} \ket{E_n}_L\otimes \ket{E_n}_R \simeq \ket{\phi^\mu_{\b_0/2}} \label{4.29}
\ee
This thermofield double state canonically purifies the entanglement wedge $\mW_E$ outside of the horizon of the single-sided black hole into a two-sided black hole (see Figure \ref{pic-jt-eow-3}). Geometrically, this purification is just extending the no man's island all the way to the left boundary. Therefore, the bulk fields on the no man's island can be reconstructed by $\mK_L$ using the standard HKLL kernel. 

Same logic applies to the tilde algebra, namely that $\{\tilde \mK_L,\ket{\phi_{\b_0/2}}\}$ is the canonical purification of $\{\tilde\mK_R,\ket{\phi_{\b_0/2}}\}$. However, from the the chord diagram viewpoint, all generators in $\mA_{L,R}$ and $\tilde \mA_R$ are simple but the generators in $\tilde \mA_L$ is complex (especially the matter operator $\tilde M_L$ that survives in the semiclassical JT limit) as discussed in Section \ref{sec:4.3}. This suggests a different interpretation for the canonical purification of the tilde algebras. 

In Section \ref{sec:4.3} and \ref{sec:3.2} we argue that changing perspective from $\tilde \mA_R$ to $\mA_R$ using the isomorphism $\td$ just corresponds to different bulk slicing and should not change the dual geometry. However, this slicing change picture does not apply for the relation between $\tilde \mA_L$ and $\mA_L$ as we know from Section \ref{sec:4.1} that generators of $\tilde \mA_L$ act non-geometrically (in the chord sense) but $\mA_L$ consists of simple operators acting on the left boundary of a two-sided system. Even though the correlation functions of $\tilde\mA_{L,R}$ in $\ket{\b_0/2}$ and those of $\mA_{L,R}$ in $\ket{\mu,\b_0/2}$ are the same due to the unitary equivalence $U$ induced by $\td$ in Lemma \ref{lm4}, we want to stress that their bulk geometric interpretations should still depend on how geometric quantities emerge from the discrete chord diagrams. 

This is analogous to the black hole radiation process with the emergence of island after Page time \cite{Almheiri:2019psf,Penington:2019npb}. In this case, we will have an old single-sided black hole and old radiations in a state $\ket{\Psi}$ with huge entanglement. Simple operators in the old radiation have non-thermal correlation with bulk fields $\chi_b$ in the entanglement wedge of the old black hole. However, since the old radiation claims the island, one can build a complex operator $\chi_c$ from the algebra of the old radiation such that it behaves like a bulk field living on the island. This complex operator $\chi_c$ should have thermal correlation with $\chi_b$. In this picture, $\chi_c$ is in the old radiation system and does not have a holographic boundary representation. Alternatively, we can implement a unitary $\mU$ that collapses all old Hawking radiations into a new black hole and shorten the Einstein-Rosen bridge to generate a thermofield double state with the old black hole, which is equivalent to the canonical purification of the latter \cite{Maldacena:2013xja,Engelhardt:2022qts}. Under this unitary transformation $\mU$, the complex operator is mapped to a simple operator $\chi_s$ on the boundary of the new black hole, the state $\ket{\Psi}$ is mapped to the thermofield double state, and the correlation between $\chi_s$ and $\chi_b$ in the thermofield double state is the same as the correlation between $\chi_c$ and $\chi_b$ in the original state $\ket{\Psi}$. 

\begin{figure}
	\centering
        \includegraphics[height=1.8cm]{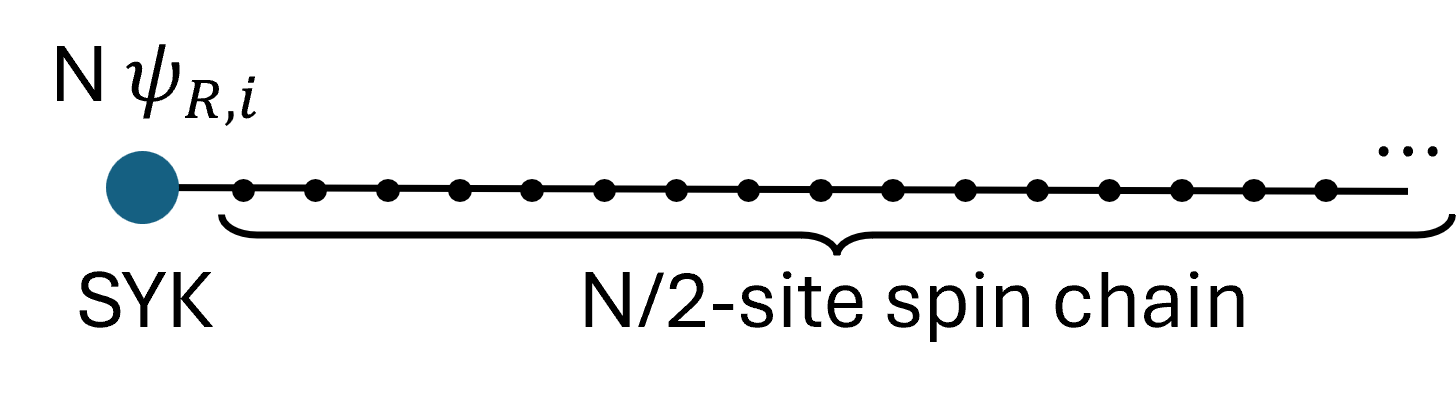}
        \caption{An SYK model with $N$ Majorana fermions coupled to a spin chain with $N/2$ sites as a boundary setup for the black hole radiation process. \label{pic-syk-SC}}
\end{figure}

The boundary description to above black hole radiation process using SYK model can be set up as follows. We can couple an SYK model of even $N$ Majorana fermions $\psi_{R,i}$ with a spin chain with $N/2$ sites. They both have a $2^{N/2}$ dimensional Hilbert space. The SYK model simulates a black hole, the spin chain simulates the radiation bath as shown in Figure \ref{pic-syk-SC}, and the coupling mediates a radiation process. After Page time, the spin chain has $O(N)$ entanglement with the SYK model and will claim an island in the dual bulk spacetime. Then there exists the above operator $\chi_c$ written in the Pauli basis of the spin chain in a complex way. On the other hand, we can use the Jordan-Wigner transformation to unitarily map Pauli operators on the spin basis to Majorana basis 
\be 
\psi_{L,2i-1}=\G_R\s_{1,z}\cdots\s_{i-1,z}\s_{i,y},\quad \psi_{L,2i}=\G_R \s_{1,z}\cdots \s_{i-1,z}\s_{i,x},\quad i=1,\cdots,N/2 \label{4.30}
\ee
where $\G_R=\i^{N/2}\psi_{R,1}\cdots \psi_{R,N}$ is added to guarantee anti-commutation with $\psi_{R,i}$. It follows that $\chi_c$ is a complex operator in terms of the $N$ Majorana fermions $\psi_{L,i}$. In the double scaled limit, the Majorana fermion basis is ensemble averaged by \eqref{eq:2.5} and we should consider the chord basis as the same level of complexity as the Majorana basis. In this sense, the complex operator $\tilde M_L$ should be identified with $\chi_c$. This gives a simple interpretation of the non-geometric feature of $\tilde \mA_L$ that it corresponds to a complex operator in the old Hawking radiation. It follows that we can understand the unitary equivalence $U$ as the complex unitary $\mU$ in the last paragraph to transform the radiation system into the canonical purification of the SYK model or the old black hole in a thermofield double state \eqref{4.29}. In other words, it encodes the bulk information of the no man's island into the old radiation in a complex way. An interesting difference is that $\mU$ is supposed to be an unitary acting on the radiation or spin chain system only (if we ignore $\G_R$ for the anti-commutation with $\psi_{R,i}$) but $U$ acts on both $\tilde\mA_L$ and $\tilde \mA_R$, though its action on $\tilde\mA_R$ can be understood as a bulk gauge redundancy of changing slicing. Probably one reason for this is that above finite $N$ model has factorized Hilbert space but the DSSYK has taken infinite $N$ limit and does not have a factorized Hilbert space. 

With this understanding, the no man's island has the same essence as the island of the black hole radiation process. In a generic sense, the boundary algebra $\tilde \mA_R$ corresponds to a single-sided black hole with a huge amount of entanglement with another system, which may not have a geometric or gravitational description.

\subsection{A different boundary algebra for a pure state single-sided black hole without entanglement} \label{sec:4.5}

Besides the EoW brane by explicit construction in the DSSYK, it has been argued in \cite{Kourkoulou:2017zaj} that projections into simple basis on one side of a thermofield double state will generate an EoW brane behind the horizon of the other side. The projection and the initial thermofield state is just an auxiliary preparation for a specific pure state in a single SYK model, which is dual to an intrinsic single-sided black hole in a pure state without entanglement. The operation in \cite{Kourkoulou:2017zaj} can be briefly summarized as follows.

Consider a thermofield double state of two SYK models with fermions $\psi_{L,i}$ and $\psi_{R,i}$. The entanglement wedge of each SYK model is the same as the the causal wedge of each side. Taking the same Jordan-Wigner transformation \eqref{4.30}, apply projections into the eigen basis of $\s_{i,z}=\i\psi_{L,2i-1}\psi_{L,2i}$ for each $i$, which gives a sequence $s=\{s_i\}$ of measurement outcomes $s_i=\pm$ for $i=1,\cdots,N/2$. Since this projection is in a high energy basis, it produces a high energy shockwave from the left boundary and behind the right horizon, which can be understood as an EoW brane. After projection, the state becomes a product state $\ket{s}\ket{B_s}$ between two SYK models and has zero entanglement. The pure state $\ket{B_s}$ is dual to a single-sided black hole without entanglement with another system and with an EoW brane behind the horizon.

On the other hand, it was shown in \cite{Kourkoulou:2017zaj} that such a pure state $\ket{B_s}$ is a thermal state for boundary two-point functions of $\psi_{R,i}$ with the same fermion indices and independent on the measurement outcome $s_i$. Therefore, using the HKLL reconstruction, one can reconstruct the region outside of the horizon. On the other hand, the two-point function of $\psi_{R,2i-1}$ and $\psi_{R,2i}$ is the product of two thermal two-point functions with coefficient $s_i$. Measuring the sign of all these off-diagonal two-point functions, one can infer the data $s_i$. With the data of $s_i$ and all $\psi_{R,i}$ fermions, it was argued in \cite{Almheiri:2018xdw} that the region between the EoW brane and the horizon, namely the no man's island, can be reconstructed fully from right boundary in terms of a quantum error correction code. Indeed, suppose the no man's island has smooth geometry in $\ket{B_s}$, then it is in the code subspace $\mH_{code}$ and the projection on the left keeps $\mH_{code}$ invariant. Therefore, the projection can be understood as an erasion error, against which the bulk fields on the no man's island is protected by the general feature of quantum error correction \cite{Harlow:2016vwg}. This seems to be a contradiction to our statement that no man's island cannot be reconstructed from the right boundary.

However, as we explained at the end of Section \ref{sec:4.2} that there are two different ways to define the boundary algebra of the single-sided black hole with or without the exponential wormhole length operator $e^{-\bar \l _b}$ independently.\footnote{We thank Hong Liu for discussing on this point.} If we have this operator, we can reconstruct all bounded operators, which in the semiclassical JT limit definitely include all bulk fields acting on the no man's island of state $\ket{\b_0/2}$ in \eqref{eq:4.9} (or equivalently the thermal EoW brane state $\ket{\mu,\b_0/2}$ in \eqref{eq:4.10}). It would be interesting to write down an explicit reconstruction map for no man's island from the boundary using the semiclassical JT version of $e^{-\bar \l _b}$. On the other hand, we do not need $e^{-\bar \l _b}$ to reconstruct the outside horizon region. Since $\tilde \mA_R$ is a subalgebra of $\mB(\mH)$ and the whole bulk contains no man's island, we can view this relation as a generalized case of the subregion-subalgebra duality \cite{Leutheusser:2022bgi}. 

As we stressed in the introduction, the alternative definition of the boundary algebra by including $e^{-\bar \l _b}$ to $\tilde \mA _R$ can be viewed as characterizing a different boundary theory. Since this operator helps reconstruct the no man's island, the boundary algebra with  $e^{-\bar \l _b}$ corresponds to a single-sided black hole in a pure state without entanglement with another system, unlike the case in the previous subsection. Therefore, the ambiguity of the boundary algebra definition is indeed a nice feature to describe two types of single-sided black holes in a unified framework.

In the following, we will give a physical explanation why $e^{-\bar \l _b}$ is essential to reconstruct the no man's island and make a concrete connection to \cite{Kourkoulou:2017zaj}. If we write $\bar\l_b$ in terms of fermion operators of the left and right SYK model, it becomes the size operator  \cite{Qi:2018bje,Lin:2023trc}
\be 
\bar \l_b= \f{\lam}{2p}\sum_{i=1}^{N}(1+ \i\psi_{L,i}\psi_{R,i}) \label{eq:4.31}
\ee
One may question why $e^{-\bar \l _b}$ can be regarded as a right boundary operator in this context. It has been shown in \cite{Gao:2019nyj} that we can write $e^{\i g \bar \l _b}$ as the sum over measurements on some left Pauli basis followed by a measurement-based unitary transformation on the right, namely the teleportation protocol. Therefore, in the context of \cite{Kourkoulou:2017zaj} where the left SYK model has been fully projected, the exponential wormhole length operator $e^{-\bar \l _b}$ can be regarded as a right boundary operator. Indeed, it has been proposed in \cite{Kourkoulou:2017zaj} that we can use a measurement-based eternal traversable wormhole evolution with a new Hamiltonian \cite{Maldacena:2018lmt}
\be 
H_{new}=H_{SYK}+\i \e \sum_{i=1}^N  s_i\psi_{R,2i-1}\psi_{R,2i} \label{eq:5.5}
\ee
to reconstruct the region behind the horizon. One can choose an appropriate sign of $\e$ such that the second term of \eqref{eq:5.5} is equivalent to the two-sided coupling (but slightly different from \eqref{eq:4.31}) between two SYK models, which produces a negative energy that generate a traversable wormhole \cite{Gao:2016bin,Gao:2019nyj,Maldacena:2017axo,Maldacena:2018lmt}. The new evolution allows the boundary trajectory extends beyond the original causal wedge and let it see the region behind the horizon. In this sense, adding $e^{-\bar \l _b}$ to the boundary algebra can be regarded as allowing a universal or averaged version of the measurement-based single-sided operation, which is essential for reconstructing the no man's island. If we disallow such an operation in the definition of the boundary algebra, we will only have $\tilde \mA_R$ and $\tilde \mK_R$, which cannot claim the no man's island.

Another technical remark is that an analogy to \eqref{eq:5.5} as a new Hamiltonian for evolution is very similar to the EoW brane Hamiltonian $\tilde H_{R,0}$ but with a different parameter range\footnote{We thank Ahmed Almheiri for discussing on this point.}
\be 
\hat H_{R,0}=H_{R,0}+\mu_* q^{N_0}r^{N_1},\quad \mu_*=-\f{e^{-(\mu_r+1/2)\lam}}{\sqrt{1-e^{-\lam}}},\quad \mu_r+1/2<0
\ee
Note that in the EoW brane setup of $\tilde H_{R,0}$, we require $\mu_r+1/2>0$ and the big $q$-Hermite polynomials have continuous spectrum and obey the orthogonality \eqref{eq:3.10}. However, with $\mu_r+1/2<0$ the big $q$-Hermite polynomials have another branch of discrete eigenvalues additional to the continuous branch \cite{koekoek2010hypergeometric}
\begin{align}
\int_{0}^{\pi}\f{d\t}{2\pi}\f{(e^{2\i\t},e^{-2\i\t},q;q)_{\infty}}{(ae^{\i\t},ae^{-\i\t};q)_{\infty}(q;q)_{n}}H_{m}(\cos\t;a|q)H_{n}(\cos\t;a|q)&\nn\\
+\sum_{\substack{k\\ 1\leq a q^k\leq a}} w_k \f{(q;q)_\infty}{(q;q)_n} H_{m}(x_k;a|q)H_{n}(x_k;a|q)& =\d_{nm} \\
\quad x_k=\f{a q^k+(a q^k)^{-1}}{2},\quad w_k=\f{(a^{-2};q)_\infty(1-a^2 q^{2k})(a^2;q)_k}{(q;q)_\infty(1-a^2)(q;q)_k}& q^{-3k^2/2-k/2}(-1/a^4)^k
\end{align}
where $a>1$ and $x_k$ is the discrete eigenvalues additional to $\cos \t$. We expect these discrete eigenvalues are the DSSYK version of the discrete bound state of Morse potential with negative $\mu_r+1/2$. It has been shown in \cite{Maldacena:2018lmt} that these discrete energy levels are eternal traversable wormhole states and if the temperature is low enough the evolution with $\hat H_{R,0}$ should be mostly dominated by these discrete energy levels. In this way, the boundary matter operator $H_{R,1}$ can evolve beyond the causal wedge and see the no man's island.

\section{Conclusion and discussion}
\label{sec:discussion}

In this paper we studied a single-sided black hole with an EoW brane behind the horizon in the DSSYK model. We proposed chord diagrammatic rules which reproduces earlier results in~\cite{Okuyama:2023byh}. Building on this, we have constructed the bulk Hilbert space and the boundary algebra for the single-sided black hole. We provide three different perspectives to understand the EoW brane in the DSSYK model. The first way is to deform the boundary Hamiltonian in a way such that it reduces to Morse potential in the JT limit. This deformation is chosen in a simple way such that it allows an exact diagonalization in terms of the big $q$-Hermite polynomials. Together with the matter operator, they form the boundary algebra $\tilde \mA_R$ for a single-sided black hole. The second way is to write an EoW brane as a $q$-coherent state in the original Hilbert space of DSSYK. We prove an isomorphism between the boundary algebra of each perspective to show that these two pictures are equivalent. This isomorphism between algebras further induces an unitary equivalence $U$, which leads to a surprising conclusion: the boundary algebra $\tilde \mA_R$ of a single-sided black hole is still a type $\text{II}_1$ von Neumann factor. The third way is to construct a family of matter-brane states, which includes some matter chords but behaves in the same way as the EoW brane $q$-coherent state if probed by boundary Hamiltonian. Such matter-brane states correspond to EoW branes with various boundary conditions for the bulk matter fields in the dual spacetime. As a byproduct, these matter-brane states give exact diagonalization of the original DSSYK Hamiltonian and solve its spectrum in the full Hilbert space in an elegant way.

Given the type II$_1$ boundary algebra $\tilde \mA_R$, we have attempted a few different ways to understand its commutant. The matter generator $\tilde H_{L,1}$ of the commutant $\tilde \mA_L$ cannot be written in terms of some simple operators on chord number basis, which implies that $\tilde \mA_L$ is ``non-geometric" if we regard chord basis as a kind of discretized geometry. On the other hand, to reconstruct all bounded operators from $\tilde\mA_R$ we just need to include a simple exponential wormhole length operator $e^{-\bar \l _b}$. Taking the semiclassical JT limit, we utilize the unitary equivalence $U$ to argue that the right boundary algebra reduces to $\tilde \mK_R$ or $\mK_R$, which only claims the spacetime region outside the horizon. The region between the horizon and the EoW brane cannot be reconstructed by the boundary of the single-sided black hole and we name it as "no man's island". It is claimed by the commutant $\tilde \mK_L$ or $\mK_L$, which is the canonical purification of $\tilde \mK_R$ or $\mK_R$. We explain that the no man's island is essentially the same as the island in the black hole radiation process after Page time and interpret $U$ as the complex encoding map of the island into the old radiation. The boundary algebra $\tilde \mA_R$ corresponds to a single-sided black hole with a huge amount of entanglement with another system, which may not have a geometric or gravitational description.

On the other hand, we can define the boundary algebra of a single-sided black hole differently by including the exponential wormhole length operator $e^{-\bar \l _b}$. This algebra is the full algebra $\mB(\mH)$ and claims the no man's island in the semiclassical JT limit. Since this operator can be written as the coupling term of the eternal traversable wormhole, we can interpret it as an averaged version of the measurement-based single-sided operation in \cite{Kourkoulou:2017zaj} to reconstruct the no man's island in a pure state of a single SYK model. This alternative boundary algebra corresponds to a single-sided black hole in a pure state without entanglement with another system. This ambiguity of the boundary algebra definition is a nice feature to describe two types of single-sided black holes in a unified framework.

The followings are discussions and future directions.

\subsubsection*{Trumpet Amplitudes and Divergence}
\begin{figure}
	\centering
	\begin{subfigure}{0.3\textwidth}
		\includegraphics[width=\textwidth]{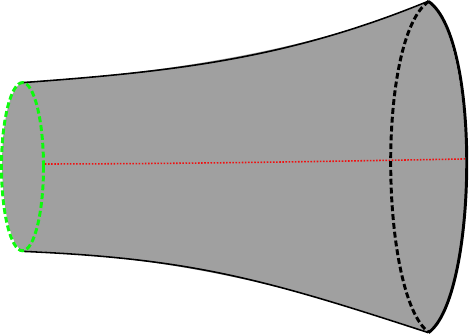}
        \captionsetup{format=hang}
        \caption{}
	\end{subfigure}
        \hspace{5em}
	\begin{subfigure}{0.3\textwidth}
		\includegraphics[width=\textwidth]{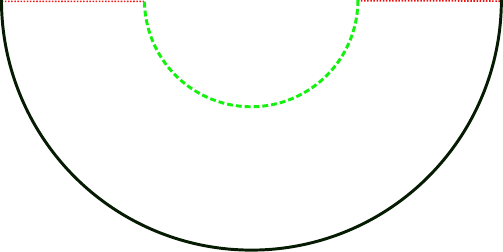}
        \captionsetup{format=hang}
        \caption{}
        \end{subfigure}
        \caption{By slicing open a trumpet geometry in the presence of an EoW brane as shown in (a), we get a two-sided geometry with two disconnected Cauchy surfaces, which we interpret as two single-sided black holes entangled across a pair of EoW branes.}
        \label{cut a trumpet}
\end{figure}

The one-sided chord diagrammatics introduced in Section~\ref{sec:EoW brane} can be naturally generalized to define trumpet amplitudes in the DSSYK model. In JT gravity trumpet amplitudes are divergent. In the pure gravity case they diverge as $\lim_{b\rightarrow 0}\frac{1}{b}$ where $b$ is the length of the EoW brane, while in the presence of matter, the divergence becomes exponential as $\lim_{b\rightarrow 0}e^{\frac{1}{b}}$. We will see that, using our diagrammatic rules, we can write the DSSYK trumpet amplitude as a sum of chord diagrams defined on a disk with a topological defect, and they have similar divergences as in JT gravity as a result of proliferating vacuum bubbles in these diagrams. We propose schemes to regularize these amplitudes by removing vacuum bubbles. By slicing open a regularized trumpet we can then define two disconnected single-sided blackholes entangled across a pair of EoW branes as shown in Figure~\ref{cut a trumpet}. They can be interpreted as thermofield double states with disconnected bulk geometry. We will focus on the pure gravity (a single type of chord) case in this paper and leave discussions for multiple chords to future work.

In the pure gravity case, a naive trumpet amplitude with boundary operator insertion $\hat O$ which is a polynomial in $H$ (we will simply write $H$ for $\tilde H_{R,0}$ in the following for simplicity) can be written as the sum of all diagonal elements of the operator $H^k$ in the chord number basis $\ket{n}$: 
\begin{equation}\label{eq:5.1}
     \Tr(\hat O)=\sum_{n} O_{nn}
\end{equation}
we use $\Tr$ to distinguish this naive trace from the regularized ones to be defined later. Diagrammatically, this is just the sum over chord diagrams as shown in Figure~\ref{diagonal-element}, which can be wrapped up to obtain chord diagrams on a disk with a defect in Figure~\ref{diagonal-element-warpped}. The presence of the defect changes the topology of the disk. As a result we have to take into consideration non-trivial windings and self-interactions of chords.  

\begin{figure}
	\centering
	\begin{subfigure}{0.4\textwidth}
		\includegraphics[width=\textwidth]{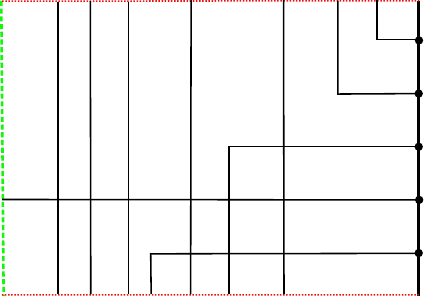}

        \caption{}
        \label{diagonal-element} 
	\end{subfigure}
        \hspace{5em}
	\begin{subfigure}{0.3\textwidth}
		\includegraphics[width=\textwidth]{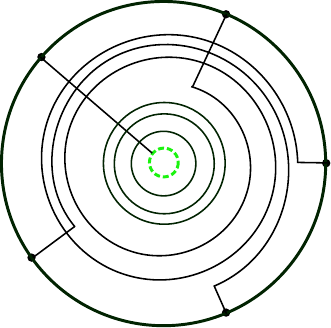}
        \captionsetup{format=hang}
        \caption{}
        \label{diagonal-element-warpped}
        \end{subfigure}

        \caption{Examples of chord diagrams representing diagonal elements. Here (a) contributes to the diagonal matrix element $(H^5)_{77}$ ($H$ is short for $\tilde H_{R,0}$ here). By identifying the top and bottom boundaries we get the chord diagrams on a disk with a defect as shown in (b). }
\end{figure}

Clearly the naive sum \eqref{eq:5.1} diverges. This is a result of the fact that monomials $H^k$ are non-compact operators in the sense that their matrix elements do not decay as $n$ becomes larger. The fact is reflected in the diagrammatics through the emergence of vacuum bubbles, which are defined to be the collection of chords disconnected from the boundary. For example, in Figure~\ref{diagonal-element-warpped} the vacuum bubble is the three concentric loops floating in the middle of the diagram. In Appendix~\ref{app:f} we included a detailed analysis of the structure of vacuum bubbles, where we prove that proliferating vacuum bubbles leads to a linear divergence. There are multiple ways to regularize this divergence and obtain renormalized trumpet amplitudes satisfying linearity, cyclicity and positivity.\footnote{These are the properties of a trace function, the reason we demand them is that we expect a GNS construction using this state.}

Firstly, we can formally divide by the divergent trace of the unit operator $\Tr(\mathbb I)$. This amounts to delete all diagrams either with vacuum bubbles or EoW brane chords (chords connected to the EoW brane). This regularization leads to the following trace 
\begin{equation}
    \tr[f(H)]=\int_0^\pi\frac{d\theta}{\pi}f[E(\theta)]
\end{equation}
where $f(H)$ is an arbitrary function of $H$ inserted on the boundary. Although satisfying the definition of a trace, this regularization seems too coarse in the sense that it eliminates all information about the EoW brane. A more deliberate regularization can be achieved by simply removing all vacuum bubble diagrams, which leads to the regularized trace agreeing with results in~\cite{Okuyama:2023byh}  
\begin{equation}
    \widetilde {\tr}[f(H)]=\int_0^\pi \frac{d\theta}{2\pi}f[E(\theta)]\frac{1-\tilde\mu\cos\theta}{1-2\tilde\mu\cos\theta+\tilde\mu^2}
\end{equation}
where $\tilde \mu=\mu\sqrt{1-q}$. Proofs of these formulas are again included in Appendix~\ref{app:f}. In fact, in the pure gravity case there are infinitely many ways to regularize the trumpet amplitude by integrating $f[E(\t)]$ with positive measures. However most of these regularization lacks a clear geometric representation as the ones defined above. Different trace choices lead to different construction of a type I$_\infty$ algebra. This is consistent with the discussion in \cite{Xu:2024hoc}.

The divergence of trumpet amplitudes persists in the multiple-chord case and becomes more severe. In Appendix~\ref{app:f} we proved that in this case, since each chord in the vacuum bubble can be either a gravitational chord or a matter chord, the divergence from vacuum bubbles becomes exponential. Due to this exponential divergence and the non-Abelian nature of the boundary algebra, the regularizations discussed above are no longer valid in this case. One reasonable choice is to first truncate operators $\tilde H_{R,i}$ in the subspace with a fixed total chord number $K=N_0+N_1$, then normalized by the regular trace of a unit operator in this subspace, and take $K$ to infinity in the end
\be \label{eq:5.2}
\tau(f(\tilde H_{R,i}))=\lim_{K\ra\infty}\Tr(f(P_{N_0+N_1=K}\tilde H_{R,i}P_{N_0+N_1=K}))/\Tr(P_{N_0+N_1=K})
\ee
where $f(\tilde H_{R,i})$ is a polynomial of $\tilde H_{R,i}$ and $P_{N_0+N_1=K}$ is the projection to the subspace $\mH_{N_0,N_1}$ with $K=N_0+N_1$. In the semiclassical JT gravity view point, this is the regularization of the trumpet with fixed distance (if the matter chord number is $O(1)$) between the EoW brane and the asymptotic boundary.\footnote{We thank Daniel Jafferis for discussing on this point.} This regularization $\tau$ should define a positive linear function qualified as a trace if we quotient out some null operators that map to zero under $\tau$. Note that this is a completely different choice of trace function from $\avg{\w|\cdot|\w}$, which is the conventional choice for DSSYK. Therefore, we expect this new trace $\tau$ leads to different dynamics and algebra properties, which will be left for future works. 

\subsubsection*{Bulk dual of matter-brane states}

In Section \ref{sec:3.4}, we have constructed an overcomplete set of states for $H_{R,0}$ that behaves identically as a $q$-coherent state. Therefore, we may understand them as various types of EoW branes in DSSYK model. On the other hand, matter-brane states can contain arbitrary numbers of matter chords, which can be probed by matter operators $H_{R,1}$ because it can contract with the matter chords in the matter-brane states. In the JT gravity language, this means that the matter-brane states are EoW branes with different boundary conditions to bulk matter fields. The matter brane can absorb and emit bulk matter fields with non-trivial amplitudes. These matter-brane states are analogous to EoW branes by various types of projections in \cite{Kourkoulou:2017zaj,Almheiri:2018xdw}, where the off-diagonal boundary matter correlators are non-thermal but depends on the state of the EoW brane. Since the boundary matter correlation is different from that in the $q$-coherent state, many conclusions in this paper about the semiclassical JT limit of the EoW brane could potentially have some changes. In particular, we may ask to what extent the no man's island will be modified.

\subsubsection*{The quantum JT gravity counterpart for the single-sided black hole}
We can think of the DSSYK model as a UV completion of JT gravity. For example, by taking the JT limit of the DSSYK Hamiltonians we can restore Liouville or Morse Hamiltonians describing dynamics in JT gravity with different boundary conditions as we show in \eqref{2.47} and \eqref{3.14}. In the quantum regime, we can systematically include conformal matter fields and couple them with the Schwarzian mode from the quantized JT gravity. For this bottom-up construction, it has been shown in \cite{Penington:2023dql,Kolchmeyer:2023gwa,Penington:2024sum} that the boundary algebra of quantum JT gravity plus conformal matter is a type II$_\infty$ factor. In the main scope of this paper, we focus on either the type II$_1$ level or the semiclassical JT level, where the former is more ultraviolet than quantum JT gravity and the latter is more infrared than quantum JT gravity. To fit the quantum JT gravity into the frame work of DSSYK, we need to take an intermediate scaling, which is briefly mentioned in Appendix \ref{app:d6} that we can fix the ratio $\phi_b=\d/(2\lam)$ in \eqref{semiJT} while taking $\d,\lam\ra 0$ limit. We show in Appendix \ref{app:d6} that the thermal EoW brane state $\ket{\mu,\b_0/2}$ has the same inner product integral with the quantum JT result but more comprehensive examination is required in the future to track how the type II$_1$ algebra reduces to a type II$_\infty$ algebra. Since the matter operators in DSSYK is pairwise connected by a chord, the conformal matter in the quantum JT regime should be limited to generalized free fields. If this is the case, we expect that the no man's island should still exist in the JT version of the single-sided black hole with an EoW brane.

The top-down construction from DSSYK to quantum JT gravity has many advantages. While it is not obvious in prior how to define the quantum EoW brane state in JT gravity, the UV complete DSSYK model gives a simple candidate for the former. Moreover, from the elaborate definition of the EoW brane in Section \ref{sec:2}, we find a nontrivial unitary equivalence $U$ that maps between two algebras in a non-geometric way. This is very nontrivial from the bottom-up viewpoint of quantum JT gravity and worth understanding how it acts at the quantum JT scale if it survives in the aforementioned limit.

\subsubsection*{Towards an algebraic Page phase transition}
In~\cite{Penington:2019kki,Almheiri:2019qdq} the Page transition of an evaporating black hole is identified using a replica trick calculation. As proposed in ~\cite{Engelhardt:2023xer}, the boundary algebraic structure is different before and after the Page time. As discussed in Section \ref{sec:4.4}, the algebra $\tilde \mA_R$ for a single-sided black hole with an EoW brane in the DSSYK could be understood as the phase after Page time, where the no man's island is analogous to the island of the old radiation. Somehow, the extended algebra with $e^{-\bar{\l}_b}$ is similar to a single-sided black hole at the beginning or early stage of the radiation. However, a concrete infinite $N$ holographic model that realizes the dynamical phase transition in the boundary algebra formalism is still in lack. In particular, if we expand around the Page time, we should have a one-parameter family of von Neumann algebras to interpolate an algebraic type transition from type I to type III while undergoing probably a type II in the intermediate regime \cite{Engelhardt:2023xer}. Such a one-parameter family of von Neumann algebras will deepen our understanding on the entanglement phase transition in an accurate way.

\acknowledgments

We thank Ahmed Almheiri, Andreas Blommaert, Thomas Faulkner, Daniel Jafferis, Marius Junge, David Kolchmeyer, Nima Lashkari, Adam Levine, Hong Liu, Xiao-liang Qi, Zhencheng Wang, Jiuci Xu for useful discussions. PG is grateful for the NHETC of Rutgers University and the Institute of Advanced Studies, where part of this work was done. PG is supported by the Fundamental Research Funds for the Central Universities. XC is supported by the DOE award number DE-SC0015655.

\appendix

\section{Transformation from matter-brane basis the chord number basis} \label{app:a}
In the following, we show that all matter-brane states also form a overcomplete basis of $\mH$ by working out the inverse transformation to the chord number (binary string) basis. For one-matter chord states, by \eqref{eq:11} we have 
\begin{equation}
\ket{0^{n_{1}}10^{m}}=\f{[n_{1}]!_{q}}{2\pi\i}\oint\f{d\a}{\a^{n_{1}+1}}\ket{(\a);0^{m}}
\end{equation}
For two-matter chord states, using (\ref{eq:43}), we can write a contour
integral
\begin{align}
\ket{(\a_{2},\a_{1});0^{m}} & =\oint_{|z|>1}\f{dz}{2\pi\i}\sum_{n=0}^{\infty}\f{(\a_{2}r/\a_{1})^{n}}{z(1-q^{n}/z)(q;q)_{n}}\ket{\mB_{\a_{2}}1\mB_{z\a_{1}}10^{m}}\nonumber \\
 & =\oint_{|z|>1}\f{dz}{2\pi\i}\sum_{k,n=0}^{\infty}\f{q^{nk}(\a_{2}r/\a_{1})^{n}}{z^{k+1}(q;q)_{n}}\ket{\mB_{\a_{2}}1\mB_{z\a_{1}}10^{m}}\label{eq:93-1}
\end{align}
Suppose the inverse integral transformation is in the following form
\begin{align}
&\ket{\mB_{\a_{2}}1\mB_{\a_{1}}10^{m}}=  \oint_{|\a_1|<|y|<1}\f{dy}{2\pi\i}K(\a_i;y)\ket{(\a_{2},\a_{1}/y);0^{m}}\nonumber \\
= & \sum_{k,n=0}^{\infty}\oint_{|\a_1|<|y|<1}\f{dy}{2\pi\i}K(\a_i;y)\oint_{|z|>1}\f{dz}{2\pi\i}\f{q^{nk}y^{n}(\a_{2}r/\a_{1})^{n}}{z^{k+1}(q;q)_{n}}\ket{\mB_{\a_{2}}1\mB_{z\a_{1}/y}10^{m}}\nonumber \\
= & \sum_{k,n=0}^{\infty}\oint_{|z|>1/|y|}\f{dz}{2\pi\i}\oint_{|\a_1|<|y|<1}\f{dy}{2\pi\i}K(\a_i;y)\f{q^{nk}y^{n-k}(\a_{2}r/\a_{1})^{n}}{z^{k+1}(q;q)_{n}}\ket{\mB_{\a_{2}}1\mB_{z\a_{1}}10^{m}}\label{eq:94}
\end{align}
where in the last step we shift $z\ra zy$. Assume the kernel $K(\a_i;y)$ is expanded
as power series of $y$
\begin{equation}
K(\a_i;y)=\sum_{s=0}^\infty\kappa_{s}y^{s-1}\label{eq:95}
\end{equation}
and the integral over $y$ in (\ref{eq:94}) leads to\footnote{The contour integral over $y$ can be done as a angular integral with fixed magnitude $|y|$, which leads to $n=k-s\geq 0$. This makes the integral over $z$ with some fixed magnitude $|z|>1$.}
\begin{equation}
\ket{\mB_{\a_{2}}1\mB_{\a_{1}}10^{m}}=\sum_{k=0}^{\infty}\sum_{s=0}^{k}\oint_{|z|>1}\f{dz}{2\pi\i}\kappa_{s}\f{q^{k(k-s)}(\a_{2}r/\a_{1})^{k-s}}{z^{k+1}(q;q)_{k-s}}\ket{\mB_{\a_{2}}1\mB_{z\a_{1}}10^{m}}\label{eq:96-1}
\end{equation}
We require
\begin{equation}
\sum_{s=0}^{k}\kappa_{s}\f{q^{k(k-s)}(\a_{2}r/\a_{1})^{k-s}}{(q;q)_{k-s}}=1\label{eq:97}
\end{equation}
and the sum over $k$ in (\ref{eq:96-1}) gives $1/(z-1)$, which
leads to $z=1$ after the $z$ contour integral (assuming $|z|>1$).
Indeed, (\ref{eq:97}) can be solved by
\begin{equation}
\kappa_{s}=\sum_{p=0}^{s}\f{(-\a_{2}rq^{s}/\a_{1})^{p}q^{-p(p-1)/2}}{(q;q)_{p}}
\end{equation}
and the kernel (\ref{eq:95}) is
\begin{align}
K(\a_i;y)&=\sum_{s=0}^\infty \sum_{p=0}^s \f{(-\a_{2}rq^{s}/\a_{1})^{p}q^{-p(p-1)/2}y^{s-1}}{(q;q)_{p}}=\sum_{s=0}^\infty \sum_{p=0}^\infty \f{(-\a_{2}ryq^{s+1}/\a_{1})^{p}q^{p(p-1)/2}y^{s-1}}{(q;q)_{p}} \nn\\
&=\sum_{s=0}^\infty (\a_{2}ryq^{s+1}/\a_{1};q)_\infty y^{s-1}=\f{(\a_{2}rqy/\a_{1};q)_{\infty}}y{}_{2}\phi_{1}\left(\left.\begin{array}{c}
q,0\\
\a_{2}rqy/\a_{1}
\end{array}\right|q;y\right) \nn\\
&\equiv K(\a_{2}r/\a_{1};y) \label{3.93}
\end{align}
where in the second step we shift $s\ra s+p$ and the $q$-hypergeometric function $_2\phi_1$ is defined in \eqref{app:b4}. It is clear that the sum in \eqref{3.93} is convergent for $|\a_1|<|y|<1$. It follows that the chord number states can be written as
\begin{equation}
\ket{0^{n_{2}}10^{n_{1}}10^{m}}=\f{[n_{1}]!_{q}[n_{2}]!_{q}}{(2\pi\i)^{2}}\oint\f{d\a_{1}d\a_{2}}{\a_{1}^{n_{1}+1}\a_{2}^{n_{2}+1}}\oint\f{dy}{2\pi\i}K(\a_{2}r/\a_{1};y)\ket{(\a_{2},\a_{1}/y);0^{m}}\label{eq:100-1}
\end{equation}
where the contour obeys $|\a_2|r<|\a_{1}|<|y|<1$.

For higher matter chord states, the idea of inverse transformation
is similar. The recurrence relation (\ref{eq:88-1}) can be written
as a contour integral
\begin{equation}
\ket{(\a_{k},\a_{k-1},\cdots,\a_{1});0^{m}}=\oint_{|z|>1}\f{dz}{2\pi\i}\sum_{n=0}^{\infty}\f{(\a_{k}r^{k-1}/\a_{1})^{n}}{z(1-q^{n}/z)(q;q)_{n}}\ket{\mB_{\a_{k}}1(z\a_{k-1},\cdots,z\a_{1});0^{m}}
\end{equation}
where the integrand is just $\a_{2}r\ra\a_{k}r^{k-1}$ of (\ref{eq:93-1}).
Taking the same replacement in $K$, the inversion is given by
\begin{equation}
\ket{\mB_{\a_{k}}1(\a_{k-1},\cdots,\a_{1});0^{m}}=\oint_{|\a_1|<|y|<1}\f{dy}{2\pi\i}K(\a_{k}r^{k-1}/\a_{1};y)\ket{(\a_{k},\a_{k-1}/y,\cdots,\a_{1}/y);0^{m}}
\end{equation}
Iterating this process for $k-1$ times, we find 
\begin{align}
\ket{0^{n_{k}}1\cdots10^{n_{1}}10^{m}}=&\oint_{}\prod_{i=1}^{k}\f{[n_{i}]!_{q}d\a_{i}}{2\pi\i\a_{i}^{n_{i}+1}}\oint\prod_{i=1}^{k-1}\f{dy_{i}}{2\pi\i}K(\a_{i+1}r^{i}/\a_{1};y_{i}) \nn\\
&\times \ket{(\a_{k},\cdots,\f{\a_{j}}{y_{j}\cdots y_{k-1}},\cdots,\f{\a_{1}}{y_{1}\cdots y_{k-1}});0^{m}}
\end{align}
where the integral contour obeys $|\a_i|r^{i-1}<|\a_{1}|<|y_j|<1$, $|\a_j|/|y_j\cdots y_{k-1}|,|\a_k|<(1-q)^{-1/2}$ for $i=2,\cdots,k$ and $j=1,\cdots,k-1$. This shows that the full Hilbert space $\mH$ can also be expanded non-uniquely by the overcomplete basis of matter-brane states. 

\section{Some $q$-analog formulae} \label{app:b}

\subsubsection*{Asymptotic expansion of $q$-Pochhammer symbol}
In \cite{katsurada2003asymptotic}, there is an asymptotic expansion
of $q$-Pochhammer symbol in $q\ra1$ limit. With $q=e^{-\lam}$,
we have 
\begin{align}
\log(q^{\a};q)_{\infty} & \app-\f{\pi^{2}}{6\lam}-(\a-\f 12)\log\lam-\log\f{\G(\a)}{\sqrt{2\pi}}+\f 14(\f 16-\a+\a^{2})\lam\label{eq:a164}\\
\log(e^{2\pi\i\mu}q^{\a};q)_{\infty} & \app-\f{\Li_{2}(e^{2\pi\i\mu})}{\lam}-(\a-\f 12)\log(1-e^{2\pi\i\mu})+(\f 16-\a+\a^{2})\f{\lam}{2(1-e^{-2\pi\i\mu})}\label{eq:a165}\\
\log(e^{2\pi\i\mu};q)_{\infty} & \app-\f{\Li_{2}(e^{2\pi\i\mu})}{\lam}+\f 12\log(1-e^{2\pi\i\mu})+\f{\lam}{12(1-e^{-2\pi\i\mu})}\label{eq:a166}
\end{align}

\subsubsection*{Some sum identities}

The $q$-Hypergeometric function is defined as
\begin{equation}\label{app:b4}
_{r+1}\phi_{s}\left(\left.\begin{array}{c}
a_{1},\cdots,a_{r+1}\\
b_{1},\cdots,b_{s}
\end{array}\right|q;z\right)=\sum_{n=0}^{\infty}\f{(a_{1};q)_{n}\cdots(a_{r+1};q)_{n}}{(b_{1};q)_{n}\cdots(b_{s};q)_{n}(q;q)_{n}}\left((-1)^{n}q^{n(n-1)/2}\right)^{s-r}z^{n}
\end{equation} 

Some identities of $q$-binomial theorem
\begin{align}
(a;q)_{n} & =\sum_{i=0}^{n}\f{(q;q)_{n}}{(q;q)_{i}(q;q)_{n-i}}q^{i(i-1)/2}(-a)^{i}\\
\f 1{(a;q)_{n+1}} & =\sum_{i=0}^{\infty}\f{(q;q)_{n+i}}{(q;q)_{i}(q;q)_{n}}a^{i}\\
(aq^{1+n};q)_{m} & =\f{(aq;q)_{m+n}}{(aq;q)_{n}}
\end{align}

The big $q$-Hermite polynomial can be expanded in terms of $q$-Hermite
polynomial \cite{szablowski2013q}
\begin{equation}
H_{n}(x;a|q)=\sum_{m=0}^{n}\chooseq nm(-a)^{m}q^{m(m-1)/2}H_{n-m}(x|q)\label{eq:99-1}
\end{equation}
Another useful formula is the linearization of $q$-Hermite polynomial
\begin{equation}
H_{n}(\cos\t|q)H_{m}(\cos\t|q)=\sum_{j=0}^{\min(n,m)}\f{(q;q)_{m}(q;q)_{n}}{(q;q)_{j}(q;q)_{m-j}(q;q)_{n-j}}H_{m+n-2j}(\cos\t|q)\label{eq:112-1}
\end{equation}
and its inverse transformation
\begin{equation}
H_{m+n}(\cos\t|q)=\sum_{k=0}^{\min(m,n)}\f{(-1)^{k}q^{k(k-1)/2}(q;q)_{m}(q;q)_{n}}{(q;q)_{k}(q;q)_{m-k}(q;q)_{n-k}}H_{n-k}(\cos\t|q)H_{m-k}(\cos\t|q)\label{eq:112}
\end{equation}
Sum of two $q$-Hermite polynomials for $|\r|<1$
\begin{equation}
\sum_{n=0}^{\infty}\f{\r^{n}}{(q;q)_{n}}H_{n}(\cos\t_{1}|q)H_{n}(\cos\t_{2}|q)=\f{(\r^{2};q)_{\infty}}{(\r e^{\i(\pm\t_{1}\pm\t_{2})};q)_{\infty}}\label{eq:app266}
\end{equation}
A useful series of two $q$-Hermite polynomials
\begin{equation}
\sum_{p=0}^{\infty}\f{t^{p}}{(q;q)_{p}}H_{p+m}(\cos\t|q)H_{p+n}(\cos\t'|q)=\f{(t^{2};q)_{\infty}}{(te^{\i(\pm\t\pm\t')};q)_{\infty}}Q_{m,n}(\cos\t,\cos\t'|t,q)\label{eq:152}
\end{equation}
with 
\begin{align}
 & Q_{m,n}(\cos\t,\cos\t'|t,q)=Q_{n,m}(\cos\t',\cos\t|t,q)\nonumber \\
 & =\sum_{s=0}^{n}(-1)^{s}q^{\f{s(s-1)}2}\chooseq ns\f{t^{s}}{(t^{2};q)_{m+s}}H_{n-s}(\cos\t'|q)Q_{m+s}(\cos\t|te^{\mp\i\t'};q)
\end{align}
where $Q_{n}(\cos\t|a,b;q)$ is the Al Salam-Chihara polynomial
\begin{equation}
Q_{n}(\cos\t|a,b;q)=(ae^{\i\t};q)_{n}e^{-\i n\t}{}_{2}\phi_{1}\left(\left.\begin{array}{c}
q^{-n},be^{-\i\t}\\
a^{-1}q^{-n+1}e^{-\i\t}
\end{array}\right|q;a^{-1}qe^{\i\t}\right)\label{eq:154}
\end{equation}
which has orthogonality ($|a,b|<1$)
\be 
\int \f{d\t}{2\pi} \f{(e^{\pm2\i \t},r^2,q;q)_\infty}{(ae^{\pm \i\t},be^{\pm \i \t};q)_\infty(r^2,q;q)_n}Q_n(\cos\t|a,b;q)Q_m(\cos\t|a,b;q)=\d_{mn} \label{app:borth}
\ee

Generating functions \cite{koekoek2010hypergeometric} with $x=\cos\t$
\begin{align}
\f 1{(te^{\pm\i\t};q)_{\infty}} & =\sum_{n=0}^{\infty}\f{H(x|q)}{(q;q)_{n}}t^{n}\label{eq:a186}\\
\f{(\g e^{\i\t}t;q)_{\infty}}{(e^{\i\t}t;q)_{\infty}}{}_{2}\phi_{1}\left(\left.\begin{array}{c}
\g,0\\
\g e^{\i\t}t
\end{array}\right|q;e^{-\i\t}t\right) & =\sum_{n=0}^{\infty}\f{(\g;q)_{n}}{(q;q)_{n}}H_{n}(x|q)t^{n}\label{eq:a209}\\
\f{(at,bt;q)_{\infty}}{(te^{\pm\i\t};q)_{\infty}} & =\sum_{n=0}^{\infty}\f{Q_{n}(x;a,b|q)}{(q;q)_{n}}t^{n}\label{eq:a211}
\end{align}

Sum of Al-Salam-Chihara polynomials \cite{szablowski2013q}. 
\begin{enumerate}
\item For $ab=\a\b$ and $|t|<1$, we have 
\begin{align}
 & \sum_{n=0}^{\infty}\f{t^{n}}{(q;q)_{n}(ab;q)_{n}}Q_{n}(\cos\t|a,b;q)Q_{n}(\cos\t'|\a,\b;q)\nonumber \\
= & \f{(at^{2},\a te^{\i\t},be^{-\i\t},\f{ab}{\a}te^{\i\t},ate^{\pm\i\t'};q)_{\infty}}{(ab,at^{2}e^{\i\t},te^{\i(\pm\t\pm\t')};q)_{\infty}}{}_{8}W_{7}(\f{at^{2}}qe^{\i\t};\f a{\a}t,\f a{\b}t,ae^{\i\t},te^{\i(\t+\t')},te^{\i(\t-\t')};q,be^{-\i\t})\label{eq:a201}
\end{align}
where $_{8}W_{7}$ is a special $q$-Hypergeometric function
\begin{equation}
_{8}W_{7}(a;a_{1},\cdots,a_{5};q,x)={}_{8}\phi_{7}\left(\left.\begin{array}{c}
a,q\sqrt{a},-q\sqrt{a},a_{1},\cdots,a_{5}\\
\sqrt{a},-\sqrt{a},\f{qa}{a_{1}},\cdots,\f{qa}{a_{5}}
\end{array}\right|q;z\right)\label{eq:a202}
\end{equation}
\item For $|\r_{1}|,|\r_{2}|<1$, we have
\begin{equation}
\sum_{n=0}^{\infty}\f{\r_{1}^{n}}{(q;q)_{n}(\r_{2}^{2};q)_{n}}Q_{n}(\cos\t|\r_{2}e^{\mp\i\t'};q)Q_{n}(\cos\t''|\r_{2}e^{\mp\i\t'}/\r_{1};q)=\f{(\r_{1}^{2},\r_{2}e^{\i(\pm\t\pm\t'')};q)_{\infty}}{(\r_{2}^{2},\r_{1}e^{\i(\pm\t\pm\t')};q)_{\infty}}\label{eq:a205}
\end{equation}
\end{enumerate}

\section{Calculus illustration using matter-brane state} \label{app:c}

Since the subspace with fixed number of matter chord can be organized
into $q$-coherent states, in which $H_{R,0}$ acts as if $\tilde{H}_{R,0}$
with $\mu=\a_{1}r$, we can define the energy eigen basis for $H_{R,0}$
for $k$-matter subspace
\begin{align}
\ket{(\a_{k},\cdots,\a_{1});\t} & =\sum_{n=0}^{\infty}\f{\psi_{n}(\t;\a_{1}r\sqrt{1-q}|q)}{\sqrt{[n]!_q}}\ket{(\a_{k},\cdots,\a_{1});0^{n}}\\
H_{R,0}\ket{(\a_{k},\cdots,\a_{1});\t} & =E(\t)\ket{(\a_{k},\cdots,\a_{1});\t}\\
\psi_{n}(\t;a|q)&\equiv H_n(\cos\t;a|q)/\sqrt{(q;q)_n}
\end{align}
Similarly, for the EoW brane Hamiltonian $\tilde{H}_{R,0}$ in multiple
matter chord subspace, we can construct the energy eigen basis using
$\a$-coherent states. Acting $\tilde{H}_{R,0}$ on $\ket{(\a_{k},\cdots,\a_{1});0^{m}}$,
we have
\begin{align}
\tilde{H}_{0}\ket{(\a_{k},\cdots,\a_{1});0^{m}}= & [m]_{q}\ket{(\a_{k},\cdots,\a_{1});0^{m-1}}+\a_{1}rq^{m}\ket{(\a_{k},\cdots,\a_{1});0^{m}}\nonumber \\
 & +\ket{(\a_{k},\cdots,\a_{1});0^{m+1}}+\mu r^{k}q^{m}\ket{(q\a_{k},\cdots,q\a_{1});0^{m}}
\end{align}
Similar to the first line of (\ref{eq:52}), we can define 
\begin{equation}
\ket{\mu;(\a_{k},\cdots,\a_{1});0^{m}}=\sum_{n}\f{(\mu r^{k-1}/\a_{1})^{n}}{(q;q)_{n}}\ket{(q^{n}\a_{k},\cdots,q^{n}\a_{1});0^{m}}
\end{equation}
such that
\begin{align}
\tilde{H}_{R,0}\ket{\mu;(\a_{k},\cdots,\a_{1});0^{m}}&=[m]_{q}\ket{\mu;(\a_{k},\cdots,\a_{1});0^{m-1}}\nn\\
&+\a_{1}rq^{m}\ket{\mu;(\a_{k},\cdots,\a_{1});0^{m}}+\ket{\mu;(\a_{k},\cdots,\a_{1});0^{m+1}}
\end{align}
This means that the effect of EoW brane is completely screened by
matter chords and we can regard $\tilde{H}_{R,0}$ as if the brane tension
being $\mu=\a_{1}r$. The corresponding energy basis is 
\begin{align}
\ket{\mu;(\a_{k},\cdots,\a_{1});\t} & =\sum_{n=0}^{\infty}\f{\psi_{n}(\t;\a_{1}r\sqrt{1-q}|q)}{\sqrt{[n]!_q}}\ket{\mu;(\a_{k},\cdots,\a_{1});0^{n}}\\
\tilde{H}_{0}\ket{\mu;(\a_{k},\cdots,\a_{1});\t} & =E(\t)\ket{\mu;(\a_{k},\cdots,\a_{1});\t}
\end{align}

\subsection{A crossing element}

\begin{figure}
\begin{centering}
\includegraphics[totalheight=2cm]{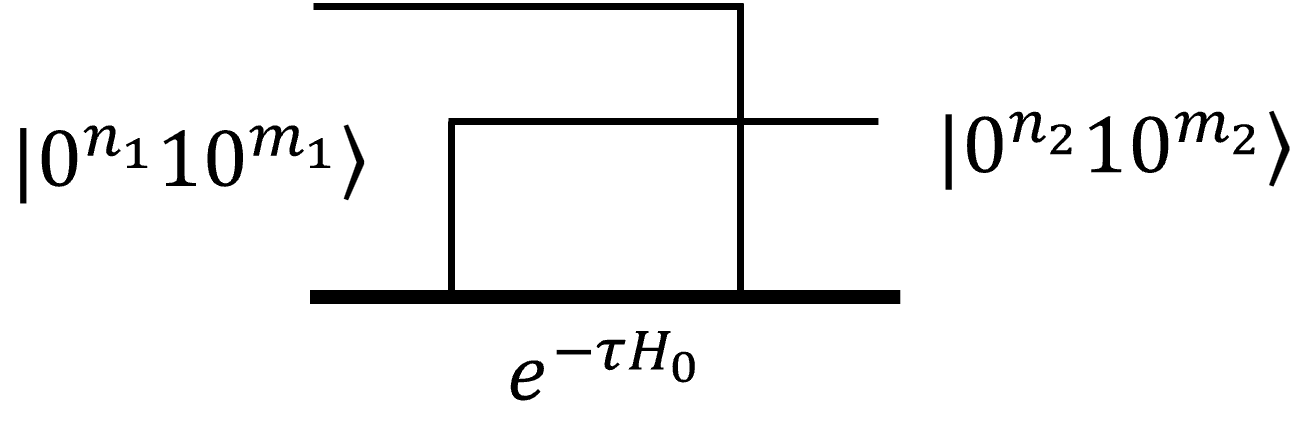}
\par\end{centering}
\caption{The crossing element $\protect\mX(\tau)$. \protect\label{fig:The-crossing-element.}}
\end{figure}

For our illustration purpose of this computation method, let us consider the crossing element $\mX(\tau)$
in Figure \ref{fig:The-crossing-element.}. It is a matrix element
in the one-matter chord basis from $\ket{0^{n_{1}}10^{m_{1}}}$ to
$\ket{0^{n_{2}}10^{m_{2}}}$ though it involves an intermediate step
into the two-matter chord subspace. We have 
\begin{align}
 & \mX(\tau)\ket{0^{n_{1}}10^{m_{1}}}=\tilde{a}_{1}e^{-\tau H_{0}}\ket{0^{n_{1}}10^{m_{1}}1}\nonumber \\
= & \f{[n_{1}]!_{q}[m_{1}]!_{q}}{(2\pi\i)^{3}}\oint\f{d\a_{1}d\a_{2}dyK(\a_{2}r/\a_{1};y)}{\a_{1}^{m_{1}+1}\a_{2}^{n_{1}+1}}\tilde{a}_{1}e^{-\tau H_{0}}\ket{(\a_{2},\a_{1}/y);0^{0}}\nonumber \\
= & \f{[n_{1}]!_{q}[m_{1}]!_{q}}{(2\pi\i)^{3}}\oint\f{d\a_{1}d\a_{2}dyK(\a_{2}r/\a_{1};y)}{\a_{1}^{m_{1}+1}\a_{2}^{n_{1}+1}}\int_{0}^{\pi}\f{d\t}{2\pi}\r(\t;\a_{1}r)e^{-\tau E(\t)}\tilde{a}_{1}\ket{(\a_{2},\a_{1}/y);\t}\nonumber \\
= & q_{m}\f{[n_{1}]!_{q}[m_{1}]!_{q}}{(2\pi\i)^{2}}\oint\f{d\a_{1}d\a_{2}dyK(\a_{2}r/\a_{1};y)}{\a_{1}^{m_{1}+1}\a_{2}^{n_{1}+1}}\int_{0}^{\pi}\f{d\t}{2\pi}\r(\t;\a_{1}r)e^{-\tau E(\t)}\nonumber \\
 & \times\sum_{m_{2}=0}^{\infty}\f{\psi_{m_{2}}(\t,\a_{1}r\sqrt{1-q}/y|q)}{\sqrt{[m_{2}]!_{q}}}r^{m_{2}}\sum_{k_{1},k_{2}}\f{(\a_{1}r/y)^{k_{1}}\a_{2}^{k_{2}}}{(\a_{2}rq^{k_{1}}y/\a_{1};q)_{\infty}[k_{1}]!_{q}[k_{2}]!_{q}}\ket{0^{k_{1}+k_{2}}10^{m_{2}}}\label{eq:79-1}
\end{align}
where $\tilde{a}_{1}$ means annihilate the outmost matter chord from
the boundary, in the second line we used (\ref{eq:100-1}), in the
third line we used (\ref{eq:69}) with $\mu=\a_{1}r/y$, and in the
fourth line we used (\ref{3.7}) and (\ref{eq:47}).

To simplify (\ref{eq:79-1}), we shift $\a_{1}\ra\a_{1}y$ and complete
the contour integral over $y$
\begin{equation}
\oint\f{dy}{2\pi\i y^{m_{1}}}K(\a_{2}r/(\a_{1}y);y)=\sum_{s=0}^{\infty}\f{(-\a_{2}rq^{m_{1}}/\a_{1})^{s}q^{s(s+1)/2}}{(q;q)_{s}}
\end{equation}
We can further integrate over $\a_{2}$
\begin{equation}
\oint\f{d\a_{2}}{2\pi\i\a_{2}^{n_{1}-s-k_{2}+1}(\a_{2}rq^{k_{1}}/\a_{1};q)_{\infty}}=\f{(rq^{k_{1}}/\a_{1})^{n_{1}-s-k_{2}}}{(q;q)_{n_{1}-s-k_{2}}}
\end{equation}
and sum over $s$
\begin{equation}
\sum_{s=0}^{n_{1}-k_{2}}\f{(-q^{m_{1}-k_{1}+1})^{s}q^{s(s-1)/2}}{(q;q)_{s}(q;q)_{n_{1}-k_{2}-s}}=\f{(q;q)_{m_{1}+n_{1}-k_{1}-k_{2}}}{(q;q)_{m_{1}-k_{1}}(q;q)_{n_{1}-k_{2}}}
\end{equation}
It follows that 
\begin{align}
\mX(\tau)\ket{0^{n_{1}}10^{m_{1}}} & =\f{[n_{1}]!_{q}[m_{1}]!_{q}}{2\pi\i}q_{m}\sum_{m_{2},k_{1},k_{2}}\oint\f{d\a_{1}r^{m_{2}+n_{1}+k_{1}-k_{2}}}{\a_{1}^{m_{1}+n_{1}-k_{1}-k_{2}+1}}\int_{0}^{\pi}\f{d\t}{2\pi}\r(\t;\a_{1}r)e^{-\tau E(\t)}\nonumber \\
 & \times\f{q^{k_{1}(n_{1}-k_{2})}(q;q)_{m_{1}+n_{1}-k_{1}-k_{2}}\psi_{m_{2}}(\t,\a_{1}r\sqrt{1-q}|q)}{(q;q)_{m_{1}-k_{1}}(q;q)_{n_{1}-k_{2}}\sqrt{[m_{2}]!_{q}}[k_{1}]!_{q}[k_{2}]!_{q}}\ket{0^{k_{1}+k_{2}}10^{m_{2}}}\label{eq:130}
\end{align}
Using the expansion (\ref{eq:99-1}) and (\ref{eq:a186}), we can expand
all terms in (\ref{eq:130}) as powers of $\a_{1}$ and complete the
$\a_{1}$ contour integral
\begin{align}
&\mX(\tau)\ket{0^{n_{1}}10^{m_{1}}}\nn\\
= & \f{[n_{1}]!_{q}[m_{1}]!_{q}}{2\pi\i}q_{m}\sum_{m_{2},k_{1},k_{2},n,m}\oint\f{d\a_{1}r^{m_{2}+n_{1}+k_{1}-k_{2}}}{\a_{1}^{m_{1}+n_{1}-k_{1}-k_{2}+1}}\int_{0}^{\pi}\f{d\t}{2\pi}\r(\t)e^{-\tau E(\t)}\f{(\a_{1}r\sqrt{1-q})^{n}}{(q;q)_{n}}H_{n}(\cos\t|q)\nonumber \\
 & \times\chooseq{m_{2}}m\f{q^{k_{1}(n_{1}-k_{2})}(q;q)_{m_{1}+n_{1}-k_{1}-k_{2}}(-\a_{1}r\sqrt{1-q})^{m}q^{m(m-1)/2}H_{m_{2}-m}(\cos\t|q)}{(q;q)_{m_{1}-k_{1}}(q;q)_{n_{1}-k_{2}}(q;q)_{m_{2}}(1-q)^{-m_{2}/2}[k_{1}]!_{q}[k_{2}]!_{q}}\ket{0^{k_{1}+k_{2}}10^{m_{2}}}\nonumber \\
= & q_{m}\sum_{m_{2},k_{1},k_{2},m}\f{(q;q)_{n_{1}}(q;q)_{m_{1}}r^{m_{1}+m_{2}+2n_{1}-2k_{2}}}{(1-q)^{(m_{1}+n_{1}-k_{1}-k_{2}-m_{2})/2}}\int_{0}^{\pi}\f{d\t}{2\pi}\r(\t)e^{-\tau E(\t)}\f{H_{m_{1}+n_{1}-k_{1}-k_{2}-m}(\cos\t|q)}{(q;q)_{m_{1}+n_{1}-k_{1}-k_{2}-m}}\nonumber \\
 & \times\f{(-)^{m}q^{k_{1}(n_{1}-k_{2})+m(m-1)/2}(q;q)_{m_{1}+n_{1}-k_{1}-k_{2}}H_{m_{2}-m}(\cos\t|q)}{(q;q)_{m}(q;q)_{m_{1}-k_{1}}(q;q)_{n_{1}-k_{2}}(q;q)_{m_{2}-m}(q;q)_{k_{1}}(q;q)_{k_{2}}}\ket{0^{k_{1}+k_{2}}10^{m_{2}}}\label{eq:143}\\
= & q_{m}\sum_{m_{2},k_{1},k_{2}}\f{(q;q)_{n_{1}}(q;q)_{m_{1}}r^{m_{1}+m_{2}+2n_{1}-2k_{2}}}{(1-q)^{(m_{1}+n_{1}-k_{1}-k_{2}-m_{2})/2}}\int_{0}^{\pi}\f{d\t}{2\pi}\r(\t)e^{-\tau E(\t)}\nonumber \\
 & \times\f{q^{k_{1}(n_{1}-k_{2})}H_{m_{1}+m_{2}+n_{1}-k_{1}-k_{2}}(\cos\t|q)}{(q;q)_{m_{2}}(q;q)_{m_{1}-k_{1}}(q;q)_{n_{1}-k_{2}}(q;q)_{k_{1}}(q;q)_{k_{2}}}\ket{0^{k_{1}+k_{2}}10^{m_{2}}}\label{eq:133-1}
\end{align}
where in the last step we sum over $m$ using (\ref{eq:112}).

\subsection{Time evolution in 1-matter chord subspace}

We also need to write down the time evolution in 1-matter chord subspace
\begin{align}
 & e^{-\tau H_{0}}\ket{0^{n_{1}}10^{m_{1}}}\nonumber \\
= & \f{[n_{1}]!_{q}\sqrt{[m_{1}]!_{q}}}{2\pi\i}\oint_{0}\f{d\a_{1}}{\a_{1}^{n_{1}+1}}\int_{0}^{\pi}\f{d\t}{2\pi}\r(\t;\a_{1}r)\psi_{m_{1}}(\t;\a_{1}r\sqrt{1-q}|q)e^{-\tau E(\t)}\ket{(\a_{1});\t}\nonumber \\
= & \f{[n_{1}]!_{q}\sqrt{[m_{1}]!_{q}}}{2\pi\i}\oint_{0}\f{d\a_{1}}{\a_{1}^{n_{1}+1}}\int_{0}^{\pi}\f{d\t}{2\pi}\r(\t;\a_{1}r)\psi_{m_{1}}(\t;\a_{1}r\sqrt{1-q}|q)e^{-\tau E(\t)}\nonumber \\
 & \times\sum_{n_{2}m_{2}}\f{\a_{1}^{n_{2}}}{[n_{2}]!_{q}}\f{\psi_{m_{2}}(\t;\a_{1}r\sqrt{1-q}|q)}{\sqrt{[m_{2}]!_{q}}}\ket{0^{n_{2}}10^{m_{2}}}\nonumber \\
= & \sum_{n_{2}m_{2},m,n,k}\int_{0}^{\pi}\f{d\t}{2\pi}\r(\t)e^{-\tau E(\t)}\f{(q;q)_{n_{1}}(q;q)_{m_{1}}(1-q)^{(n_{2}+m_{2}-n_{1}-m_{1})/2}}{(q;q)_{n_{2}}(q;q)_{m}(q;q)_{m_{1}-m}(q;q)_{k}(q;q)_{m_{2}-n}(q;q)_{n}}\d_{n_{1},m+n+k+n_{2}}\nonumber \\
 & \times(-)^{m+n}r^{m+n+k}q^{m(m-1)/2+n(n-1)/2}H_{m_{1}-m}(\cos\t|q)H_{m_{2}-n}(\cos\t|q)H_{k}(\cos\t|q)\ket{0^{n_{2}}10^{m_{2}}}\label{eq:134}\\
= & \sum_{n_{2}m_{2},k}\int_{0}^{\pi}\f{d\t}{2\pi}\r(\t)e^{-\tau E(\t)}\f{(q;q)_{n_{1}}(q;q)_{m_{1}}(1-q)^{(n_{2}+m_{2}-n_{1}-m_{1})/2}}{(q;q)_{m_{2}}(q;q)_{n_{1}-k}(q;q)_{k-n_{2}}(q;q)_{n_{2}}(q;q)_{m_{1}+n_{2}-k}}\nonumber \\
 & \times(-)^{k+n_{2}}r^{n_{1}-n_{2}}q^{(n_{2}-k)(n_{2}-k+1)/2}H_{n_{1}+m_{2}-k}(\cos\t|q)H_{m_{1}+n_{2}-k}(\cos\t|q)\ket{0^{n_{2}}10^{m_{2}}}\label{eq:133}
\end{align}
where in the fourth step we plug in $m=n_{1}-n-k-n_{2}$ and then
shift $k\ra-k+n_{1}-n$ and sum over $n$ using (\ref{eq:112}).

\subsection{Consistent check 1: two-point function}

Using the $\a$-coherent basis, we can easily compute the two-point
function without counting the chord diagram combinatorics. Let us
compute the following element acting on 0-matter chord subspace
\begin{align}
&a_{1}e^{-\tau H_{0}}a_{1}^{\dag}\ket{0^{n_{1}}}\nn\\
= & \sum_{n_{2}m_{2},n}\int_{0}^{\pi}\f{d\t}{2\pi}\r(\t)e^{-\tau E(\t)}\f{(q;q)_{n_{1}}(1-q)^{(n_{2}+m_{2}-n_{1})/2}}{(q;q)_{n_{2}}(q;q)_{n_{1}-n_{2}-n}(q;q)_{m_{2}-n}(q;q)_{n}}\nonumber \\
 & \times(-)^{n}r^{n_{1}-n_{2}+m_{2}}q^{n(n-1)/2}H_{m_{2}-n}(\cos\t|q)H_{n_{1}-n_{2}-n}(\cos\t|q)\ket{0^{n_{2}+m_{2}}}\nonumber \\
= & \sum_{n_{2}m_{2}}\int_{0}^{\pi}\f{d\t}{2\pi}\r(\t)e^{-\tau E(\t)}\f{(r^{2};q)_{n_{1}-n_{2}}(q;q)_{n_{1}}(1-q)^{(m_{2}-n_{2})/2}}{(q;q)_{n_{2}}(q;q)_{m_{2}}(q;q)_{n_{1}-n_{2}}}\nonumber \\
 & \times r^{n_{2}+m_{2}}H_{m_{2}}(\cos\t|q)H_{n_{2}}(\cos\t|q)\ket{0^{m_{2}+n_{1}-n_{2}}}\nonumber \\
= & \sum_{m_{2},n_{2}}\sqrt{\chooseq{n_{1}}{n_{1}-n_{2}}\chooseq{m_{2}+n_{1}-n_{2}}{n_{1}-n_{2}}}\f{(r^{2};q)_{n_{1}-n_{2}}\ket{0^{n_{1}+m_{2}-n_{2}}}\avg{0^{n_{2}}|r^{N_{0}}e^{-\tau H_{0}}r^{N_{0}}|0^{m_{2}}}}{\sqrt{[n_{2}]!_{q}[m_{2}]!_{q}[n_{1}+m_{2}-n_{2}]!_{q}/[n_{1}]!_{q}}}\label{eq:135}
\end{align}
where in the first step we used (\ref{eq:134}) with $m_{1}=0$ (that
also implies $m=0$), in the second step we shift $m_{2}\ra m_{2}+n$
and $n_{2}\ra n_{1}-n-n_{2}$ and sum over $n$ from 0 to $n_{1}-n_{2}$.
From (\ref{eq:135}), we see it exactly matches with the result in
\cite{Berkooz:2018jqr}.

\subsection{Consistent check 2: OTOC and $q$-deformed 6$j$ symbol}

An OTOC in DSSYK is defined as
\begin{equation}
\mF(\tau_{1},\tau_{2},\tau_{3},\tau_{4})=\avg{\w|e^{-\tau_{4}H_{0}}H_{1}^{(a)}e^{-\tau_{3}H_{0}}H_{1}^{(b)}e^{-\tau_{2}H_{0}}H_{1}^{(a)}e^{-\tau_{1}H_{0}}H_{1}^{(b)}|\w}
\end{equation}
where the $H_{1}^{(a)}$ pair and the $H_{1}^{(b)}$ pair are respectively
contracted as two crossing matter chords. This quantity is equivalent
to 
\begin{equation}
\avg{\w|e^{-\tau_{4}H_{0}}a_{1}e^{-\tau_{3}H_{0}}\mX(\tau_{2})e^{-\tau_{1}H_{0}}a_{1}^{\dag}|\w}
\end{equation}
Taking (\ref{eq:133-1}) and (\ref{eq:133}) into above formula, we
have 
\begin{align}
&\mF(\tau_{1},\tau_{2},\tau_{3},\tau_{4})=  q_{m}\int\prod_{i=1}^{4}\f{d\t_{i}}{2\pi}\r(\t_{i})e^{-\tau_{i}E(\t_{i})}V(\t_{i})\nonumber \\
&V(\t_{i})=  \sum_{k,k_{1},m_{1},m_{2},m_{3},n_{3}}(-1)^{k+n_{3}}q^{(k-n_{3})(k-n_{3}-1)/2}r^{k_{1}+m_{1}+m_{2}+m_{3}-n_{3}}H_{m_{1}}(\cos\t_{1}|q)\nonumber \\
 & \times\f{H_{m_{1}+m_{2}-k_{1}}(\cos\t_{2}|q)H_{m_{3}-k_{1}-k}(\cos\t_{3}|q)H_{m_{2}+n_{3}-k}(\cos\t_{3}|q)H_{m_{3}+n_{3}}(\cos\t_{4}|q)}{(q;q)_{k_{1}-k}(q;q)_{m_{1}-k_{1}}(q;q)_{m_{3}}(q;q)_{k-n_{3}}(q;q)_{n_{3}}(q;q)_{m_{2}+n_{3}-k}}\label{eq:139}
\end{align}
where we have set some terms like $(q;q)_{k_{2}}(q;q)_{-k_{2}}=H_{0}(x|q)=1$
in the sum. The range of the numbers in the sum of (\ref{eq:139})
is up to the constraint that all $(q;q)_{n}$ term should have $n\geq0$.
To simplify the sum, we can shift $k\ra k+n_{3}$, then $k_{1}\ra k+k_{1}$,
then $m_{2}\ra m_{2}+k,m_{1}\ra m_{1}+k_{1},m_{3}\ra m_{3}+n_{3}-k_{1}$,
then $n_{3}\ra k_{1}-n_{3}$ and have
\begin{align}
V(\t_{i})= & \sum_{k,k_{1},m_{1},m_{2},m_{3},n_{3}}(-1)^{k}q^{k(k-1)/2}r^{2k+k_{1}+m_{1}+m_{2}+m_{3}}H_{m_{1}+k_{1}}(\cos\t_{1}|q)\nonumber \\
 & \times\f{H_{m_{1}+m_{2}}(\cos\t_{2}|q)H_{m_{2}}(\cos\t_{3}|q)H_{m_{3}}(\cos\t_{3}|q)H_{k_{1}+m_{3}-2n_{3}}(\cos\t_{4}|q)}{(q;q)_{k}(q;q)_{m_{1}-k}(q;q)_{m_{2}}(q;q)_{k_{1}-n_{3}}(q;q)_{m_{3}-n_{3}}(q;q)_{n_{3}}}\nonumber \\
= & \sum_{k_{1},m_{1},m_{2},m_{3},n_{3}}r^{k_{1}+m_{1}+m_{2}+m_{3}}(r^{2};q)_{m_{1}}H_{m_{1}+k_{1}}(\cos\t_{1}|q)\nonumber \\
 & \times\f{H_{m_{1}+m_{2}}(\cos\t_{2}|q)H_{m_{2}}(\cos\t_{3}|q)H_{m_{3}}(\cos\t_{3}|q)H_{k_{1}+m_{3}-2n_{3}}(\cos\t_{4}|q)}{(q;q)_{m_{1}}(q;q)_{m_{2}}(q;q)_{k_{1}-n_{3}}(q;q)_{m_{3}-n_{3}}(q;q)_{n_{3}}}\nonumber \\
= & \sum_{k_{1},m_{1},m_{2},m_{3}}r^{k_{1}+m_{1}+m_{2}+m_{3}}(r^{2};q)_{m_{1}}H_{m_{1}+k_{1}}(\cos\t_{1}|q)H_{m_{1}+m_{2}}(\cos\t_{2}|q)\nonumber \\
 & \times\f{H_{m_{2}}(\cos\t_{3}|q)H_{m_{3}}(\cos\t_{3}|q)H_{k_{1}}(\cos\t_{4}|q)H_{m_{3}}(\cos\t_{4}|q)}{(q;q)_{m_{1}}(q;q)_{m_{2}}(q;q)_{k_{1}}(q;q)_{m_{3}}}
\end{align}
where in the second step we sum over $k$ from 0 to $m_{1}$, and
in the third step we sum over $n_{3}$ from 0 to $\min(k_{1}.m_{3})$
using (\ref{eq:112-1}). Comparing with (4.10) of \cite{Berkooz:2018jqr},
we find exact match except a shift $\t_{1}\ra\t_{2}\ra\t_{3}\ra\t_{4}\ra\t_{1}$,
which is indeed the symmetry of the $q$-deformed $6j$ symbol.

\section{Two-point functions in the thermal EoW brane state}\label{app:d}

\subsection{Thermal EoW brane state in semiclassical JT gravity}\label{app:d1}

In JT gravity, the state $\ket{\mu_{r},\b/2}$ dual to a single sided
black hole with an EoW brane behind the horizon can be prepared by
an Euclidean path integral shown in Figure \ref{pic-jt-eow}. In this state, there
is an insertion of an EoW brane on the EAdS$_{2}$ boundary followed
by an Euclidean evolution of length $\b/2$ with Hamiltonian \eqref{3.14}.
Note that $\b$ here is just a parameter for this state, which is
not the temperature of the black hole. By canonical quantization,
the propagator in geodesic length basis is \cite{Gao:2021uro}
\begin{align}
G_{\mu_{r},\b}(L_{2},L_{1}) & =\int dke^{-\b k^{2}/(2\phi_{b})}\p(k)(z_{1}z_{2})^{-1/2}W_{-\mu_{r},\i k}(z_{1})W_{-\mu_{r},\i k}(z_{2})\label{eq:app201}\\
\p(k) & =\f 1{\pi^{2}}k\sinh2\pi k|\G(1/2+\mu_{r}-\i k)|^{2}
\end{align}
where $z_{i}=4e^{-L_{i}}$ and $W_{-\mu,\i k}(z)$ is the Whittaker
function. The norm of $\ket{\mu_{r},\b/2}$ can be computed
as an Euclidean path integral in Figure \ref{pic-jt-eow-2}. This can be understood
as the $L_{i}\ra-\infty$ limit of the propagator \ref{eq:app201}
\begin{equation}
\avg{\mu_{r},\b/2|\mu_{r},\b/2}=\lim_{L_{1,2}\ra\infty}N(L_{1})N(L_{2})G_{\mu,\b}(L_{1},L_{2})
\end{equation}
where $N(L)=\pi e^{z/2}z^{\mu_{r}+1/2}$ is a function of $L$ that
cancels the singular part of $G_{\mu_{r},\b}$ due to $W_{-\mu_{r},\i k}(z)\sim e^{-z/2}z^{-\mu_{r}}$when
$z\ra+\infty$. It follows that
\begin{equation}
\avg{\mu_{r},\b/2|\mu_{r},\b/2}=\int dke^{-\b k^{2}/(2\phi_{b})}k\sinh2\pi k|\G(1/2+\mu_{r}-\i k)|^{2}\label{eq:app204-1}
\end{equation}
which can also be derived using the method in \cite{Yang:2018gdb} and integrating
the length $\l$ of the Hartle-Hawking state with a weight $e^{-\mu_{r}\l}$
from the EoW brane. 

As a generic state $\ket{\mu_{r},\b/2}$ is prepared in the quantized
JT gravity, it does not correspond to a fixed semiclassical spacetime
geometry because the energy has large fluctuations. To achieve a state
with fixed geometric dual, we need to take the semiclassical limit
with a heavy EoW brane
\begin{equation}
\mu_{r}=\tilde{\mu}\phi_{b},\quad\phi_{b}\ra\infty
\end{equation}
In this limit, we can rescale $k\ra k\phi_{b}$ and expand the integrand
\ref{eq:app204-1} in the exponential form
\begin{equation}
\avg{\mu_{r},\b/2|\mu_{r},\b/2}\sim\int dkke^{\phi_{b}(-\b k^{2}/2+2\pi k+(\tilde{\mu}-\i k)\log(\tilde{\mu}-\i k)+(\tilde{\mu}+\i k)\log(\tilde{\mu}+\i k))}\label{eq:app206}
\end{equation}
Since $\phi_{b}$ is large, this integral is dominated by the saddle
$k_{s}$ obeying equation
\begin{equation}
2\pi-k_{s}\b=\i\log\f{\tilde{\mu}-\i k_{s}}{\tilde{\mu}+\i k_{s}}\iff k_{s}/\tilde{\mu}=\tan(\pi-k_{s}\b/2)\label{eq:app207}
\end{equation}

This saddle corresponds to the fixed geometry of a black hole with
an EoW brane behind horizon as we can check the consistency as follows.
In Lorentzian signature, the ADM energy is $E=\Phi_{h}^{2}/(2\phi_{b})$
where $\Phi_{h}$ is the horizon dilaton value \cite{Gao:2021uro}, which compares
with the Boltzman factor in (\ref{eq:app201}) identifies $k_{s}=\Phi_{h}/\phi_{b}$.
The inverse temperature of the black hole is given by 
\begin{equation}
\b_{BH}=2\pi\phi_{b}/\Phi_{h}=2\pi/k_{s}\label{eq:app208}
\end{equation}
The trajectory of the EoW brane is a geodesic, which in global coordinate
is
\begin{equation}
\cos\s=r\cos T,\quad r=\f{\tilde{\mu}}{\sqrt{k_{s}^{2}+\tilde{\mu}^{2}}},\quad ds^{2}=\f{-dT^{2}+d\s^{2}}{\sin^{2}\s},\quad\s\in[0,\pi]
\end{equation}
Given $r\in(0,1)$, the Euclidean version of this geodesic is $\cos\s=r\cosh T$,
which is anchored on the EAdS boundary $\s=0$ and $T=\pm\arccosh(1/r)$,
which in Rindler coordinate $ds^{2}=d\r^{2}+\sinh^{2}\r d\phi^{2}$
corresponds to two angles at $\pm\phi_{0}=\pm(\pi-\arcsin\sqrt{1-r^{2}})$
(see Figure \ref{pic-jt-eow-2}). For a semiclassical black hole, the ratio between
the arc angle to $2\pi$ should be equal to $\b/\b_{BH}$
\begin{equation}
1-\f 1{\pi}\arcsin\sqrt{1-r^{2}}=\f{\b}{\b_{BH}}\iff\f{k_{s}}{\sqrt{k_{s}^{2}+\tilde{\mu}^{2}}}=\sin(\pi-k_{s}\b/2)
\end{equation}
which perfectly matches with the saddle equation (\ref{eq:app207}).
We can further define $x_{s}=k_{s}\b$, and the saddle is fixed by
a single parameter $\b\tilde{\mu}$ that 
\begin{equation}
x_{s}/(\b\tilde{\mu})=\tan(\pi-x_{s}/2)
\end{equation}
Clearly, a heavier EoW brane leads to a cooler black hole.

\subsection{Thermal EoW brane state in DSSYK}\label{app:d2}

The thermal EoW brane state in DSSYK is defined in \eqref{eq:4.10}, whose inner product
is given by
\begin{align}
\avg{\mu,\b_{0}/2|\mu,\b_{0}/2} & =\avg{\mB_{\mu}|e^{-\b_{0}H_{R}}|\w}=\int\f{d\t}{2\pi}\sum_{n}\f{\mu^{n}(1-q)^{n/2}}{(q;q)_{n}}\r(\t)H_{n}(\cos\t|q)e^{-\b_{0}E(\t)}\nonumber \\
 & =\int\f{d\t}{2\pi}\f{(e^{\pm2\i\t},q;q)_{\infty}e^{-\b_{0}E(\t)}}{(\mu(1-q)^{1/2}e^{\pm\i\t};q)_{\infty}}\label{eq:app211}
\end{align}
Similar to the JT case, the energy of the state $\ket{\mu,\b_{0}/2}$
has a large fluctuation in the whole spectrum and it does not dual
to a fixed geometry. For this purpose, we need to take the limit of
heavy EoW brane
\begin{equation}
\lam\ra0,\quad\b_{0}\ra\hat{\b}/\sqrt{1-q},\quad\mu=-e^{-\mu_{0}}/\sqrt{1-q}\label{eq:app212}
\end{equation}
which compared with \eqref{3.13} indicates a large $\mu_{r}\sim O(1/\lam)$. Using formulae in Appendix \ref{app:b}, (\ref{eq:app211})
becomes
\begin{equation}
\avg{\mu,\b_{0}/2|\mu,\b_{0}/2}\app\int\f{d\t}{\sqrt{2\pi\lam}}\f{e^{\mu_{0}/2}\sin\t e^{\f{\Li_{2}(-e^{-\i\t-\mu_{0}})+\Li_{2}(-e^{\i\t-\mu_{0}})-(\pi-2\t)^{2}/2-2\hat{\b}\cos\t}{\lam}-\hat{\b}\cos\t}}{\sqrt{\cosh\f{\mu_{0}+\i\t}2\cosh\f{\mu_{0}-\i\t}2}}\label{eq:app213}
\end{equation}
The saddle equation for $\t$ is
\begin{equation}
\pi-2\t_{s}+\hat{\b}\sin\t_{s}=-\arctan\f{\sin\t_{s}}{e^{\mu_{0}}+\cos\t_{s}}\label{eq:app214}
\end{equation}
In order to have a saddle close to the ground energy $E_{0}$, we
need to consider the case with $\t_{s}$ close to $\pi$. This requires
the following semiclassical JT limit with a small parameter $\d$
\begin{equation}
\t_{s}=\pi-\d k_{s}/2,\quad\mu_{0}=\tilde{\mu}\d/2,\quad\hat{\b}=\b/\d,\quad\d\ra0\label{eq:app215}
\end{equation}
Taking small $\d$ limit of (\ref{eq:app214}) leads to 
\begin{equation}
\pi-\b k_{s}/2=\arctan\f{k_{s}}{\tilde{\mu}}
\end{equation}
which exactly matches with the semiclassical saddle equation (\ref{eq:app207})
of JT gravity. Furthermore, if we expand the exponent of (\ref{eq:app213})
in linear order of $\d$, we find the leading term is
\begin{equation}
\sim e^{\f{\d}{2\lam}(-\b k^{2}/2+2\pi k+(\tilde{\mu}-\i k)\log(\tilde{\mu}-\i k)+(\tilde{\mu}+\i k)\log(\tilde{\mu}+\i k))} \label{app:d18}
\end{equation}
which compared with (\ref{eq:app206}) identifies 
\begin{equation}
\phi_{b}=\d/(2\lam) \label{app:d19}
\end{equation}
This leading exponential term \eqref{app:d18} evaluated with the saddle solution is exponentially divergent because the ground energy $E_0\sim O(1/\lam)$ is negative and large. Therefore, the semiclassical limit of $\ket{\mu,\b_0/2}$ has a singular norm $1/\mN_\lam\sim O(e^{\d S/\lam})$ with $S>0$, and we can divide this norm to define an algebraic state
\be 
\ket{\phi^\mu_{\b_0/2}}\simeq \lim_{\lam\ra 0}\mN_\lam^{1/2}\ket{\mu,\b_0/2}
\ee
This algebraic state is only well-defined in terms of expectation values of operators in the sense of \eqref{4.10-2} and \eqref{4.11}.

This analysis explains the subtle parameter range for a thermal EoW brane
state in DSSYK to have a semiclassical geometry dual. It has two steps.
The first step is taking JT limit $\lam\ra0$ with a heavy EoW brane;
the second step is to take a low temperature/semiclassical limit and
tune the heavy brane tension $\mu_{r}$ much lighter than $1/\lam$.
We summarize it as follows
\begin{equation}
\t=\pi-\d k/2,\quad\mu_{r}=\tilde{\mu}\d/(2\lam),\quad\b_{0}=\f{\b}{\d\sqrt{1-e^{-\lam}}},\quad\lam\ll\d\ll1
\end{equation}

\subsection{Two-point function in thermal EoW brane state} \label{app:d3}

The Euclidean two-point function in a thermal EoW brane state is
\begin{align}
G_{\b_{0}}(\tau_{1},\tau_{2})=&\avg{\mu,\b_{0}/2|e^{\tau_{1}H_{R}}M_{R}e^{-(\tau_{1}-\tau_{2})H_{R}}M_{R}e^{-\tau_{2}H_{R}}|\mu,\b_{0}/2}\nn\\
=&\avg{\mB_{\mu}|e^{-(\b_{0}-\tau_{12})H_{R}}M_{R}e^{-\tau_{12}H_{R}}M_{R}|\w}
\end{align}
where we used the tracial property of $\ket{\w}$ and $\mB_{R,\mu}$
is an operator function of $H_{R}$ and commutes with $H_{R}$. Using
the second equality of (\ref{eq:135}), we have
\begin{align}
G_{\b_{0}}(\tau_{1},\tau_{2})= & \sum_{m_{2}}\int_{0}^{\pi}\f{d\t_{2}}{2\pi}\r(\t_{2})e^{-\tau_{12}E(\t_{2})}\f{(1-q)^{m_{2}/2}}{(q;q)_{m_{2}}}r^{m_{2}}H_{m_{2}}(\cos\t_{2}|q)\avg{\mB_{\mu}|e^{-(\b_{0}-\tau_{12})H_{R}}|0^{m_{2}}}\nonumber \\
= & \int_{0}^{\pi}\f{d\t_{1}d\t_{2}}{(2\pi)^{2}}\r(\t_{1})\r(\t_{2})e^{-(\b_{0}-\tau_{12})E(\t_{1})-\tau_{12}E(\t_{2})}\sum_{m_{2}}\f{r^{m_{2}}}{(q;q)_{m_{2}}}H_{m_{2}}(\cos\t_{2}|q)\nonumber \\
 & \times H_{m_{2}}(\cos\t_{1}|q)\sum_{n}\f{\mu^{n}(1-q)^{n/2}}{(q;q)_{n}}H_{n}(\cos\t_{1}|q)\nonumber \\
= & \int_{0}^{\pi}\f{d\t_{1}d\t_{2}}{(2\pi)^{2}}\f{\r(\t_{1})\r(\t_{2})(r^{2};q)_{\infty}}{(\mu(1-q)^{1/2}e^{\pm\i\t_{1}};q)_{\infty}(re^{\i(\pm\t_{1}\pm\t_{2})};q)_{\infty}}e^{-(\b_{0}-\tau_{12})E(\t_{1})-\tau_{12}E(\t_{2})}\label{app:202}
\end{align}
where in the last step we used \eqref{eq:app266}. To find the approximation
in the JT limit, we first rescale $\tau_{i}=\hat{\tau}_{i}/\sqrt{1-q}$
such that $\hat \tau_i$ is the same order as $\hat \b$, and take the so-called large $p$
limit \cite{Mukhametzhanov:2023tcg}
\begin{equation}
\lam\ra0,\quad\bar{\t}=(\t_{1}+\t_{2})/2\sim O(1),\quad\w\lam=\t_{1}-\t_{2}\sim O(\lam)\label{eq:app223}
\end{equation}
as well as the scaling (\ref{eq:app212}). The integrand in (\ref{app:202})
becomes
\begin{equation}
\f{(e^{\pm2\i\bar{\t}}q^{\pm\i\w},q,q,q^{2\D};q)_{\infty}e^{-\f{2\hat{\b}\cos\bar{\t}}{\lam}-\hat{\b}\cos\bar{\t}+(\hat{\b}-2\hat{\tau}_{12})\w\sin\bar{\t}}}{(-e^{\i\bar{\t}-\mu_{0}}q^{-\i\w/2},-e^{-\i\bar{\t}-\mu_{0}}q^{\i\w/2},e^{\pm2\i\bar{\t}}q^{\D},q^{\D\pm\i\w};q)_{\infty}}\label{eq:app204}
\end{equation}
where the Boltzman factor is expanded in the first two orders of $\lam$
in the exponent. By formulae in Appendix \ref{app:b}, we can approximate (\ref{app:202})
as
\begin{align}
G_{\b_{0}}(\tau_{1},\tau_{2}) & \app\sqrt{\f{\pi}{2\lam}}e^{\mu_{0}/2}\int\f{d\bar{\t}}{2\pi}\f{(2\sin\bar{\t})^{2\D+1}e^{-\hat{\b}\cos\bar{\t}}}{\sqrt{\cosh\f{\mu_{0}-\i\bar{\t}}2\cosh\f{\mu_{0}+\i\bar{\t}}2}}e^{\f{\Li_{2}(-e^{-\i\bar{\t}-\mu_{0}})+\Li_{2}(-e^{\i\bar{\t}-\mu_{0}})-(\pi-2\bar{\t})^{2}/2-2\hat{\b}\cos\bar{\t}}{\lam}}\nonumber \\
 & \int\f{d\w}{2\pi}\f{\G(\D+\i\w)\G(\D-\i\w)}{\G(2\D)}e^{\w((\hat{\b}-2\hat{\tau}_{12})\sin\bar{\t}+\bar{\t}/2)}\left[\f{\cosh\f{\mu_{0}+\i\bar{\t}}2}{\cosh\f{\mu_{0}-\i\bar{\t}}2}\right]^{\i\w/2}\label{eq:app225-2}
\end{align}
Note that the integral over $\bar{\t}$ is exactly the same as (\ref{eq:app213})
and is determined by the same saddle $\bar{\t}=\t_{s}$ and the 1-loop
determinant in $\lam\ra0$ limit. Divided by (\ref{eq:app213}), the
normalized two-point function is
\begin{align}
 & g_{\b_{0}}(\tau_{12})=\f{G_{\b_{0}}(\tau_{1},\tau_{2})}{\avg{\mu,\b_{0}/2|\mu,\b_{0}/2}}\nonumber \\
= & (2\sin\t_{s})^{2\D}\int\f{d\w}{2\pi}\f{\G(\D+\i\w)\G(\D-\i\w)}{\G(2\D)}\left(e^{-\i\left((\hat{\b}-2\hat{\tau}_{12})\sin\t_{s}+\arctan\f{\sin\t_{s}}{e^{\mu_{0}}+\cos\t_{s}}\right)}\right)^{\i\w}\nonumber \\
= & \left(\f{\sin\t_{s}}{\cos(\hat{\tau}_{12}\sin\t_{s}+\pi/2-\t_{s})}\right)^{2\D}\label{eq:app225}
\end{align}
where we used identity
\begin{equation}
\int\f{d\w}{2\pi}\f{\G(\D+\i\w)\G(\D-\i\w)}{\G(2\D)}z^{\i\w}=\left(\f z{(1+z)^{2}}\right)^{\D}\label{eq:184-1}
\end{equation}
A consistency check for (\ref{eq:app225}) is that it equals one when
$\hat{\tau}_{12}=0$. 

Now we take the semiclassical JT limit (\ref{eq:app215}) and analytic continue $\hat \tau_{12}\ra \i t_{12}/\d$, which leads
to
\begin{equation}
g_{\b_{0}}(\tau_{12})\ra\avg{\phi^\mu_{\b_0/2}|M_R(t_1)M_R(t_2)|\phi^\mu_{\b_0/2}}=\left(\f{(\d\pi/\b_{BH})}{\i\sinh(\f{\pi}{\b_{BH}}t_{12})}\right)^{2\D}
\end{equation}
where we used (\ref{eq:app208}). This is the thermal two-point function
of matter fields if we absorb a factor of $\d^{\D}$ into the definition of $M_R(t_i)$.

\subsection{Modular-flowed two-point function in thermal EoW state} \label{app:d4}

Let us consider the following modular-flowed Euclidean two-point function in a thermal EoW state
\begin{align}
G_{\b_{0}}(s;\tau_{12}) & =\avg{\mu,\b_{0}/2|e^{\tau_{1}H_{R}}\mB_{R,\mu}^{\i s}M_{R}\mB_{R,\mu}^{-\i s}e^{-(\tau_{1}-\tau_{2})H_{R}}M_{R}e^{-\tau_{2}H_{R}}|\mu,\b_{0}/2}\nonumber \\
 & =\avg{\mB_{\mu}|e^{-(\b_{0}-\tau_{12})H_{R}}\mB_{R,\mu}^{\i s}M_{R}\mB_{R,\mu}^{-\i s}e^{-\tau_{12}H_{R}}M_{R}|\w}
\end{align}
where the modular flow before or after Euclidean evolution does not
matter because $\mB_{R,\mu}$ is an operator function of $H_{R}$
and commutes with $H_{R}$. Here the modular operator is different from \eqref{4.21} by a Boltzmann factor $e^{-\b_0H_R}$, which can recovered by shifting $\tau_1\ra \tau_1-\i \b_0 s$ in the end. Since $\mB_{R,\mu}$ can be diagonalized
in energy basis, similar to (\ref{app:202}) we have
\begin{align}
G_{\b_{0}}(s;\tau_{12})= & \sum_{m_{2}}\int_{0}^{\pi}\f{d\t_{2}}{2\pi}\f{\r(\t_{2})e^{-\tau_{12}E(\t_{2})}(1-q)^{m_{2}/2}r^{m_{2}}H_{m_{2}}(\cos\t_{2}|q)}{(q;q)_{m_{2}}(\mu(1-q)^{1/2}e^{\pm\i\t_{2}};q)_{\infty}^{-\i s}}\avg{\mB_{\mu}|e^{-(\b_{0}-\tau_{12})H_{R}}\mB_{R,\mu}^{\i s}|0^{m_{2}}}\nonumber \\
= & \int_{0}^{\pi}\f{d\t_{1}d\t_{2}}{(2\pi)^{2}}\f{\r(\t_{1})\r(\t_{2})(r^{2};q)_{\infty}(\mu(1-q)^{1/2}e^{\pm\i\t_{1}};q)_{\infty}^{-\i s}e^{-(\b_{0}-\tau_{12})E(\t_{1})-\tau_{12}E(\t_{2})}}{(\mu(1-q)^{1/2}e^{\pm\i\t_{1}};q)_{\infty}(re^{\i(\pm\t_{1}\pm\t_{2})};q)_{\infty}(\mu(1-q)^{1/2}e^{\pm\i\t_{2}};q)_{\infty}^{-\i s}}\label{app:202-1}
\end{align}
In the JT limit, we first rescale $\tau_{i}=\hat{\tau}_{i}/\sqrt{1-q}$
and take scaling (\ref{eq:app223}) and (\ref{eq:app212}). The additional
modular flow related term becomes
\begin{equation}
\f{(-e^{\i\bar{\t}-\mu_{0}}q^{-\i\w/2},-e^{-\i\bar{\t}-\mu_{0}}q^{\i\w/2};q)_{\infty}^{-\i s}}{(-e^{\i\bar{\t}-\mu_{0}}q^{\i\w/2},-e^{-\i\bar{\t}-\mu_{0}}q^{-\i\w/2};q)_{\infty}^{-\i s}}\ra e^{2\i s\w\arctan\f{\sin\bar{\t}}{e^{\mu_{0}}+\cos\bar{\t}}}\label{eq:app204-2}
\end{equation}
which is an $O(1)$ phase and does not change the saddle equation.
Adding this piece to (\ref{eq:app225-2}) does not change the saddle $\bar{\t}=\t_{s}$ and the 1-loop determinant in $\lam\ra0$ limit. Divided by (\ref{eq:app213}), the normalized modular-flowed two-point function is
\begin{align}
 & g_{\b_{0}}(s;\tau_{12})=\f{G_{\b_{0}}(s;\tau_{12})}{\avg{\mu,\b_{0}/2|\mu,\b_{0}/2}}\nonumber \\
= & (2\sin\t_{s})^{2\D}\int\f{d\w}{2\pi}\f{\G(\D+\i\w)\G(\D-\i\w)}{\G(2\D)}\left(e^{-\i\left((\hat{\b}-2\hat{\tau}_{12})\sin\t_{s}+(2\i s+1)\arctan\f{\sin\t_{s}}{e^{\mu_{0}}+\cos\t_{s}}\right)}\right)^{\i\w}\nonumber \\
= & \left(\f{\sin\t_{s}}{\cos(\hat{\tau}_{12}\sin\t_{s}-\i s(2\t_{s}-\pi-\hat{\b}\sin\t_{s})+\pi/2-\t_{s})}\right)^{2\D}\label{eq:app225-1}
\end{align}
If we continue $\hat{\tau}_{12}\ra\i\hat{t}_{12}-\i\hat \b s$ in Lorentzian signature,
we immediately find that the modular flow is linear in time evolution
for $\hat{t}_{12}$ even before the semiclassical JT limit.

Now we take the semiclassical JT limit (\ref{eq:app215}) and continue $\hat \tau_{12}\ra \i (t_{12}-\b s)/\d$, which leads
to
\begin{equation}
g_{\b_{0}}(s;\tau_{12})\ra\avg{\phi^\mu_{\b_0/2}|\D_{\mu}^{\i s} M_R(t_1)\D_{\mu}^{-\i s} M_R(t_2) |\phi^\mu_{\b_0/2}}=\left(\f{(\d\pi/\b_{BH})}{\i\sinh\left[\f{\pi}{\b_{BH}}t_{12}-\pi s\right]}\right)^{2\D}
\end{equation}
This clearly shows that the modular flow is essentially the time evolution
with a rescaling of $\b_{BH}$. This justifies that the modular extension of the boundary algebra $\mK_{\mu,*}$ keeps it unchanged.

\subsection{Two-sided two-point function in thermal EoW brane state} \label{app:d5}

The two-sided Euclidean two-point function in a thermal EoW brane state
is
\begin{align}
G_{LR}(\tau_{1},\tau_{2}) & =\avg{\mu,\b_{0}/2|e^{\tau_{1}H_{L}}M_{L}e^{-\tau_{1}H_{L}+\tau_{2}H_{R}}M_{R}e^{-\tau_{2}H_{R}}|\mu,\b_{0}/2}\nonumber \\
 & =\avg{\w|\mB_{R,\mu}^{1/2}e^{-(\b_{0}/2+\bar{\tau}_{12})H_{R}}M_{R}\mB_{R,\mu}^{1/2}e^{-(\b_{0}/2-\bar{\tau}_{12})H_{R}}M_{R}|\w}
\end{align}
where $\bar{\tau}_{12}\equiv\tau_{1}+\tau_{2}$. Using the second
equality of (\ref{eq:135}), we have
\begin{align}
&G_{LR}(\tau_{1},\tau_{2})\nn\\
=& \sum_{m_{2}}\int_{0}^{\pi}\f{d\t_{2}}{2\pi}\f{\r(\t_{2})e^{-(\b_{0}/2-\bar{\tau}_{12})E(\t_{2})}(1-q)^{m_{2}/2}}{(q;q)_{m_{2}}(\mu(1-q)^{1/2}e^{\pm\i\t_{2}};q)_{\infty}^{1/2}}r^{m_{2}}H_{m_{2}}(\cos\t_{2}|q)\avg{\w|\mB_{R,\mu}^{1/2}e^{-(\b_{0}/2+\bar{\tau}_{12})H_{R}}|0^{m_{2}}}\nonumber \\
= & \int_{0}^{\pi}\f{d\t_{1}d\t_{2}}{(2\pi)^{2}}\f{\r(\t_{1})\r(\t_{2})e^{-(\b_{0}/2+\bar{\tau}_{12})E(\t_{1})-(\b_{0}/2-\bar{\tau}_{12})E(\t_{2})}}{(\mu(1-q)^{1/2}e^{\pm\i\t_{2}};q)_{\infty}^{1/2}(\mu(1-q)^{1/2}e^{\pm\i\t_{1}};q)_{\infty}^{1/2}}\nonumber \\
 & \times\sum_{m_{2}}\f{r^{m_{2}}}{(q;q)_{m_{2}}}H_{m_{2}}(\cos\t_{2}|q)H_{m_{2}}(\cos\t_{1}|q)\nonumber \\
= & \int_{0}^{\pi}\f{d\t_{1}d\t_{2}}{(2\pi)^{2}}\f{\r(\t_{1})\r(\t_{2})(r^{2};q)_{\infty}e^{-(\b_{0}/2+\bar{\tau}_{12})E(\t_{1})-(\b_{0}/2-\bar{\tau}_{12})E(\t_{2})}}{(\mu(1-q)^{1/2}e^{\pm\i\t_{2}};q)_{\infty}^{1/2}(\mu(1-q)^{1/2}e^{\pm\i\t_{1}};q)_{\infty}^{1/2}(re^{\i(\pm\t_{1}\pm\t_{2})};q)_{\infty}}
\end{align}
where in the last step we used (\ref{eq:app266}). Comparing with
(\ref{app:202}), the main difference is an additional factor
\begin{equation}
\f{(\mu(1-q)^{1/2}e^{\pm\i\t_{1}};q)_{\infty}^{1/2}}{(\mu(1-q)^{1/2}e^{\pm\i\t_{2}};q)_{\infty}^{1/2}}=\f{(-e^{\i\bar{\t}-\mu_{0}}q^{-\i\w/2},-e^{-\i\bar{\t}-\mu_{0}}q^{\i\w/2};q)_{\infty}^{1/2}}{(-e^{\i\bar{\t}-\mu_{0}}q^{\i\w/2},-e^{-\i\bar{\t}-\mu_{0}}q^{-\i\w/2};q)_{\infty}^{1/2}}\ra e^{-\w\arctan\f{\sin\bar{\t}}{e^{\mu_{0}}+\cos\bar{\t}}}
\end{equation}
where we have taken the large $p$ limit (\ref{eq:app223}). This
additional term does not change the saddle. It follows that the normalized
two-sided two-point function is 

\begin{align}
 & g_{LR}(\bar{\tau}_{12})=\f{G_{LR}(\tau_{1},\tau_{2})}{\avg{\mu,\b_{0}/2|\mu,\b_{0}/2}}\nonumber \\
= & (2\sin\t_{s})^{2\D}\int\f{d\w}{2\pi}\f{\G(\D+\i\w)\G(\D-\i\w)}{\G(2\D)}\left(e^{-2\i\hat{\bar{\tau}}_{12}\sin\t_{s}}\right)^{\i\w}\nonumber \\
= & \left(\f{\sin\t_{s}}{\cos(\hat{\bar{\tau}}_{12}\sin\t_{s})}\right)^{2\D}\label{eq:app225-3}
\end{align}
where $\hat{\bar{\tau}}_{12}\equiv\bar{\tau}_{12}\sqrt{1-q}$. Taking
the semiclassical JT limit (\ref{eq:app215}) and analytic continuation
$\hat{\bar{\tau}}_{12}\ra\i(t_{1}+t_{2})/\d$ leads to 
\begin{equation}
g_{LR}(\bar{\tau}_{12})\ra\avg{\phi_{\b_{0}/2}^{\mu}|M_{L}(t_{1})M_{R}(t_{2})|\phi_{\b_{0}/2}^{\mu}}=\left(\f{(\d\pi/\b_{BH})}{\cosh\f{\pi}{\b_{BH}}(t_{1}+t_{2})}\right)^{2\D}
\end{equation}
This shows that the two-sided two-point function is nothing but the
two-sided correlation in a thermofield double state with inverse temperature
$\b_{BH}$.

\subsection{General $2n$-point functions in a thermal EoW brane state}
For a generic $2n$-point function in regular DSSYK model, the Feynmann rule is given in \cite{Berkooz:2018jqr}. From above computations, it is easy to see the additional rule for correlations in a thermal EoW brane state is just adding the factor $\f{1}{(\mu(1-q)^{1/2}e^{\pm\i \t_i})}$ to one of the Boltzmann factor $e^{-\tau_i E(\t_i)}$ in the computation. In the large $p$ limit, suppose the corresponding Boltzmann factor is parameterized as $\t_i=\bar \t+\w_i \lam/2$ where $\bar \t,\w_i\in O(1)$, we expand this factor as 
\be 
\f{1}{(\mu(1-q)^{1/2}e^{\pm\i\t_i};q)_{\infty}}\sim \f {e^{\f{\Li_{2}(-e^{-\i\bar{\t}-\mu_{0}})+\Li_{2}(-e^{\i\bar{\t}-\mu_{0}})}{\lam}}e^{\w_i \arctan \f{\sin\bar \t}{e^{\mu_0}+\cos \bar \t}}} {\sqrt{(1+e^{-\i\bar{\t}-\mu_{0}})(1+e^{\i\bar{\t}-\mu_{0}})}}
\ee
where the first exponential deforms the saddle in the same way as before, and the second exponential will enter the matter correlation function as a shift of Euclidean time. Following the idea of \cite{Mukhametzhanov:2023tcg} and \cite{Jafferis:2022uhu}, one can show that the an un-crossed $2n$-point function becomes generalized free fields and a crossed $2n$-point function is suppressed by $\lam$ for Lorentzian time scale $t\sim O(1)$ and thus vanishes in the semiclassical JT limit.

\subsection{The quantum JT gravity limit} \label{app:d6}

Here we briefly mention how to get the quantum JT gravity limit, where matter is coupled to the quantum Schwarzian mode \cite{Kolchmeyer:2023gwa}. We just need to keep the parameter $\phi_b$ in \eqref{app:d19} fixed while taking $\d,\lam\ra 0$ limit. For example, the inner product \eqref{eq:app211} has the integrand under this limit as
\begin{align} 
\f{(e^{\pm2\i\t},q;q)_{\infty}e^{-\b_{0}E(\t)}}{(\mu(1-q)^{1/2}e^{\pm\i\t};q)_{\infty}}=\f{(q^{\pm 2\i\phi_b k},q;q)_\infty e^{\f{\b\cos k \lam \phi_b}{\phi_b \lam(1-q)}}}{(q^{(\tilde \mu\pm \i k)\phi_b};q)_\infty} \nn\\
\sim e^{-\b k^2 \phi_b^2/2}\phi_b k\sinh (2\pi \phi_b k)|\G(\phi_b(\tilde \mu-\i k))|^2
\end{align}
where we used \eqref{eq:a164}. Rescaling $k\ra k/\phi_b$, $\tilde \mu \ra (\mu_r+1/2)/\phi_b$ exactly leads to the quantum JT result \eqref{eq:app204-1}. Including matter correlations, the computation should be similar, which we leave for a future examination.

\section{Boundary Algebra in the Free Case} \label{app:e}
\subsection{An infinite series for the commutant}
In this subsection we show that in the $Q_{ij}=0$ case we can construct the operator~\eqref{4.2} explicitly from a series expansion. The strategy is to find an infinite series representation of operators $a_{L,0}$ and $a_{L,0}^\dagger$, and therefore $\tilde H_{L,0}$, in terms of $a_{R,0}$ and $a_{R,1}$ and deform each term using the isomorphism $\td$. In the following for notational simplicity we use the abbreviation $a_{R,i}\equiv i$, $a_{R,i}^\dagger\equiv i^\dagger$. In expressions where we need to emphasize operators with given $i$ we denote $a_{R,0}=a$, $a_{R,1}=b$, $a_{L,0}=a_L$, $H_{R,0}=H_R$ and $H_{R,1}=M_R$.

For the free case the wormhole length operator $\bar q^n$ reduces to the vacuum projector $P_\w$ and thus the $\td$ isomorphism acts as $\tilde a=a+\mu P_\w$. First we write $a_L$ and $a_L^\dagger$ in terms of normal ordered series in $a, a^\dagger$ and $b,b^\dagger$. Since $a_L^\dagger$ creates a chord, we expect that all terms in the series should contain one more creation operator than annihilation operators, so we can arrange the series according to the number of annihilation operators, and we define the order $n$ of a term to be the number of annihilation operators it contains. Note that, partial sums up to order $n$ completely fix the action of a series on states with chord numbers $N\leq n$, as higher order terms annihilate all states with $N\leq n$. Therefore, we could construct the series expansion of $a_L$ in an inductive way: suppose we have already fixed the series up to order $n$, then by considering the action on chord states with $N=n+1$ we could fix the order $n+1$ terms. Thus we could obtain a general expression valid to all orders. Now we start from the order $0$ term, for the series to have correct action on $\ket{\w}$, it has to start with $a^\dagger$, that is $a_L^\dagger=a^\dagger+\ldots$. 

Now suppose we have fixed the partial sum $a_L^{\dagger(n)}$ up to order $n$, consider an arbitrary chord state $\ket{w^{(n+1)}}=\ket{i_1i_2\ldots i_{n+1}}$ with $n+1$ chords, clearly $a_L^{\dagger(n)}$ acts on such a state as 
\begin{equation}\label{app:e1}
    a_L^{\dagger(n)}\ket{w^{(n+1)}}=\ket{i_10i_2\ldots i_{n+1}}
\end{equation}
If $i_1=0$, the action of $a_L^{\dagger(n)}$ on $\ket{w^{(n+1)}}$ is already as expected, and thus no more term needs to be added. But if $i_1=1$, then $a^{(n)}_L\ket{w^{(n+1)}}=\ket{10i_2\ldots i_{n+1}}$, which disagrees with the expected result $a^\dagger_L\ket{w^{(n+1)}}=\ket{01i_2\ldots i_{n+1}}$. In this case we need to add the following term to get the correct action on $\ket{w^{(n+1)}}$
\begin{equation}\label{app:e2}
   i^\dagger_{n+1}\ldots i_2^\dagger b^\dagger a^\dagger b i_2\ldots i_{n+1}-i^\dagger_{n+1}\ldots i_2^\dagger a^\dagger b^\dagger b i_2\ldots i_{n+1}
\end{equation}
It is easy to verify that this extra term ensures the correct action on the state $\ket{w^{(n+1)}}$. 

We can repeat this procedure for all states $\ket{w^{(n+1)}}$ with $i_1=1$, notice that extra terms can be added for each $\ket{w^{(n+1)}}$ independently as in the free case different chord states are orthogonal in the sense that
\begin{equation}\label{app:e3}
    i_{1}i_{2}\ldots i_n\ket{i_1'i_2'\ldots i_n'}=\delta_{i_1i_1'}\delta_{i_2i_2'}\ldots \delta_{i_ni_n'}
\end{equation}
thus the extra term added for $\ket{w^{(n+1)}}$ does not affect the action on $\ket{w'^{(n+1)}}$ with $w\neq w'$. Collecting all extra terms we obtain the $n+1$'th order term
\begin{equation}\label{app:e4}
    a_L^{\dagger(n+1)}=a_L^{\dagger(n)}+\sum_{w^{(n+1)},s_1=1} i^\dagger_{n+1}\ldots i_2^\dagger b^\dagger a^\dagger b i_2\ldots i_{n+1}-i^\dagger_{n+1}\ldots i_2^\dagger a^\dagger b^\dagger b i_2\ldots i_{n+1} 
\end{equation}
Hence we get the following series expansion of $a_L$ in terms of right operators
\begin{equation}\label{app:e5}
    a^\dagger_L=a^\dagger+\sum_{w,s_1=1}i^\dagger_{n+1}\ldots i_2^\dagger b^\dagger a^\dagger b i_2\ldots i_{n+1}-i^\dagger_{n+1}\ldots i_2^\dagger a^\dagger b^\dagger b i_2\ldots i_{n+1} 
\end{equation}

Now we replace all operators with their tilded counterparts to get the new operator $\tilde a_{L+}\equiv \td(a_L^\dagger)$\footnote{Here we use the notation $\tilde a_{L+}$ instead of $\tilde a_{L}^\dagger$ to emphasize that $\tilde a_{L+}$ is not the conjugate operator of $\tilde a_L$.} This is done by simply replacing all $a$ with $a+\mu P_\w$ in the series~\eqref{app:e5}. In the free case we have obvious relations $aP_\w=bP_\w=P_\w a^\dagger=P_\w b^\dagger=0$, which drastically simplify the result. Since all terms (except the first one) in~\eqref{app:e5} contains a sequence $bi_2\ldots i_n$, it turns out that all extra terms including $\mu P_\w$ vanish as they will eventually hit the $b$ operator to their left. So we conclude that
\begin{equation}\label{app:e6}
    \tilde a_{L+}=a_L^\dagger
\end{equation}

Next we turn to $a_L$ and the corresponding $\td(a_L)\equiv\tilde a_{L}$, which is more complicated than the creation operator and contains a non-trivial extra contribution. By the same argument as above we see that each term in the series representing $a_L$ contains one less creation operator than annihilation operator, and the series starts with $a$, that is $a_L=a+\ldots$, now suppose we have already constructed the partial sum $a^{(n)}_L$ up to order $n$. Then we look at its action on $\ket{w^{(n+1)}}$. Clearly $a^{(n)}_L$ annihilates the second chord from the left if it is a gravity chord and annihilates the state if the second chord from the left is a matter chord. Therefore, if $\ket{w^{n+1}}=\ket{00i_3i_4\ldots i_{n+1}}$ or $\ket{11i_3i_4\ldots i_{n+1}}$, $a^{(n)}_L$ already has the correct action and no extra term needs to be added. For each state of the form $\ket{w^{n+1}}=\ket{10i_3i_4\ldots i_{n+1}}$, $a_{L}^{(n)}\ket{w^{(n+1)}}=\ket{1i_3i_4\ldots i_{n+1}}$ while what we expect is that $a_L\ket{w^{(n+1)}}=0$, so we need to add the following extra term to cancel the contribution from $a^{(n)}_L$
\begin{equation}\label{app:e7}
    -i_{n+1}^\dagger\ldots i_4^\dagger i_3^\dagger b^\dagger bai_3i_4\ldots i_{n+1}
\end{equation}
On the other hand, for states of the form $\ket{w^{(n+1)}}=\ket{01i_3i_4\ldots i_{n+1}}$, $a^{(n)}_L\ket{w^{(n+1)}}=0$ while what we expect is $a_L\ket{w^{(n+1)}}=\ket{1i_3i_4\ldots i_{n+1}}$, so we have to add the following term
\begin{equation}\label{app:e8}
    i_{n+1}^\dagger\ldots i_4^\dagger i_3^\dagger b^\dagger abi_3i_4\ldots i_{n+1}
\end{equation}

Again we repeat this procedure for each $\ket{w^{(n+1)}}$, and again for different $\ket{w^{(n+1)}}$ extra terms are added independently. We thus obtain the following series expansion for $a_L$ in terms of right operators
\begin{equation}\label{app:e9}
    a_L=a-\sum_{\substack{w\\i_1=1\\i_2=0}}i_{n+1}^\dagger\ldots i_4^\dagger i_3^\dagger b^\dagger bai_3i_4\ldots i_{n+1}+
    \sum_{\substack{w\\i_1=0\\i_2=1}}i_{n+1}^\dagger\ldots i_4^\dagger i_3^\dagger b^\dagger abi_3i_4\ldots i_{n+1}
\end{equation}
here the sum is over all chord states with two or more chords. 

Next we replace all operators with tilded ones. Note that replacing $a$ with $\tilde a$ in the second term on the r.h.s. of~\eqref{app:e9} does not result in extra terms as all projectors $P_\w$ eventually hit the annihilation operator $b$. For the third term, the only non-trivial contribution is 
\begin{equation}\label{app:e10}
    \mu\sum_{\substack{w\\i_1=0\\i_2=1}}i_{n+1}^\dagger\ldots i_4^\dagger i_3^\dagger b^\dagger P_\w bi_3i_4\ldots i_{n+1}=\mu\sum_{w,i_1=1} P_w
\end{equation}
here we have used the fact that $i_{n+1}^\dagger\ldots i_4^\dagger i_3^\dagger b^\dagger P_\w bi_3i_4\ldots i_{n+1}=P_{w'}$ where $P_{w'}$ is the projector onto the chord state specified by the string $w'=bi_3\ldots i_{n+1}$. It turns out that the sum on the l.h.s. of~\eqref{app:e10} over all $w$ satisfying $i_1=0$, $i_2=1$ can be turned into a sum over all $w$ with $i_1=1$ on the r.h.s., which is just the sum over all chord states where the leftmost chord is a $M$ chord. We thus have
\begin{equation}\label{app:e11}
    \tilde a_L=a_L+\mu P_\w+\mu P_b
\end{equation}
where we have defined $P_b$ to be the projector onto the subspace spanned by all chord states with the leftmost chord a matter chord. 

Collecting all results above we conclude that the following operator 
\begin{equation}\label{app:e12}
    \tilde H_L=a_L+a_L^\dagger+\mu(P_\w+P_b)
\end{equation}
should be in the commutant $\{\tilde H_R,\tilde M_R\}'$. In fact the projectors can be written in terms of left creation and annihilation operators as it can be easily verified that
\begin{equation}\label{app:e13}
    P_\w+P_b=1-a_L^\dagger a_L
\end{equation}
We can ignore the constant part and redefine $\tilde H_L$ as 
\begin{equation}\label{app:e14}
    \tilde H_L=a_L+a_L^\dagger-\mu a_L^\dagger a_L
\end{equation}
It is very easy to verify that this operator indeed commutes with both $\tilde H_R$ and $\tilde M_R$ using the commutation relations~\eqref{2.19} with $q^{N_0}r^{N_1}$ replaced by $P_\w$ in the free case. $[\tilde H_L,\tilde M_R]=0$ is obvious, so we only need to check $[\tilde H_L,\tilde H_R]=0$:
\begin{align}\label{app:e15}
    [\tilde H_L,\tilde H_R]&=[a_L+a_L^\dagger-\mu a^\dagger_La_L,a_R+a_R^\dagger+\mu P_\w]\nonumber\\
    &=[a_L+a_L^\dagger,a_R+a_R^\dagger]-\mu[a_L^\dagger a_L,a_R+a_R^\dagger]+\mu[a_L+a_L^\dagger,P_\w]-\mu^2[a^\dagger_La_L,P_\w]\nonumber\\
    &=\mu(P_\w a_L-a_L^\dagger P_\w+a^\dagger_LP_\w-P_\w a_L)=0
\end{align}
where we have used the obvious fact that $a_LP_\w=P_\w a_L^\dagger$ for the free case. Thus we have found a non-trivial element of the commutant in the free case. 

\subsection{The proof of Theorem \ref{thm7} for the free case}
\label{freecaseproof}
\begin{proof} 
    For the case where $q=0$ or $r=0$ the proof of Theorem~\ref{thm7} needs to be slightly modified, as we cannot directly get the spectral projections in the same way as above. We will discuss two cases, either $q=r=0$ or $q\neq 0$ but $r=0$. The former is just the free case while the latter is the case where we have infinitely heavy matter particles. We first consider the free case where $Q_{ij}=0$. In this case we have $aP_\w=bP_\w=P_\w a^\dagger =P_\w b^\dagger=0$, $aa^\dagger=bb^\dagger=1$ and $ab^\dagger=ba^\dagger=0$. With these relations, by repeatedly multiplying $H_R$ and $M_R$ operators on the left and right of $P_\w$ and taking linear combinations we can generate all operators of the form
    \begin{equation}
        i_1^\dagger i_2^\dagger\ldots i_m^\dagger P_\w i_1i_2\ldots i_n
    \end{equation}
    Typically we can generate projections to each chord state since these projectors also take the above form. Thus we can still generate $B(\mathcal H)$.
    
    Next we consider the case where $r=0$ but $q\neq 0$. In this case $q^{N_0}r^{N_1}=q^{N_0}P_{N_1=0}$ where the projector is onto the subspace spanned by states with no matter chord. Now the spectrum of $q^{N_0}r^{N_1}$ is just $\{0\}\cup\{q^{n}\}$ where $n\geq 0$. By the spectral theory we can obtain the projection onto states with different numbers $n$ (including zero) of gravity chords and no matter chords, which we denote as $P_{0,n}$. Now we show that we can construct projectors to all states with these operators together with $H_R$ and $M_R$. To do this we first note that an arbitrary chord state can be written as 
    \begin{equation}
        \ket{w}=\ket{i_1i_2\ldots i_n}=|\underbrace{000\ldots 0}_{n_1\text{ chords}} w'\rangle
    \end{equation}
    where $w'=i_{n_1+1}i_{n_1+2}\ldots i_{n}$. In the case where $n_1=0$ we have $w'=w$. Note that as the crossing factor between gravity and matter chords is $0$, the projector onto $\ket{w}$ can be written as $P_w\sim w'^\dagger P_{0,n_1}w'$, where we write $\sim$ as there is in general some non-trivial coefficient depending on the configuration of gravity chords in $w'$. Now using the relation $P_{0,n}b^\dagger=b P_{0,n}=0$ and $ab^\dagger=ba^\dagger=0$, we can generate all operators of the form $w'^\dagger P_{0,n_1}w'$ by repeatedly multiplying $H_R$ and $M_R$ to be left and right of $P_{0,n_1}$ and taking linear combinations. Thus we can get projections to arbitrary chord states and generate the whole algebra as before.
\end{proof}

\section{Chord Diagrammatics on Trumpets}\label{app:f}
\subsection{Vacuum bubbles on the trumpet}
In this appendix we discuss detailed structures of chord diagrams on a trumpet. We start from the distinction between boundary chords, EoW brane chords and vacuum chords. Throughout this section we will write $\tilde H_{R,0}=H$ and $\tilde H_{R,1}=M$ for notational simplicity. In the following we will denote chord diagrams on the disk with a defect `defect chord diagrams'.
\begin{definition}
    In a defect chord diagram, an EoW brane chord is a chord connected to the defect. A boundary chord is one connected to the boundary on both ends. A vacuum chord is one which is not connected to the boundary. The connected part of the diagram is defined as the collection of all boundary chords and EoW brane chords. The vacuum bubble of the diagram is defined as the collection of all vacuum chords. A diagram is called connected if it contains no vacuum bubbles. 
\end{definition}

In Figure~\ref{diagonal-element-warpped} we see that the vacuum bubble consists of concentric loops which intersects only the EoW brane chords but not the boundary chords, and do not self-intersect. This is not a coincidence, as guaranteed by the following lemma
\begin{figure}
	\centering
		\includegraphics[width=0.4\textwidth]{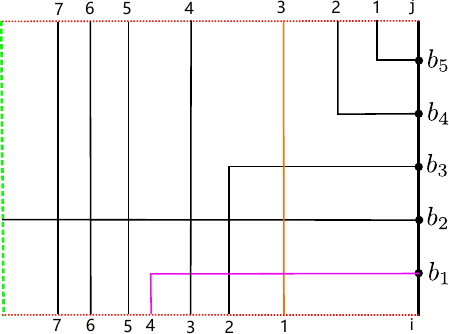}
        \caption{A labeled chord diagram with defect.}
        \label{labeled-diagonal-element}
\end{figure}
\begin{lemma}
    In an arbitrary defect chord diagram, chords in the vacuum bubble part neither self-intersect nor intersect boundary chords. In other words, the vacuum bubble part always consists of a set of concentric loops `floating' in the middle, which intersects only EoW brane chords. 
\end{lemma}
\begin{proof}
    We first prove that the vacuum bubble consists only of non-intersecting concentric loops. An arbitrary defect chord diagram can be sliced open as shown in Figure~\ref{labeled-diagonal-element} with incoming and outgoing open chords on the bottom and top slices. We label these open chords and boundary operator insertions as illustrated in Figure~\ref{labeled-diagonal-element}. Incoming and outgoing open chords with identical labels are identified. Now we can specify each chord by tracking the labels, for example one of the boundary chords can be represented as 
    \begin{equation}
        b_5\rightarrow1\rightarrow 3\rightarrow 4\rightarrow b_1
    \end{equation}
    Another example is the vacuum chord $7\rightarrow 7$. Here the arrows denote the lines in the diagram which connects bottom open chords to top ones. For example the orange line in Figure~\ref{labeled-diagonal-element} is denoted as $1\rightarrow 3$. Note that, since new chords are created on the right of all existing chords, the top label of each line is always greater than its bottom label. For example the orange line above has top label $3$ and bottom label $1$. Since vacuum chords do not intersect the boundary, they have to be represented as
    \begin{equation}
        i_1\rightarrow i_2\rightarrow i_3\ldots \rightarrow i_1
    \end{equation}
    But we have $i_1\leq i_2\leq i_3\ldots$, so the only possible vacuum chord is given by $i_1\rightarrow i_1$, which is a closed loop, and clearly loops for different $i$'s never cross, so the vacuum bubble part of the diagram is always a set of concentric loops. Also, from this argument we see that all lines $i\rightarrow j$ with $i<j$ must be part of a boundary chord.

    Now we proceed to the proof that vacuum bubbles never intersect boundary chords. To do this we consider the leftmost incoming chord annihilated on the boundary, denoted as $i_L\rightarrow b_L$. In Figure~\ref{labeled-diagonal-element} this is  the pink line $4\rightarrow b_1$. Obviously all lines to the right of the leftmost chord are either connected to the boundary or take the form $i\rightarrow j$ with $i<j$, so none of them could be a vacuum chord. On the other hand, all lines to the left of the leftmost chord must take the form $i\rightarrow i$ and thus form vacuum chords, which do not intersect any boundary chord. Therefore, the vacuum bubble part of the diagram is always a set of concentric loops which intersects only EoW chords but not boundary chords. 
\end{proof}

With the discussion above we can now discuss the diagrammatic representation of the naive trace~\eqref{eq:5.1}. For a monomial $H^k$, $\Tr(H^k)$ is just given by the sum of all defect chord diagrams with $k$ insertions on the boundary. We can group these diagrams according to the numbers of EoW brane chords, denoted as $b$. Diagrams with same $b$'s can be further grouped according to their connected parts, as two different diagrams could have identical connected parts but differ in the numbers of vacuum chords. A diagram with $b$ EoW brane chords and $n_v$ vacuum chords can be written as $D=\mu^bq^{bn_v}D_{c}$, where $D_c$ is the value of the connected part of the diagram\footnote{Here we have extracted the factor $\mu^b$ out of $D_c$ to emphasize that we are grouping the diagrams according to $b$.}. Here the factor $q^{bn_v}$ comes from the fact that once a vacuum chord is added into a diagram with $b$ EoW brane chords it creates $b$ crossings. With the notations defined above the naive trace can be written as
\begin{align}\label{app:f1}
    \Tr(H^k)&=\sum_{b=0}^{\infty}\mu^b\sum_{n_v=0}^{\infty}q^{bn_v}\sum_{c} D_c(k,b)\nonumber\\
    &=\sum_{b=0}^{\infty}\frac{\mu^b}{1-q^{b}}\sum_c D_c(k,b)
\end{align}
where $D_c(k,b)$ denotes connected diagrams with $k$ boundary insertions and $b$ EoW brane chords, note that brane tension factors are not included in $D_c(k,b)$ as mentioned above. The sum~\eqref{app:f1} takes the same form as the trumpet amplitude discussed in~\cite{Okuyama:2023byh}. 

The leading term $b=0$ is divergent, which now has a diagrammatic origin: for a connected diagram with no EoW brane chord, we can add an arbitrary number of vacuum chords without introducing extra suppressing factors, which leads to a divergence. On the contrary, for a connected diagram with at least one EoW brane chord, the factors $q^{bn_v}$ converge and thus there is no divergence. This divergence can be understood as the discrete analog of the divergent trumpet amplitude in~\cite{Gao:2021uro} which goes as $\lim_{b\rightarrow0}\frac{1}{b}$ in the limit of zero brane length. In fact if we take $q=e^{-\Delta}$ and take the continuum limit where $b\Delta\ll1$, the factor $\frac{1}{1-q^b}$ also goes as $\frac{1}{\Delta b}$ as in JT gravity. Now we see the advantage of the diagrammatic approach in the  DSSYK model: we can unambiguously identify the set of diagrams causing the divergence, which leads to natural ways to remove the divergence, instead of a hard cutoff in the JT gravity case. 

Before moving on to regularization of the trumpet amplitude, we should point out that, although seemingly impossible, the sum $\sum_c D_c(k,b)$ can in fact be calculated using the expression for the naive trace introduced in~\cite{Okuyama:2023byh}. We first prove that $\sum_c D_{c}(k,b)$ is finite despite that it contains an infinite number of diagrams, as the winding number of each boundary chord can be arbitrary. For a given diagram consider all chords created on the boundary, for example in Figure~\ref{labeled-diagonal-element} there are two such chords $b_5\rightarrow 1$ and $b_4\rightarrow 2$, this collection of chords then winds around the defect, and finally each chord is annihilated on the boundary with some winding number. In our example the chord $b_5\rightarrow 1$ is annihilated on $b_1$ with winding number $2$ and $b_4\rightarrow 2$ is annihilated on $b_3$ with winding number $0$. Now consider the set of all diagrams with boundary chords created and annihilated at the same points on the boundary, but each of them may have different winding numbers. For example we consider all diagrams with boundary chords of the form $b_5\rightarrow 1\rightarrow \ldots\rightarrow b_1$ and $b_4\rightarrow 2\rightarrow \ldots\rightarrow b_3$. Since the number of boundary insertions is finite, there is only a finite number of such sets. Now let $n_1,n_2\ldots n_k$ be the winding numbers of each boundary chord in a given diagram with $k$ boundary chords, it is clear that the diagram converges at least as fast as $q^{n_1+n_2+\ldots n_k}$ since such a diagram contains at least $\sum n_i$ crossings from self-intersections of chords with non-trivial winding numbers. Therefore the sum of all diagrams in each set converges. Since the number of sets is also finite, the sum $\sum_c D_c(k,b)$ converges as well. 

In~\cite{Okuyama:2023byh} the naive trace is calculated to be
\begin{align}\label{app:f2}
    \Tr(H^k)&=\sum_{n} (H^k)_{nn}\nonumber\\
    &=\int_0^\pi\frac{d\theta}{2\pi}\frac{\rho(\theta)}{(\tilde \mu e^{\pm i\theta};q)_\infty} E^k(\theta)|\tilde \psi_n(\theta)|^2\nonumber\\
    &=\sum_{b=0}\frac{\tilde\mu^b}{1-q^b}\int_0^\pi \frac{d\theta}{2\pi}E^k(\theta)2\cos(b\theta)
\end{align}
where in the second step we have inserted the completeness relation of the $\ket{\t}$ basis and in the third step we have used the resummation formula~\eqref{eq:app266}. Since $\tilde \mu=\mu\sqrt{1-q}$, we can identify $D_c(k,b)$ with the terms multiplying $\frac{\mu^b}{1-q^b}$ in the sum above. We can do some sanity checks to see that the expression makes sense. Plugging~\eqref{eq:3.10} into the integrals above, we immediately see that $D_c(k,b)=0$ for $k<b$, as $k$ boundary insertions cannot produce more than $k$ EoW brane chords. We also have $D_c(k,k)=1$, as there is only one way for $k$ boundary insertions to give $k$ EoW brane chords, that is all of them should be connected to the EoW brane. Another observation is that $D_c(k,b)\neq 0$ only when $k\equiv b\mod 2$, this corresponds to the simple fact that only in this case the number of incoming and outgoing chords can be the same. With the discussion above we can now proceed to the regularization of the trumpet amplitudes to obtain renormalized traces which are finite. 

\subsection{Regularization by overall normalization}
The most naive normalization procedure we can think of is to simply divide by a divergent overall normalization constant, which is the naive trace of the unit operator $\Tr (\mathbb I)$. Since $\mathbb I$ has no boundary insertions, diagrammatically this trace contains an infinite sum of diagrams with arbitrary numbers of vacuum bubbles:
\begin{equation}
\tikz[baseline=(n.base)] {
    \node (n) {$\Tr(\mathbb I)=$};
    \fill[green] (2,0) circle (0.1);
    \draw[thick] (2,0) circle (1);
    \node at (3.5,0) {$+$};
    \fill[green] (5,0) circle (0.1);
    \draw[thick] (5,0) circle (1);
    \draw (5,0) circle (0.8);
    \node at (6.5,0) {$+$};
    \fill[green] (8,0) circle (0.1);
    \draw[thick] (8,0) circle (1);
    \draw (8,0) circle (0.8);
    \draw (8,0) circle (0.6);
    \node at (9.5,0) {$+$};
    \fill[green] (11,0) circle (0.1);
    \draw[thick] (11,0) circle (1);
    \draw (11,0) circle (0.8);
    \draw (11,0) circle (0.6);
    \draw (11,0) circle (0.4);
    \node at (12.6,0) {$+\ldots$};
}
\end{equation}

Now we formally divide the naive trace by the trace of identity to form a renormalized trace
\begin{equation}\label{app:f3}
    \tr(H^k)=\frac{\Tr(H^k)}{\Tr(\mathbb I)}
\end{equation}
We expect this definition to preserve positivity of the trace as it is the quotient of two divergent positive quantities. This division has a diagrammatic explanation. First note that all terms with $b\neq 0$ in~\eqref{app:f2} is finite while $\Tr(\mathbb I)$ is divergent, thus they all vanish after this normalization. Therefore we are left with only $b=0$ terms, meaning that the length of the EoW brane is set to zero, which can be understood as a `regularized conical amplitude'. From previous discussion we know that the $b=0$ sector of the naive trace can be written as~\eqref{app:f2}
\begin{equation}
    \tikz[baseline=(n.base)] {
    \node (n) {$\sum_c D_c(k,b=0)\times$};
    \node at (1.7,0) {$\bigg\{$};
    \fill[green] (3,0) circle (0.1);
    \draw[thick] (3,0) circle (1);
    \node at (4.5,0) {$+$};
    \fill[green] (6,0) circle (0.1);
    \draw[thick] (6,0) circle (1);
    \draw (6,0) circle (0.8);
    \node at (7.5,0) {$+$};
    \fill[green] (9,0) circle (0.1);
    \draw[thick] (9,0) circle (1);
    \draw (9,0) circle (0.8);
    \draw (9,0) circle (0.6);
    \node at (10.8,0) {$+\ldots\bigg\}$};
    }
\end{equation}
Thus the renormalization~\eqref{app:f3} just cancels the sum of all vacuum chords, which gives
\begin{equation}
    \tr(H^k)=\frac{\Tr(H^k)}{\Tr(\mathbb I)}=\sum_c D_c(k,b=0)
\end{equation}
by taking the $b=0$ term in~\eqref{app:f2} this sum is given by the simple integral
\begin{equation}
    \sum_c D_c(k,b=0)=\int_0^\pi\frac{d\theta}{\pi} E^k(\theta)
\end{equation}
or more generally, for an arbitrary function $f(H)$ we have
\begin{equation}
    \tr[f(H)]=\int_0^\pi\frac{d\theta}{\pi}f[E(\theta)]
\end{equation}
We can now see that this trace is indeed positive definite as the integration measure is a positive constant on $[0,\pi]$. Therefore the integral is positive for any positive function $f$. 

\subsection{Regularization by removing vacuum bubbles}
The regularization above is simple but has the disadvantage that it only keeps the $b=0$ component, thus reducing the trumpet to a cone. It turns out that there is a more elegant way to remove the divergence by excluding vacuum bubbles in each diagram. In~\eqref{eq:5.1} the contribution of vacuum bubbles is exactly the factors $\frac{1}{1-q^b}$, so after removing vacuum bubbles we have
\begin{align}\label{app:f4}
    \widetilde\tr(H^k)&=\sum_{b=0}\tilde \mu^b\int_0^\pi \frac{d\theta}{2\pi}E^k(\theta)2\cos(b\theta)\nonumber\\
    &=\int_0^\pi \frac{d\theta}{2\pi}E^k(\theta)\frac{1-\tilde\mu\cos\theta}{1-2\tilde\mu\cos\theta+\tilde\mu^2}
\end{align}
where we use $\widetilde\tr$ to denote this new renormalized trace. In the second line we have completed the $b$ summation. More generally, for an arbitrary function $f(H)$ we now have 
\begin{equation}\label{app:f5}
    \widetilde {\tr}[f(H)]=\int_0^\pi \frac{d\theta}{2\pi}f[E(\theta)]\frac{1-\tilde\mu\cos\theta}{1-2\tilde\mu\cos\theta+\tilde\mu^2}
\end{equation}

We yet have to verify the positivity of this renormalized trace. This is obvious as the integration measure $\frac{1-\tilde\mu\cos\theta}{1-2\tilde\mu\cos\theta+\tilde\mu^2}>0$ on $[0,\pi]$. So the trace is positive for arbitrary $f>0$. This regularization of the trumpet amplitudes justifies the calculation of double-cone amplitudes in~\cite{Okuyama:2023byh} where two trumpet amplitudes are glued together after stripping off the $\frac{\tilde \mu^b}{1-q^b}$ factors.

With the results above, we can now construct the trace of the thermal density matrix, $\tr(e^{-\beta H})$ or $\widetilde{\tr}(e^{-\beta H})$, these trumpets can then be sliced open as in Figure~\ref{cut a trumpet} to produce two-sided states which are highly entangled across a pair of EoW branes. We can treat these two-sided states as `thermal field double states' in the presence of EoW branes. The special feature here is that we have the `left' and `right' component for the bulk while the entanglement is not coming from a connected geometry, which provides an example of the quantum connectivity discussed in~\cite{Engelhardt:2023xer}.

\subsection{Trumpet in the presence of matter: an exponential divergence}
It is natural to generalize the discussion of trumpet amplitudes above to the presence of matter, and we assume that there is only one type of matter for simplicity. The naive trace for an arbitrary monomial $O$ in $H$ and $M$ is again defined by the sum of diagonal matrix elements
\begin{equation}
    \Tr(\hat O)=\sum_i O_{ii}
\end{equation}
here each $i$ denotes a basis vector labeled by $\ket{i_1i_2\ldots i_n}$, with $i_k=0,1$. This basis is not orthonormal but it is nevertheless complete and countable.\footnote{Since for each total chord number $n$ there is a finite number of chord states, so the total number of different states is countable as it is a countable union of finite sets.} Similar to the pure gravity case, $\Tr(O)$ can be represented as a sum over all diagrams as shown in Figure~\ref{diagonal element matter}, which can again be warped into a chord diagram on a defected disk as shown in Figure~\ref{diagonal element matter warpped}.  
\begin{figure}
	\centering
	\begin{subfigure}{0.4\textwidth}
		\includegraphics[width=\textwidth]{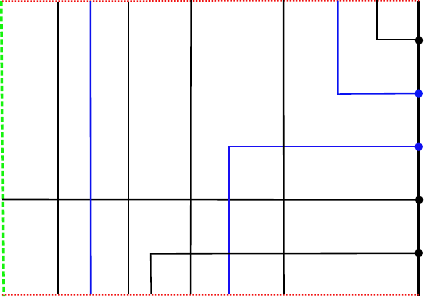}
        \caption{}
        \label{diagonal element matter} 
	\end{subfigure}
        \hspace{5em}
	\begin{subfigure}{0.3\textwidth}
		\includegraphics[width=\textwidth]{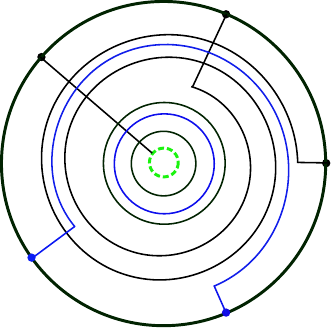}
        \caption{}
        \label{diagonal element matter warpped}
        \end{subfigure}
        \caption{Chord diagrams representing diagonal elements in the presence of matter. Here we used the diagonal matrix element for the state $\ket{0100010}$ as an example, where the matter chord is again represented by a blue line. Note that the chord configurations on the top and bottom slices are identical, thus we can identify them to obtain a defected disk.}
\end{figure}
In the presence of matter we can define boundary, vacuum and EoW chords, connected parts and vacuum bubbles in the same way as in the pure gravity case. Similar arguments as in the pure gravity case can still be applied to show that the sum of connected diagrams $\sum _cD_c(O,b)$ for a given boundary insertion $O$ is finite. Proliferating vacuum bubbles again leads to divergence, but in this case it becomes exponential. 

First consider the $b=0$ sector, each connected diagram with $b=0$ and $n_v$ vacuum chords is multiplied by a factor $2^{n_v}$ in the naive trace as each vacuum chord can be either a gravity or a matter chord. Therefore the leading divergent term is formally $\sum_c D_c(O,b=0)\times(1+2+2^2+2^3+\ldots)$, which diverges exponentially instead of linearly. This is similar to our expectation from JT gravity coupled to conformal matter~\cite{Penington:2023dql}, where for small $b$ the geometry contains a long tube connecting the EoW brane to the asymptotic boundary, and quantization of CFT on such a tube geometry leads to a negative ground energy which scales as $E_0\sim-\frac{\pi c}{6b}$, propagating this negative energy mode along the long tube then leads to an exponential divergence instead of a linear one in the pure gravity case. 

For sectors with $b\neq 0$, a vacuum bubble containing $n_h$ gravity chords and $n_m$ chords leads to a factor $q^{bn_h}r^{bn_b}$, and there are $\binom{n_h+n_b}{n_h}$ ways to arrange these chords, so summing over all possible vacuum bubbles gives the following multiplication factor for each connected diagram
\begin{equation}
    \sum_{n_h,n_b}\binom{n_h+n_b}{n_h}q^{bn_h}r^{bn_b}=\sum_n(q^b+r^b)^n
\end{equation}
which converges only if $q^b+r^b<1$, therefore in the presence of matter there could be subleading divergences coming from the $b\neq 0$ sectors, depending on crossing factors $q$ and $r$.

Regularization in the presence of multiple types of chords becomes much more complicated than the single chord case. We might ask whether simply removing all vacuum bubbles in the diagrams defines a positive definite trace. However, the proof of positivity in the pure gravity case exploits the fact that the boundary algebra is Abelian and is therefore isomorphic to $L^\infty([-\f{2}{\sqrt{1-q}},\f{2}{\sqrt{1-q}}])$. Therefore we can implement the spectral representation of the operators to prove positivity, which reduces to the positivity of the integration measure in~\eqref{app:f5}. In the presence of multiple chords we cannot `diagonalize' $H$ and $M$ operators simultaneously so the same argument no longer works. Nevertheless, it remains unclear whether or not this scheme preserves positivity of the trace in the multiple chord case. 

On the other hand, we also discuss whether the `brute force' regularization, that is simply removing all diagrams either with vacuum bubbles or EoW brane chords, still works in this case. This scheme amounts to formally divide by $\tr{(\mathbb I)}$ so it seems to remain valid. There is, however, a subtlety. In the pure gravity case this formal division can be also interpreted as taking the following limit
\begin{equation}\label{app:f6}
    \lim_{N\ra\infty}\f{\sum_{n=0}^N O_{nn}}{N}
\end{equation}
which is well defined for arbitrary $O$ which is a polynomial in $H$. In the multiple chord case, this limit should be replaced by
\begin{equation}\label{app:f7}
    \lim_{N\ra\infty}\f{\sum_{n_0+n_1\leq N} O_{ii}}{2^N}
\end{equation}
where the sum in the numerator is over all diagonal matrix elements in the subspace of no more than $N$ chords. However, the large $N$ limit of this sum does not converge to the `brute force' renormalized trace defined above, and different diagrams will come with factors depending on their detailed structures. In this case it is not clear whether cyclicity of the trace will be preserved. The root of this problem is that in the formal division, both numerator and denominator are divergent quantities which are, in fact, not well defined. In the pure gravity case this subtlety can be redeemed by interpreting the division as the limit~\eqref{app:f7}, but the same argument fails in the multi-chord case. One possible regulator in this case is provided in~\eqref{eq:5.2}, and further analysis of this regulator will be discussed in future works.

\section{Basic Concepts in Operator Algebra}\label{app:math}
In this appendix we present a short introduction to operator algebras. We will cover only the minimum necessary for the remaining parts of this paper. General discussions can be found, for example, in~\cite{takesaki2001theory,takesaki2003theory}. A short review for physicists can be found in \cite{Sorce:2023fdx}.

Let $\mH$ be a Hilbert space and $B(\mH)$ be the algebra of all bounded operators on $\mH$. Generally we are interested in various subalgebras $\mA\subseteq B(\mH)$. The center of $\mA$ is defined to be the set of all operators $a\in\mA$ which commutes with all elements in $\mA$. The commutant of $\mA$, denoted $\mA'$, is defined to be the set of all operators $a\in B(\mH)$ which commutes with all elements in $\mA$. It is easy to check that $\mA'$ is always a subalgebra of $B(\mH)$ by itself. 

Now we introduce the notion of weak convergence which is of key importance in physics. A sequence of operators $\{a_n\}\in B(\mH)$ is said to converge to $a\in B(\mH)$ in the weak operator topology if 
\begin{equation}
    \lim_{n\ra\infty}\mel{x}{a_n}{y}=\mel{x}{a}{y}\;\;\;\forall x,y\in\mH
\end{equation}
Since all physical observables are essentially matrix elements, we can consider the weak limit as the `physical' limit of operators. $\mA$ is said to be a von Neumann algebra if it is closed in weak topology, that is $\overline{\mA}^w=\mA$, where $\overline{A}^w$ denotes the weak closure. Therefore, we are always interested in von Neumann algebras in physics. There is a simple way to construct a von Neumann algebra from an arbitrary subalgebra $\mA\subseteq B(\mH)$ containing the identity operator and is closed under taking self-adjoints. The celebrated von Neumann bicommutant theorem states that $\overline{A}^w=\mA''$ in this case. That is, we can simply take the commutant twice to get the weak closure. Let $\mA_1$ and $\mA_2$ be two von Neumann algebras, we denote by $\mA_1\vee \mA_2$ the smallest von Neumann algebra containing $\mA_1$ and $\mA_2$, $\mA_1\wedge \mA_2$ the biggest von Neumann algebra contained in $\mA_1$ and $\mA_2$. In fact we have $\mA_1\vee \mA_2=(\mA_1\cup\mA_2)''$ and $\mA_1\wedge \mA_2=\mA_1\cap \mA_2$.

A von Neumann algebra $\mA$ is said to be a factor if its center contains only the identity operator. Factors are basic building blocks of general von Neumann algebras in the sense that general von Neumann algebras can be decomposed into direct sums or direct integrals of factors. For a factor $\mA$ we have $(\mA\vee \mA')'= (\mA\cup\mA')'=\mA'\cap\mA''=\mA'\cap\mA=\mathbb C$, and hence $\mA\vee \mA'=B(\mH)$. That is, a factor and its commutant generate all bounded operators. Another important feature of von Neumann algebras is that they contain enough projections such that all operators in $\mA$ can be generated by taking limits of linear combinations of projections in $\mA$. Specifically, if we take a self-adjoint operator $a\in \mA$, then all its spectral projections are still in $\mA$, we will use this fact in Section~\ref{sec:4.2}.

Von Neumann factors fall into three types according to the structure of their projections. Rather than developing this classification in detail, we focus only on its implications for the existence and properties of traces\footnote{Here by traces we mean faithful, semifinite and normal traces. Faithfulness means that the trace gives non-zero results for non-zero positive operators. Semifiniteness says that the trace is well-defined for `enough' operators in the algebra. While normality says that the trace has some notion of continuity. Details of these definitions can be found in~\cite{takesaki2003theory}}. For type I factors, we have well-defined traces which are ``quantized" on projections. That is, traces of projections take integer values after suitable normalizations. Type II factors also have well-defined traces, but they take continuous values on projections, with the range either $[0,1]$ or $[0,\infty]$ after suitable normalizations. The factor is called type II$_1$ in the former case and type II$_\infty$ in the latter case. Finally, type III factors do not admit well-defined traces. Physically, we usually identify subalgebras with subsystems. Their types can then be understood as characterizing entanglement structures between subsystems and their complements, with the entanglement most regular in the type I case, and most singular in the type III case.    

% Bibliography

%% [A] Recommended: using JHEP.bst file
%% \bibliographystyle{JHEP}
%% \bibliography{biblio.bib}

%% or
%% [B] Manual formatting (see below)
%% (i) We suggest to always provide author, title and journal data or doi:
%% in short all the informations that clearly identify a document.
%% (ii) please avoid comments such as "For a review'', "For some examples",
%% "and references therein" or move them in the text. In general, please leave only references in the bibliography and move all
%% accessory text in footnotes.
%% (iii) Also, please have only one work for each \bibitem.

\bibliographystyle{JHEP}
\bibliography{biblio.bib}
\end{document}